\documentclass[showpacs,amsmath,amssymb,twocolumn,prl,superscriptaddress,notitlepage]{revtex4-2}

\usepackage{qcircuit}
\usepackage[dvips]{graphicx}
\usepackage{amsmath,amssymb,amsthm,mathrsfs,amsfonts,dsfont}
\usepackage{subfigure, epsfig}
\usepackage{braket}
\usepackage{bm}
\usepackage{enumerate}
\usepackage{algorithm}
\usepackage{algpseudocode}
\usepackage{diagbox}
\usepackage{physics}
\usepackage{color}
\usepackage{geometry}
\usepackage{multirow}
\usepackage{comment}
\usepackage{tikz}

\usepackage[]{qcircuit}

\usetikzlibrary{arrows}
\usetikzlibrary{shapes,fadings,snakes}
\usetikzlibrary{decorations.pathmorphing,patterns}
\usetikzlibrary{calc}
\usetikzlibrary{positioning}
\usepackage[colorlinks = true]{hyperref}

\usepackage[capitalise,nameinlink]{cleveref}
\usepackage[normalem]{ulem}
\usepackage[font=small,labelfont=bf,justification=RaggedRight]{caption}
\usepackage[capitalise,nameinlink]{cleveref}

\usepackage{appendix}
\usepackage{titletoc}
\usepackage{lineno}

\graphicspath{{./figure/}}

\newtheorem{theorem}{Theorem}

\newtheorem{observation}{Observation}

\newtheorem{lemma}{Lemma}
\newtheorem{corollary}{Corollary}
\newtheorem{prop}{Proposition}

\newcommand{\mc}{\mathcal}
\newcommand{\mb}{\mathbf}

\newcommand{\id}{\mathbb{I}}

\newcommand{\comments}[1]{}
\DeclareMathOperator{\supp}{\mathrm{supp}}

\newcommand{\red}[1]{{\color{black} #1}}

\newcommand{\orange}[1]{{\color{orange}#1}}
\newcommand{\identity}[0]{{\mathbb{I}}}

\newcommand{\zyr}[1]{\textcolor{red}{(ZY: #1)}}

\def\>{\rangle}
\def\<{\langle}

\usepackage{hyperref} 
\hypersetup{colorlinks=true,linkcolor=blue}

\usepackage{orcidlink}

\begin{document}


\title{Auxiliary-Free Replica Shadows: Efficient Estimation of Multiple Nonlinear Quantum Properties}
\author{Qing Liu~\orcidlink{0000-0002-1576-1975}}
\email{These authors contributed equally to this work.}
\affiliation{Key Laboratory for Information Science of Electromagnetic Waves (Ministry of Education), Fudan University, Shanghai 200433, China}

\author{Zihao Li~\orcidlink{0000-0002-2766-7521}}
\email{These authors contributed equally to this work.}
\affiliation{State Key Laboratory of Surface Physics and Department of Physics, Fudan University, Shanghai, 200433, China}
\affiliation{Institute for Nanoelectronic Devices and Quantum Computing, Fudan University, Shanghai, 200433, China}
\affiliation{Center for Field Theory and Particle Physics, Fudan University, Shanghai, 200433, China}

\author{Xiao Yuan~\orcidlink{0000-0003-0205-6545}}
\email{xiaoyuan@pku.edu.cn}
\affiliation{Center on Frontiers of Computing Studies, Peking University, Beijing 100871, China}
\affiliation{School of Computer Science, Peking University, Beijing 100871, China}

\author{Huangjun Zhu~\orcidlink{0000-0001-7257-0764}}
\email{zhuhuangjun@fudan.edu.cn	}
\affiliation{State Key Laboratory of Surface Physics and Department of Physics, Fudan University, Shanghai, 200433, China}
\affiliation{Institute for Nanoelectronic Devices and Quantum Computing, Fudan University, Shanghai, 200433,
China}
\affiliation{Center for Field Theory and Particle Physics, Fudan University, Shanghai, 200433, China}

\author{You Zhou~\orcidlink{0000-0003-0886-077X}}
\email{you\_zhou@fudan.edu.cn}
\affiliation{Key Laboratory for Information Science of Electromagnetic Waves (Ministry of Education), Fudan University, Shanghai 200433, China}

\date{\today}

\begin{abstract}

Efficient estimation of nonlinear properties is a significant yet challenging task from quantum information processing to many-body physics.
Current methodologies often suffer from an exponential sampling cost or require auxiliary qubits and deep quantum circuits. To address these limitations, we propose an efficient auxiliary-free replica shadow (AFRS) framework, which leverages the power of the joint entangling operation on a few input replicas while integrating the mindset of shadow estimation. We rigorously prove that AFRS can offer exponential improvements in estimation accuracy compared with the conventional shadow method, and facilitate the simultaneous estimation of various nonlinear properties, unlike the destructive swap test. Additionally, we introduce an advanced local-AFRS variant tailored to estimating local observables with constant-depth quantum circuits, significantly simplifying the experimental implementation.
Our work paves the way for efficient and practical estimation of nonlinear properties on near-term quantum devices.
\end{abstract}
\maketitle


\let\oldaddcontentsline\addcontentsline
\renewcommand{\addcontentsline}[3]{}


\textit{Introduction---}The estimation of nonlinear properties $\tr(O\rho^t)$, such as purity, entanglement entropy \cite{Ekert2002Direct,Islam2015Measuring,johri2017entanglement}, and higher-order moments \cite{Todd2004Polynomial,Pichler2016Spectrum,elben2020mixedstate,singlezhou}, is crucial to applications in quantum foundations \cite{buhrman2001quantum,Horodecki2002Method}, quantum many-body physics \cite{Kaufmanen2016tanglement,Vermersch2019Scrambling,Elben2020topological,garcia2021quantum}, quantum error mitigation~\cite{endo2021hybrid,Cai2023QEM}, and algorithm design \cite{cotler2019cooling,Huggins2021Virtual,Koczor2021Exponential,Yamamoto2021metrology}.
Two primary methodologies have emerged for estimating nonlinear properties: randomized measurement approaches \cite{elben2023randomized}, particularly shadow estimation \cite{aaronson2019shadow,huang2020predicting}, and the generalized swap test \cite{Ekert2002Direct,Todd2004Polynomial}. 
Shadow estimation provides an unbiased estimator $\hat{\rho}$ by performing randomized projective measurements on a single-copy state $\rho$; 
however, this technique often entails heavy sampling and post-processing budgets \cite{elben2020mixedstate,singlezhou}, which scale exponentially with the system size for nonlinear estimation. Conversely, the generalized swap test employs entangling measurements among multiple copies of $\rho$, offering greater efficiency in sampling time and post-processing \cite{cotler2019cooling,zhou2024hybrid,huang2022Science}, but typically necessitates auxiliary qubits and circuit depth that scale with the qubit number $n$.  

Certain strategies may reduce the circuit depth, however, they usually require $\mc{O}(n)$ auxiliary qubits initialized in a {sophisticated} entangled state \cite{quek2024multivariate}, complicating practical implementation on near-term quantum hardware \cite{patel2016quantum,gao2019entanglement,zhang2019modular,gan2020hybrid,peng2024experimental}.
Meanwhile, while destructive measurements without auxiliary qubits \cite{cotler2019cooling,Huggins2021Virtual,subacsi2019entanglement,yirka2021qubit} have been realized in cold-atom \cite{Islam2015Measuring,Kaufmanen2016tanglement,cotler2019cooling,bluvstein2024logical} and superconducting-qubit systems \cite{huang2022Science,Brien2022purification}, they are believed to yield limited information compared with full swap tests \cite{Huggins2021Virtual,volkoff2022ancilla,zhou2024hybrid}. Other developments are still restricted to few-qubit observables which also require case-by-case tailored circuits \cite{Huggins2021Virtual}.
Consequently, efficiently realizing general measurements of complex nonlinear functions of quantum states remains an intriguing open challenge.

Here, we address the problem by proposing a general framework called auxiliary-free replica shadow (AFRS) estimation \cite{ancillafree}, as illustrated in \cref{fig:noancilla}.
AFRS inherits the mindset of shadow estimation and elaborately integrates it with fine-tuned entangling measurements and efficient classical post-processings. 
We demonstrate that AFRS achieves an exponential advantage in sampling cost over the original (single-copy) shadow method \cite{huang2020predicting,Seif2022Shadow} for nonlinear property estimation.
We further reduce the required circuit depth to a constant in key scenarios for estimating general local observables, such as many-body Hamiltonians, thereby enabling its practical implementation. Moreover, thanks to the shadow nature, the randomized measurement data can be multiplexed, making the sampling time scale logarithmically with the total number of observables, i.e., an exponential advantage compared with previous approaches \cite{Huggins2021Virtual,Koczor2021Exponential}.

\begin{figure}[t]
\begin{center}
\includegraphics[width=\linewidth]{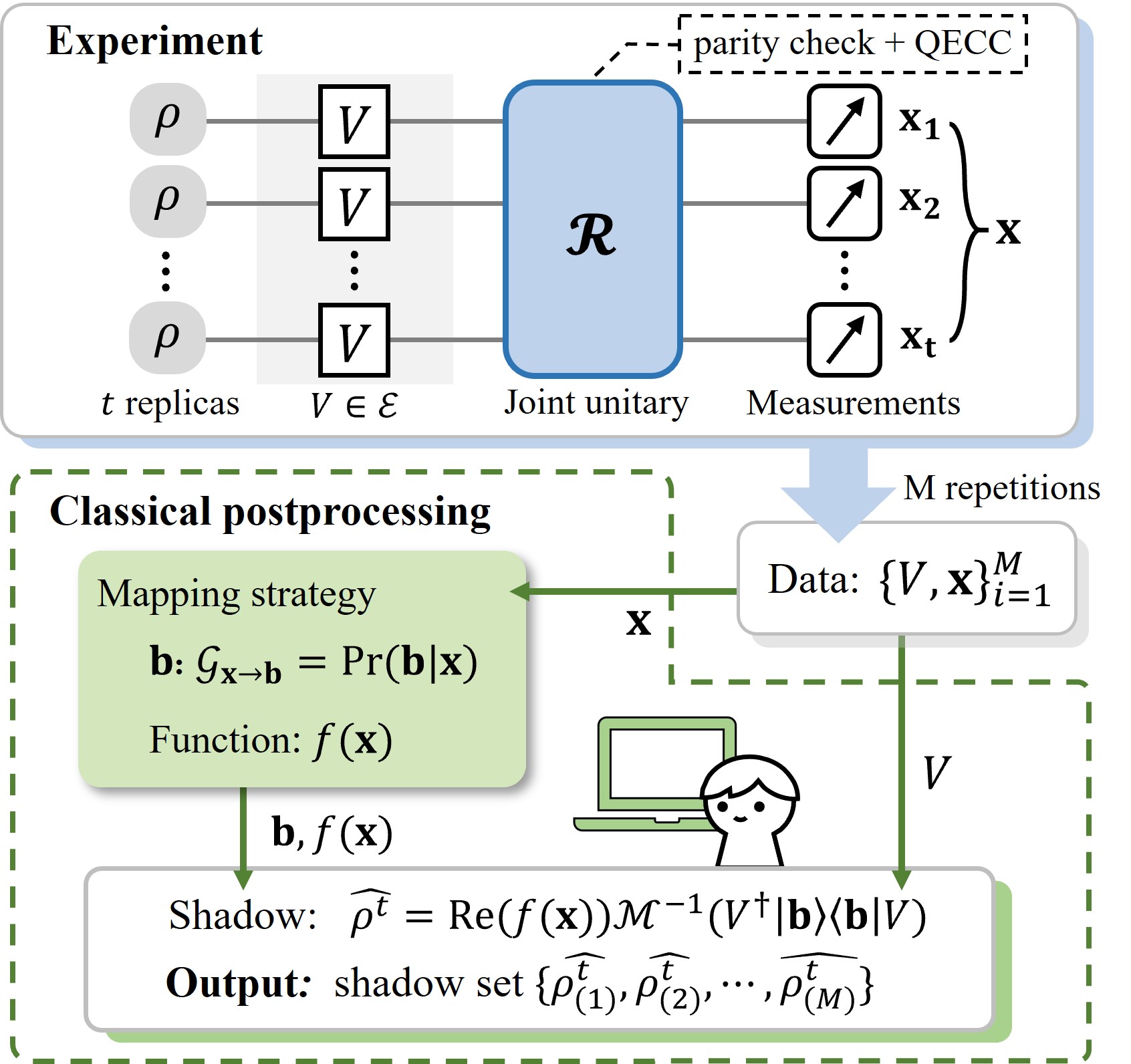}
\caption{\justifying{
An illustration of the AFRS framework. Experimental stage (up): $t$ replicas of $\rho$ are first evolved by the same random unitary $V$, then entangled via a fine-tuned joint operation $\mc{R}$, and finally measured in the computational basis to generate outcome $\mb{x}$. 
Post-processing stage (down): The mapping strategy $\mc{G}_{\mb{x}\to\mb{b}}$ transforms $\mb{x}$ to $\mb{b}$, which along with the unitary $V$ and function $f(\mb{x})$ forms an estimator $\widehat{\rho^t}$ of $\rho^t$. 
Repeating the whole procedure $M$ times generates the shadow set $\big\{\widehat{\rho^t_{(i)}}\big\}^M_{i=1}$ for estimating nonlinear properties.
}}
\label{fig:noancilla}
\end{center}
\end{figure}

\textit{Original shadow estimation---}We begin by briefly reviewing the original shadow (OS) estimation protocol \cite{huang2020predicting}.
For a quantum state $\rho$ on a $d$-dimensional Hilbert space $\mc{H}_d$, OS proceeds by first applying a random unitary $V$ on $\rho$, which is sampled from a suitable ensemble $\mc{E}$.
Typical choices for $\mc{E}$ include the $n$-qubit Cliﬀord group $\mc{E}_{\mathrm{Cl}}$ and the $n$-fold tensor product of the single-qubit Clifford group $\mc{E}_{\mathrm{LCl}}$ \cite{huang2020predicting,Hu2022Hamiltonian,hu2022Locally,ohliger2013efficient}. 
After applying the unitary $V$, the evolved state $V(\rho)\equiv V\rho V^\dagger$ is then measured in the computational basis $\left\{\ket{\mb{b}}\right\}_{\mb{b}\in[d]}$ with $[d]:=\{1,2,\ldots,d\}$, generating outcome $\mb{b}$ with probability $\Pr(\mb{b}|V)=\bra{\mb{b}}V(\rho)\ket{\mb{b}}$. 
From each measurement, a classical snapshot of the state is constructed through the estimator $\hat{\rho}=\mathcal{M}^{-1}( 
V^{\dag} \ketbra{\mb{b}}{\mb{b}} V)$, where $\mc{M}^{-1}$ is the inverse of the measurement channel, whose specific form depends on $\mc{E}$. 
This estimator equals the original state $\rho$ in expectation, i.e., $\mathbb{E}[\hat{\rho}]=\rho$.
Repeating such ``measurement + post-process" procedure $M$ times generates a shadow set $\{\hat{\rho}_{(i)}\}_{i=1}^M$, which enables efficient estimation of observable expectation value through $\tr(O\hat{\rho})$, particularly for local observables and fidelities. 
The estimation variance is upper bounded by the squared shadow norm, denoted by $\|O\|^2_{\mathrm{sh},\mc{E}}$. \red{Such norm is closely related to the Hilbert-Schmidt norm for ensemble $\mc{E}_{\mathrm{Cl}}$, while scales exponentially in the locality of target observables for ensemble $\mc{E}_{\mathrm{LCl}}$ \cite{huang2020predicting}.}
For a challenging nonlinear function like $\tr(O\rho^t)$, one can compute $\tr(O\hat{\rho}_1\hat{\rho}_2\cdots\hat{\rho}_t )$ using $t$ distinct snapshots from the shadow set. 
However, the variance of such a nonlinear estimator grows exponentially with the system size $n$ \cite{huang2020predicting,elben2020mixedstate}, making the sampling cost prohibitive even for medium-scale systems.

\textit{Framework of AFRS---}To address the above limitation of OS, we propose the framework of the AFRS protocol, which enables efficient estimation of nonlinear quantum properties through two coordinated stages. As illustrated in \cref{fig:noancilla}, in the experimental stage, we prepare $t$ identical replicas of the quantum state $\rho$ and apply the same random rotation $V\in\mc{E}$ to all replicas simultaneously. These states are then entangled through a specially designed joint operation $\mc{R}$, followed by computational measurements to produce an outcome string $\mb{x}=\{\mb{x}_1 \mb{x}_2\cdots \mb{x}_t\}\in\red{[d]^t}$.
In the classical post-processing stage, each measurement outcome $\mb{x}$ is converted to an effective result $\mb{b}\in[d]$ through a mapping strategy $\mc{G}_{\mb{x}\to\mb{b}}$, while a correlation function $f(\mb{x})$ is calculated in parallel. These elements allow us to construct an unbiased estimator $\widehat{\rho^t}$ of the state power $\rho^t$, with complete mathematical details provided later in \cref{ob:wholeCprob} and \cref{th:sh_est}.
Repeating this procedure $M$ times generates a shadow set $\big\{\widehat{\rho^t_{(i)}}\big\}^M_{i=1}$ for estimation of $\tr(O\rho^t)$ with significant higher sampling efficiency surpassing conventional shadow methods.
Notably, our framework can estimate various nonlinear functions beyond state powers, like state overlap and fidelity \cite{elben2020cross,Zhenhuan2022correlation}, purity and higher-order moments \cite{van2012Measuring,Brydges2019Probing,ketterer2019characterizing,Yu2021Optimal,liu2022detecting}, \red{R\'enyi entropies \cite{acharya2016estimating,wang2023quantum,Quek2024multivariatetrace}}, quantum Fisher information \cite{rath2021Fisher,Yu2021Fisher,zhang2025krylov}, out-of-time-ordered correlators \cite{Vermersch2019Scrambling,garcia2021quantum}, topological invariants \cite{Elben2020topological,Cian2020Chern}, etc. In particular, the $t$ input replicas can be chosen differently \cite{zhou2024hybrid} for specific applications. For simplicity, we primarily demonstrate the estimation of monomial function $\tr(O\rho^t)$ here, with complete derivations and generalizations provided in Supplemental Material (SM) \cite{supplementary}.

The AFRS protocol centers on generating effective measurement statistics that behave as if measuring $\rho^t$ directly, despite $\rho^t$ not being a physical quantum state. This is achieved through careful design of the entangling operation $\mc{R}$, which correlates information across $t$ replicas of $\rho$. The measurement statistics are generated according to the corresponding \emph{pseudo} `conditional probability'
\begin{equation}\label{eq:prob}
\Pr(\mb{b}|V)=\bra{\mb{b}}V(\rho^t)\ket{\mb{b}} =\tr\left[S_t Q_{\mb{b}} V(\rho)^{\otimes t}\right],
\end{equation}
where $S_t$ is the clockwise shift operation that cyclically permutes the $t$ replicas, acting as $S_t\ket{\mb{x}_1 \mb{x}_2\cdots \mb{x}_t}=\ket{\mb{x}_2 \mb{x}_3\cdots \mb{x}_1}$; and $Q_\mathbf{b} = t^{-1} \sum_{i=1}^t  \ketbra{\mb{b}}{\mb{b}}_i \otimes \id_d^{\otimes (t-1)}$ is a symmetric observable that is diagonal in the computational basis.
It is not hard to check that $S_t$ and $Q_\mathbf{b}$ commute and thus share a common eigenbasis $\{\ket{\psi_\mb{x}}\}$, leading to a crucial decomposition:
\begin{observation}\label{ob:wholeCprob}
    The pseudo `conditional probability' in \cref{eq:prob} can be written as
    \begin{equation}
        \Pr(\mb{b}|V) = \sum_\red{\mb{x}\in[d]^t}f(\mb{x})\Pr(\mb{b}|\mb{x}) \Pr(\mb{x}|V),
    \end{equation}
    where $f(\mb{x}):=\bra{\psi_\mb{x}}S_t\ket{\psi_\mb{x}}$ is a function depending on the measurement result $\mb{x}$, while $\Pr(\mb{b}|\mb{x}):=\bra{\psi_\mb{x}}Q_\mb{b}\ket{\psi_\mb{x}}$ and $\Pr(\mb{x}|V):=\bra{\psi_\mb{x}}V(\rho)^{\otimes t}\ket{\psi_\mb{x}}$ represent classical and quantum \emph{legal} conditional probabilities, respectively.
\end{observation}

It is straightforward to verify that $\Pr(\mb{b}|\mb{x})\geq 0$ and $\sum_{\mb{b}}\Pr(\mb{b}|\mb{x})=1$. We refer to this conditional probability as the \emph{mapping strategy}, denoted as $\mathcal{G}_{\mb{x}\to \mb{b}}:=\Pr(\mb{b}|\mb{x})$, since it is generated by classical random assignments.
Meanwhile, the second probability $\Pr(\mb{x}|V)$ can be generated in a quantum manner, i.e., by measuring $V(\rho)^{\otimes t}$ in the basis $\{\ket{\psi_\mb{x}}\}$. 
This measurement can thus be physically implemented by applying a suitable unitary transformation $\mc{R}$ on $\mc{H}_d^{\otimes t}$, followed by the computational basis measurement, where $\mc{R}$ satisfies $\mc{R} \ket{\psi_{\mb{x}}} = \ket{\mb{x}}$ for all $\mb{x}$. Consequently, the overall quantum circuit shown in \cref{fig:noancilla} is obtained.
Further elaboration and explicit details are provided in END MATTER.

The AFRS protocol provides unbiased estimation of $\rho^t$ with favorable estimation variance scaling.

\begin{figure}[t]
\centering
\includegraphics[width=\linewidth]{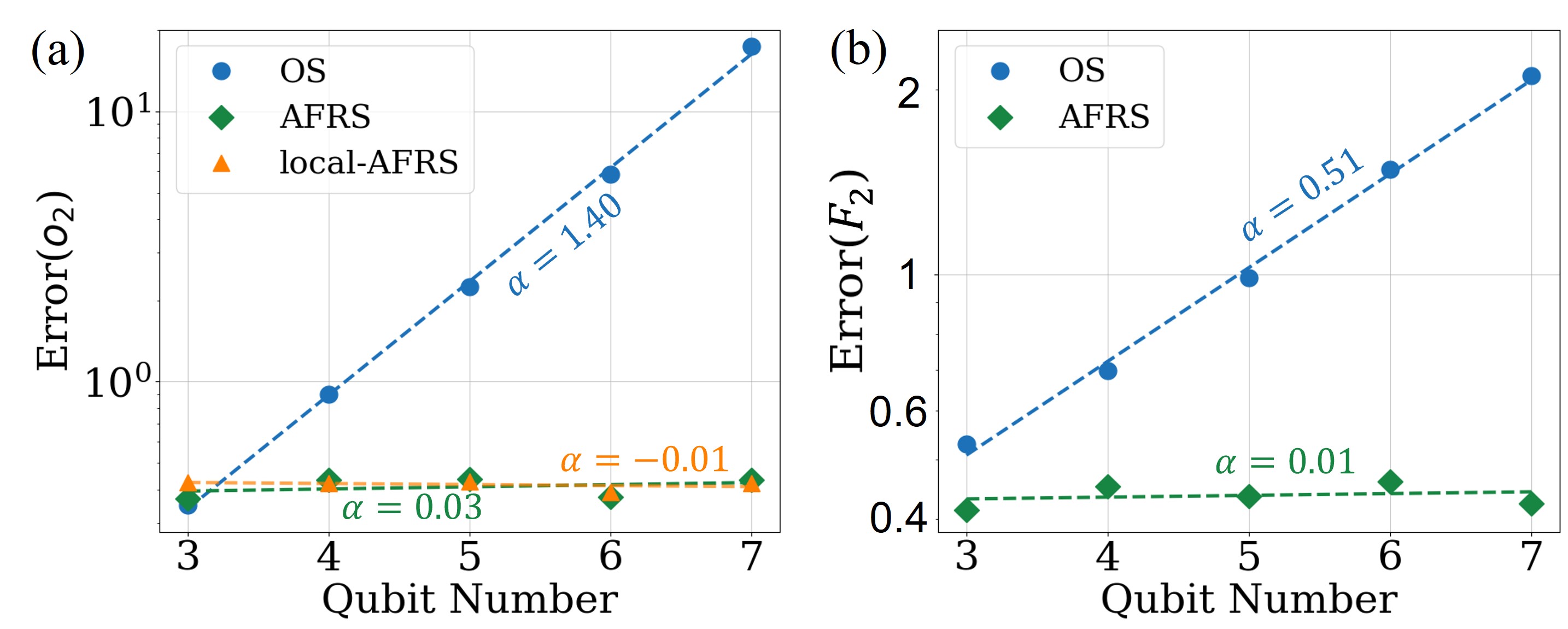}
\caption{\justifying{Scaling of estimation error with respect to the qubit number for OS, AFRS and local-AFRS protocols. The processed state is a noisy $n$-qubit GHZ state, $\rho = 0.7\ket{\text{GHZ}}\bra{\text{GHZ}}+0.3\identity_d/d$. 
In (a), 
$V\in \mc{E}_{\mathrm{LCl}}$ is used to estimate $o_2 =\tr(O\rho^2)$, with $O=Z_1Z_2$ being a local Pauli observable.
In (b), 
$V\in\mc{E}_{\mathrm{Cl}}$ is used to estimate the fidelity $F_2 =\bra{\text{GHZ}}\rho^2\ket{\text{GHZ}}$ with $O = \ket{\text{GHZ}}\bra{\text{GHZ}}$. 
The estimation errors for different protocols are compared under the same sample number $N$ of $\rho$, 
typically with  $N\ll d$ to show the asymptotic scaling. 
The parameter $\alpha$ denotes the fitting value of the slope with $\mathrm{Error}\sim \mathcal{O}(d^\alpha)$.
}} 
\label{fig:num_pauli_cliff}
\end{figure}

\begin{theorem}\label{th:sh_est}
    For a single implementation of the circuit in \cref{fig:noancilla} (i.e., single-shot, $M=1$), the estimator
    \begin{equation}\label{eq:shadow2}
        \widehat{\rho^t}:= \mathrm{Re}(f(\mb{x}))\cdot\mathcal{M}^{-1}\left( V^{\dag} \ketbra{\mb{b}}{\mb{b}} V \right)
    \end{equation}
    reproduces $\rho^t$ in expectation, i.e., $\mathbb{E}_{\{V,\mb{x},\mb{b}\}}\widehat{\rho^t} = \rho^t$. Here, \red{$\mb{x}\in [d]^t$} is the actual measurement outcome of the quantum circuit; function $f(\mb{x})$ and mapping result $\mb{b}$ derive from \cref{ob:wholeCprob}; 
    \red{the inverse channel $\mathcal{M}^{-1}$ depends only on the system dimension and the chosen unitary ensemble $\mathcal{E}$.}

The nonlinear property $\tr(O\rho^t)$ can be  estimated by $\widehat{o_t}={\rm tr}(O\widehat{\rho^t})$,
whose variance satisfies
    \begin{equation}\label{eq:varm}
        \mathrm{Var}(\widehat{o_t}) \leq \left\|O_0\right\|_{\mathrm{sh},\mathcal{E}}^2+\|O\|_{\infty}^2, 
    \end{equation}
    where $O_0=O-\operatorname{tr}(O) \mathbb{I}_d/d$ is the traceless part of $O$, and $\|\cdot\|_{\mathrm{sh},\mathcal{E}}$ is the shadow norm depending on $\mathcal{E}$ \red{\cite{huang2020predicting}}.  
\end{theorem}

\cref{th:sh_est} indicates that AFRS achieves exponential improvements in sampling efficiency for nonlinear estimation compared with OS. 
In particular, the variance of estimating $\tr(O\rho^t)$ with AFRS is almost the same as the variance of estimating linear property $\tr(O\rho)$ with OS, both with the squared shadow norm $\| O_0 \|^2_{\mathrm{sh},\mc{E}}$ serving as the dominant term. 
Here, we ignore the second term $\|O\|_\infty^2$ in \cref{eq:varm}, which is usually a constant. 
By comparison, when estimating $\tr(O\rho^t)$ using OS, one has to compute $\text{tr}\big(O_t\bigotimes_{j=1}^t \hat{\rho}_j\big)$
with $O_t:=(O\otimes \identity_d^{t-1})S_t$ being a global and high-rank operator; as a result, the estimation variance bounded by $\|O_t\|_{\mathrm{sh}, \mc{E}}^2$ scales exponentially with the qubit number $n$.
\red{This advantage is demonstrated numerically for two key regimes $\mc{E}_{\mathrm{LCl}}$ and $\mc{E}_{\mathrm{Cl}}$ in \cref{fig:num_pauli_cliff}, with further discussion and applications, including purity estimation, Fisher information bounds and R\'{e}nyi entropies, presented in the SM \cite{supplementary}.}

The practicality of AFRS protocol is enhanced by three key features: (\romannumeral 1) Following the mindset of shadow estimation, measurement data can be reused to estimate a collection of observables $\{O_i\}$ with the sampling time only scales logarithmically with the number of observables, such advantage is analyzed later; (\romannumeral 2) In the implementation of AFRS, no auxiliary qubits are needed; and (\romannumeral 3) The entangling operation $\mc{R}$ requires at most $n+1$ layers of elementary quantum (Clifford) gates for the important case of $t=2$. We systematically compile $\mc{R}$ (with the subsequent projective measurement) in END MATTER using the techniques of quantum error correction code (QECC) \cite{gottesman2002introduction} and mid-circuit measurement \cite{Gambetta2021Dynamic,DeCross2023Reuse}, which may be of independent interest. 
Note that the global nature of the joint unitary may present implementation challenges for some architectures with limited connectivity. In the following, for the sake of practical realizations on near-term hardware, we develop an advanced shallow-circuit variant that further reduces these requirements while preserving the protocol's advantages.

\textit{Local-AFRS---}Consider the observable $O$ of interest being local and acting nontrivially on the subsystem $A$, say, $O=O_{A}\otimes \id_{\bar{A}}$. 
We propose the local-AFRS protocol for such local observables, which significantly reduces the difficulty of implementing the entangling unitary $\mc{R}$ while maintaining the same performance as AFRS.
In particular, we show that one can estimate $o_t=\tr(O\rho^t)$, or more generally, the reduced operator of $\rho^t$ on $A$, say $\rho^t_A:=\tr_{\bar{A}}(\rho^t)$ with a greatly simplified quantum circuit, even to a constant depth when the qubit number in $A$ is $\mc{O}(1)$.

Our quantum circuit for the experimental stage of local-AFRS is illustrated in \cref{fig:local_afrs_sketch} (a), which is a variant of the AFRS circuit in \cref{fig:noancilla}. 
With respect to the locality property, we set the random unitary in the product form $V=V_A\otimes V_{\bar{A}}$ across the subsystems $A$ and $\bar{A}$ for all $t$ replicas, where $V_A$ is chosen from a unitary ensemble on $A$ denoted as $\mathcal{E}^{(A)}$. 
In addition, the entangling operation before the measurement can also take the product form $\mc{R}_A\otimes \mc{R}_{\bar{A}}$. 
Finally, we perform the computational basis measurement on the $t$ replicas, and obtain outcome bit-strings $\mb{x}^{A}$ and $\mb{x}^{\bar{A}}$ for 
subsystems $A$ and $\bar{A}$, respectively. 
In the classical post-processing stage of local-AFRS, we can use the unitary $V_A$ and these bit-strings to construct an estimator $\widehat{\rho^t_A}$, which equals $\tr_{\bar{A}}(\rho^t)$ in expectation. 
See \red{\cref{eq:local_AFRS_estimator} in END MATTER} for more 
details of this construction. 
The performance of our local-AFRS protocol is summarized in the following theorem.

\begin{figure}[t]
\centering
\includegraphics[width=\linewidth]{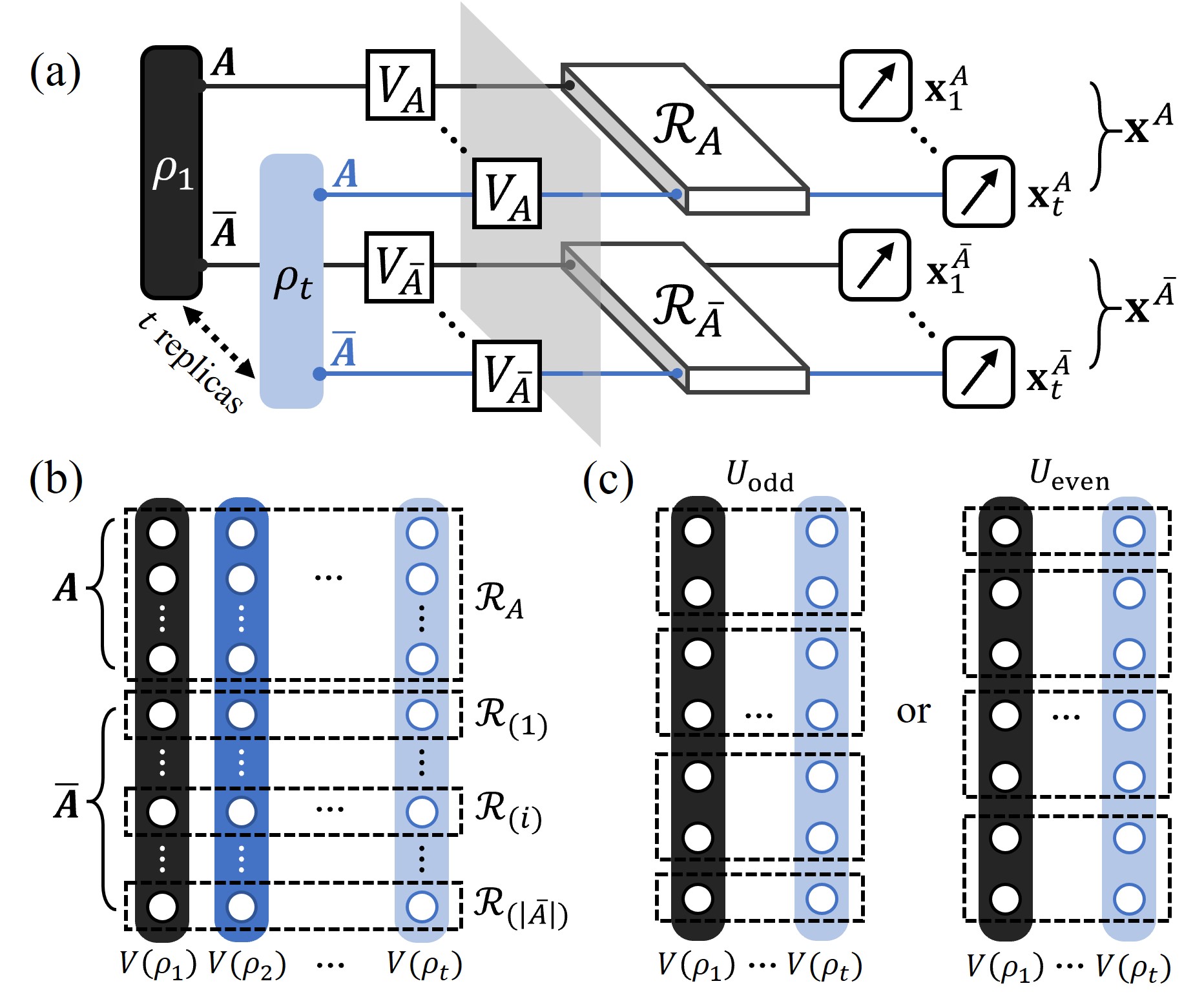}
\caption{\justifying{
\red{(a) Quantum circuit of local-AFRS. 
(b) Left-view of panel (a) along the gray cross-section. Each circle denotes a qubit from one replica of $V(\rho)$. The dashed boxed indicate subsequent joint unitaries. For a local observable $O=O_A\otimes \mathbb{I}_{\bar{A}}$, the block $\mc{R}_{\bar{A}}$ can be substituted by the qubit-wise form $\bigotimes_{i\in \bar{A}} \mc{R}_{(i)}$. (c) Example of entangling unitary for odd $n$: $U_{\text{odd}} = \bigotimes_{j=1}^{\lfloor n/2 \rfloor}\mc{R}_{2j-1,2j} \otimes \mc{R}_n$ and $U_{\text{even}} = \mc{R}_1 \otimes\bigotimes_{j=1}^{{\lfloor n/2 \rfloor}}\mc{R}_{2j,2j+1}$.}
}}
\label{fig:local_afrs_sketch}
\end{figure}

\begin{theorem}\label{th:localAFRSinformal}
The observable $O=O_A\otimes \id_{\bar{A}}$ on $\rho^t$ can be estimated by $\widehat{o_t}=\tr\!\big(\widehat{\rho^t_A}\,O_A\big)$,
whose variance satisfies
\begin{equation}\label{eq:varmLocal}
\begin{aligned}
\mathrm{Var}(\widehat{o_t}) \leq \left\|O_{A,0}\right\|_{\mathrm{sh},\mathcal{E}^{(A)}}^2 + \|O\|_{\infty} ^2,
\end{aligned}
\end{equation}
where 
$O_{A,0}$ is the traceless part of $O_A$, and 
$\|\cdot\|_{\mathrm{sh},\mathcal{E}^{(A)}}$ is the shadow norm depending on the ensemble $\mathcal{E}^{(A)}$. 
\end{theorem}

\cref{th:sh_est} can be seen as a special case of \cref{th:localAFRSinformal}, with the subsystem being the whole system $A=[n]$.
If focusing on the estimation of $O_A$, one can set the random
unitary on $\bar{A}$ to identity, say, $V=V_A\otimes \id_{\bar{A}}$. In addition, the entangling operation 
$\mc{R}_{\bar{A}}$ on $\bar{A}$ can be replaced by a tensor product of 
qubit-wise operations $\mc{R}_{(i)}$, 
as shown in \cref{fig:local_afrs_sketch} (b).
Our conclusion on the variance performance do not change under these modifications. 
In the usual case where $|A|$ is a constant, $\mc{R}_{A}$ and 
$\mc{R}_{(i)}$ can all be realized by constant-depth quantum circuits.
Consequently, local-AFRS significantly reduces the circuit depth of AFRS from $\mc{O}(n)$ to $\mc{O}(1)$, but keeps the estimation efficiency almost the same, as indicated by \cref{eq:varmLocal}.
In particular, for the important case of $t=2$, the circuit depth of local-AFRS is only $\mc{O}(|A|)$ (mainly used to realize $\mc{R}_A$). 
We demonstrate the estimation performance of local-AFRS in \cref{fig:num_pauli_cliff}~(a), which again reveals its exponential advantage over OS.

\textit{Estimating a collection of observables---}As an application of local-AFRS, we consider the important problem of estimating a collection of local observables $\mc{W}=\{O_l\}_{l=1}^L$, which may have different local supports.
Divide the collection as $\mc{W}=\cup_{k=1}^K \mc{W}_k$, where each subset $\mc{W}_k$ satisfies 
\begin{equation}\label{eq:subsetO}
\text{$\supp(O_l) \cap \supp(O_{l'})=\emptyset$}\   \text{or\, $\supp(O_l)\subseteq \supp(O_{l'})$},
\end{equation}
$\forall O_l,O_{l'}\in\mc{W}_k$.
We can use local-AFRS to estimate the observables in each 
subset $\mc{W}_k$ in the following way:
first apply a random unitary $V\in \mc{E}_{\mathrm{LCl}}$ for each replica, and then perform the entangling unitary $U_k=\bigotimes_j \mc{R}_{A_j}$ before the final projective measurement, 
such that each $O_l\in \mc{W}_k$ is supported in one of the subsystems $\{A_j\}_j$.
Note that one needs to apply $K$ different $U_k$ in the experimental stage one at a time. In the case of $t=2$, the depth of each circuit is at most $\max_l \supp(O_l)+2$. 
Consider the terms in the Hamiltonian of a transverse-field Ising model for example, 
one has $K=2$ with 
$\mc{W}_{\text{odd/even}}=\{Z_jZ_{j+1},X_j\}_{\text{odd/even}\ j}$, thus the corresponding $U_k$ are illustrated in \cref{fig:local_afrs_sketch} (c).

The performance of our estimation is shown as follows, which inherits the multiplex nature of shadow estimation and maintains the sampling advantage of AFRS.

\begin{corollary}\label{co:localAFRS}
Consider a collection of local observables $\mc{W}=\{O_l\}_{l=1}^L=\cup_{k=1}^K \mc{W}_k$ that satisfies \cref{eq:subsetO}.
In order to estimate $o_t^{(l)}=\tr(O_l\rho^t)$ within error $\epsilon$ for all $O_l\in \mc{W}$ with success probability at least $1-\delta$, 
it suffices to run the local-AFRS circuit the following number of times, 
\begin{equation}\label{eq:collM}
\begin{aligned}
M
=\frac{68 K}{\epsilon^2}\max_{O_l\in \mc{W}}\mathrm{Var}\Big(\widehat{o^{(l)}_t}\Big)\log\left(\frac{2L}{\delta}\right),
\end{aligned}
\end{equation}
where $\widehat{o^{(l)}_t}$ is the single-shot estimator for $o_t^{(l)}$. 
\end{corollary}

In particular, when all $O_l$ have constant supports, the variance 
of each single-shot estimator can be bounded by some constant according to \cref{eq:varmLocal}. 
In this case, 
the number of required runs is only $M=\mc{O}(K\log(L))$, which scales linearly in $K$ due to the implementation of $K$ distinct entangling unitaries $U_k$, and scales logarithmically in $L$ due to the shadow nature of local-AFRS. 
\red{Overall, the sample complexity scales linearly with the number of replicas $t$, and the prefactor depends only on the structure of the observables, such as their locality or Frobenius norm.}

Here are some remarks on various advantages of local-AFRS compared with previous methods.
First, when $K=\mc{O} (1)$ and $L=\mc{O} (n)$, the overall sampling time $M$ scales as $\mc{O} (\log(n))$ for local-AFRS, whereas previous methods require $M=\mc{O}(n)$ \cite{Huggins2021Virtual,Koczor2021Exponential}. 
Second, compared with Ref.~\cite{Koczor2021Exponential}, local-AFRS does not require auxiliary qubits. Moreover, the generalized controlled-swap operation in Ref.~\cite{Koczor2021Exponential} requires a circuit depth of $\mc{O}(n)$, while local-AFRS here maintains a constant-depth circuit, making it more feasible for current and near-future platforms.
Third, compared with Ref.~\cite{hakoshima2023localized}, our circuit-depth reduction in local-AFRS does not rely on any assumption about the unknown state $\rho$. 
Fourth, the auxiliary-free protocol based on Bell-like measurements proposed in Refs.~\cite{Huggins2021Virtual,Brien2022purification} primarily applies to few-qubit observables. Extending this approach to general $k$-local observables, as considered in our setting, remains both a challenging task and an open problem.
Last, compared with OS \cite{huang2020predicting,Seif2022Shadow}, our protocol harnesses the power of joint unitary and shows an exponential advantage on sampling time.
We clarify such advantages of local-AFRS more in-depth with the quantum error mitigation and virtual distillation task \cite{Huggins2021Virtual,Koczor2021Exponential,Brien2022purification} in END MATTER.

\textit{Discussions---}In this work, we proposed the AFRS framework for efficiently and practically estimating nonlinear properties of quantum states. It would serve as a significant component in various quantum information processing tasks, and inspire further developments in quantum algorithm design. 
Future work could extend AFRS to fermion and boson systems in quantum simulation 
\cite{gandhari2024precision,becker2024classical,Zhao2021Fermionic}, as well as to quantum channels and dynamics \cite{helsen2023shadow,Kunjummen2021process,liu2024virtual}. AFRS could also be further integrated with other shadow schemes, such as conjugate shadow \cite{king2024exponential}, multi-shot shadow \cite{helsen2023thrifty,zhou2023performance} (see some preliminary results in 
SM~\cite{supplementary}), shallow shadow \cite{hu2023classical,bertoni2024shallow,arienzo2023closed,Ippoliti2023Depth} and derandomization \cite{huang2021efficient,hadfield2022measurements,wu2023overlapped,gresch2023guaranteed}, to enhance its ability and performance in different scenarios.  Notably, AFRS could play a significant role in principal component analysis \cite{lloyd2014quantum,grier2024principal}. 
Finally, along with other mitigation techniques \cite{Huo2022Dual,cai2021quantum}, we expect that AFRS would enhance quantum error mitigation and improve the robustness of the shadow measurement scheme itself \cite{Chen2021Robust,koh2022classical,brieger2023stability,wu2024error,Jnane2024QEMSh}.


\textit{Acknowledgements---}Q.L.$\&$Y.Z. acknowledge 
from the Innovation Program for Quantum Science and Technology Grant Nos.~2024ZD0301900 and 2021ZD0302000, the National Natural Science Foundation of China (NSFC) Grant No.~12205048 and 12575012, the Shanghai Science and Technology Innovation Action Plan Grant No.~24LZ1400200, Shanghai Pilot Program for Basic Research - Fudan University 21TQ1400100 (25TQ003), and the start-up funding of Fudan University.
Z.L.$\&$H.Z. acknowledge the support from
Shanghai Science and Technology Innovation Action Plan (Grant No.~24LZ1400200), Shanghai Municipal Science and Technology Major Project (Grant No.~2019SHZDZX01), National Key Research and Development Program of China (Grant No.~2022YFA1404204), and NSFC (Grant No.~92165109).
X.Y.~acknowledges the support from NSFC (Grant No.~12175003 and No.~12361161602), NSAF (Grant No.~U2330201), the Innovation Program for Quantum Science and Technology (Grant No.~2023ZD0300200). 

\bibliography{BibAFRS}

\begin{thebibliography}{99}
\bibitem [{anc()}]{ancillafree}%
  \BibitemOpen
  \href@noop {} {}\bibinfo {note} {`auxiliary-free' here means that no
  additional ancilla regiser is required beyond the $t$ copies of the state of
  interest, in contrast to swap-test-type protocols where extra ancillas
  mediate the measurement.}\BibitemShut {Stop}%
\end{thebibliography}
\let\addcontentsline\oldaddcontentsline


\section{END MATTER}
\textit{The matrix form of $\mc{R}$---}Here, we give the matrix form of the unitary $\mc{R}$ expressed in the computational basis $\{|\mathbf{x}\>\}$. 
We first introduce some notation. 
Two strings $\mathbf{x},\mathbf{y} \in \red{[d^t]}$ are said to be in the same class if there is an integer $r$ such that $\tau^r(\mathbf{x})=\mathbf{y}$. That is, they can be transformed to each other by the one-step cyclic permutation $\tau$, which corresponds to the unitary $S_t$ on $t$ copies, i.e., $S_t\ket{\mathbf{x}}=\ket{\tau(\mathbf{x})}$. 
For instance, the strings $01353$ and $35301$ are in the same class in $\red{[6]^5}$. 
Note that two distinct classes have no intersection, 
and $\red{[d^t]}$ is the union of all distinct classes. 
For each class of $\red{[d^t]}$, we choose its smallest element $\mathbf{z}$ as the representative, and use $[\mathbf{z}]$ to represent the class. 
For example, $\red{[2]^4}$ can be divided into six distinct classes, and one of them is $[0011] = \{0011,0110,1100,1001\}$. 
For any class  $[\mb{z}]$ of $\red{[d]^t}$, we denote by $|[\mathbf{z}]|$ its cardinality, and define its corresponding normalized $t$-qudit states by Fourier transformation as
\begin{align}
\ket{\Psi_{[\mathbf{z}]}^{(k)}} = \frac{1}{\sqrt{|[\mathbf{z}]|}} \sum_{r=0}^{|[\mb{z}]|-1} \exp\left(\frac{rk\cdot 2\pi \rm{i}}{|[\mb{z}]|}\right) \ket{\tau^r(\mb{z})}
\end{align}
for $k=0,\dots,|[\mb{z}]|-1$. 
The following lemma shows that the set of such states forms the common orthonormal eigenbasis $\{\ket{\psi_\mb{x}}\}$ of the operators $S_t$ and $Q_\mathbf{b}$. 
\begin{lemma}\label{lem:matrixR}
$\left\{\ket{\Psi_{[\mb{z}]}^{(k)}}\right\}_{[\mb{z}],k}$
forms an orthonormal basis on $\mathcal{H}_d^{\otimes t}$. 
In addition, $\ket{\Psi_{[\mb{z}]}^{(k)}}$ is the common eigenstate of $S_t$ and $Q_{\mb{b}}$, with corresponding eigenvalue
$\exp\big(\!-\frac{k\cdot 2\pi \rm{i}}{|[\mb{z}]|}\big)$ 
and $t^{-1}\times($the number of $\mb{b}$ in the string $\mb{z})$, respectively. 
\end{lemma}

Based on the basis in \cref{lem:matrixR}, it is direct to write down the corresponding entangling unitary as 
\begin{equation}\label{eq:Rany}
\mc{R}=\sum_{[\mb{z}]}\sum_{k=0,\dots,|[\mb{z}]|-1}\ket{\tau^k(\mb{z})}\bra{\Psi_{[\mathbf{z}]}^{(k)}}.   
\end{equation}

According to the eigenvalue of $S_t$ shown in \cref{lem:matrixR}, the real part of the function $f(\mb{x})$ in \cref{th:sh_est} reads $\mathrm{Re}(f(\mb{x}))=\cos(\frac{k\cdot 2\pi }{|[\mb{z}]|})$,
where $[\mb{z}]$ is the class to which $\mb{x}$ belongs, and the integer $k$ satisfies $\tau^k(\mb{z})=\mb{x}$. 

It is also worth noting, from the eigenvalue of $Q_{\mb{b}}$ in \cref{lem:matrixR}, that the mapping strategy $\mathcal{G}_{\mb{x}\to \mb{b}}$ is just to randomly choose $\mb{b}$ from $\{\mb{x}_1,\mb{x}_2, \ldots ,\mb{x}_t\}$ of $\mb{x}$ with uniform probability $1/t$. 
In practice, instead of such a random assignment, one can alternatively choose $\mb{b}$ as the average of all $\mb{x}_i$. 
Such kind of treatment directly simulates the random assignment process and inherits the multi-shot nature \cite{helsen2023thrifty,zhou2023performance}, and would further reduce the variance in \cref{eq:varm} by a factor of $t$ in some cases. We give the form of the new estimator and the detailed variance analysis in 
SM~\cite{supplementary}.


\textit{Circuit compilation of $\mathcal{R}$---}The circuit compilation of $\mc{R}$ with the subsequent projective measurement is significant for the realization in quantum computing platforms.
Here, we focus on the $t=2$ scenario, which is the most promising to demonstrate on near-term devices. In this case, $\mc{R}$ acting on $\mc{H}_d^{\otimes 2}$ is Hermitian and has a matrix form in 
\cref{eq:Rany}.
We start from the single-qubit case with $d=2$, then extend to the general qudit case, and finally give the construction of the most challenging $n$-qubit case with $\mc{H}_d=\mc{H}_2^{\otimes n}$, which may be of independent interest. 

\begin{figure}[t]
\centering
\includegraphics[width=\linewidth]{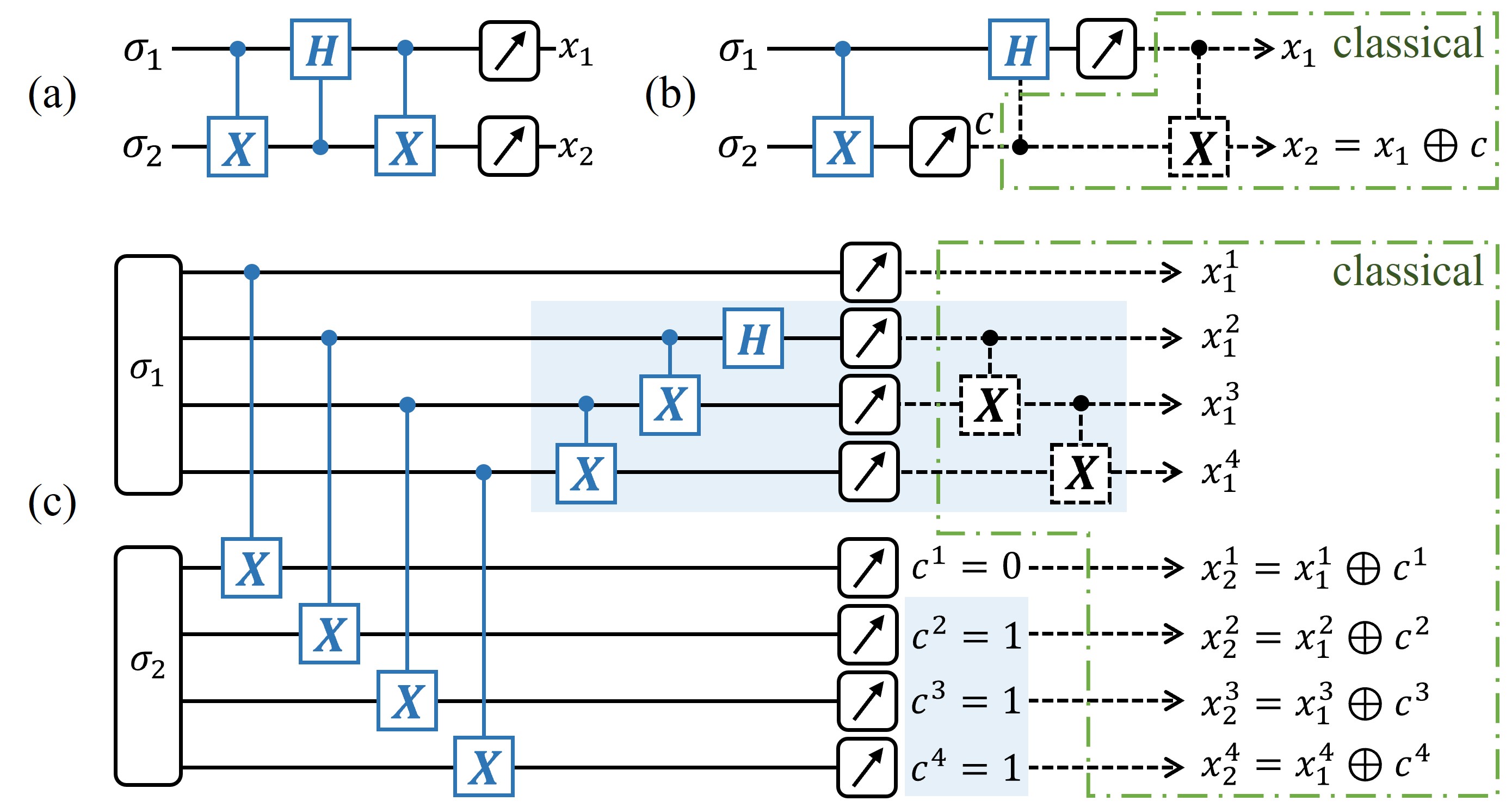}
\caption{\justifying{Circuit compilation of $\mathcal{R}$. (a) Compilation of $\mathcal{R}$ in the single-qubit case. (b) Simplification of the circuit by utilizing classical post-processing. (c) The generalization to the $n$-qubit case. Here, we take $n=4$, and assume that the parity check $c^1=0$, and $c^i=1$ for $i\in\{2,3,4\}$. For the first qubit pair, the quantum circuit is similar to the single-qubit case, but there is a joint unitary consisting of CNOTs among the last three qubits in the first replica. In addition, as in the single-qubit case shown in (b), there are some classical post-processings to calculate the final measurement labels $\mb{x}_1$ and $\mb{x}_2$.
}}
\label{fig:1Qandmore}
\end{figure}

For the single-qubit case, we have three classes, that is, $[01] = \{01,10\}, [00]=\{00\}, [11]=\{11\}$ of $\red{[2]^2}$. 
In this way, $\mc{R}$ creates superposition, say some Hadamard-like rotation, in the subspace $\mc{V}={\rm span} (\ket{01},\ket{10})$, but is inoperative on the complement subspace $\mc{V}^{\perp}={\rm span}(\ket{00},\ket{11})$. As shown in \cref{fig:1Qandmore} (a), 
\begin{equation}\label{eq:1qR}
\begin{aligned}
\mc{R}=CX_{1\rightarrow2}CH_{2\rightarrow1}CX_{1\rightarrow2}
\end{aligned}
\end{equation}
with $CX$ and $CH$ representing for the controlled-$X$ (or CNOT) and controlled-$H$ gate, respectively. An important observation is that the unitary $\mc{R}$ is followed by the projective measurement in the computational basis $\{\ket{\mb{x}}\}$, and thus the quantum gates that do not create superposition before the final measurement can be equivalently moved to the classical post-processing \cite{liu2022classically}. As such, the circuit in \cref{fig:1Qandmore} (a) can be effectively transformed to the one in \cref{fig:1Qandmore} (b). 

We can understand the structure of the circuit in \cref{fig:1Qandmore} (b) as follows. Note that the subspaces $\mathcal{V}$ and $\mc{V}^{\perp}$ depend on the parity information of the two input qubits, and $CX_{1\rightarrow2}$ followed by measurement on the second qubit aims to obtain such parity information $c$. This effectively projects (partially decoherence) the input into $\mc{V}$ or $\mc{V}^{\perp}$. Conditioned on the parity $c$, one adds a \emph{classical} controlled-$H$ on the first qubit, thus creating superposition only in $\mc{V}$ but not in $\mc{V}^{\perp}$. The measurement result of the first qubit is $x_1$, and $x_2=x_1\oplus c$ by the classical addition.
One can extend the circuit in \cref{fig:1Qandmore} (b) to the qudit case by substituting the $CX$ and $H$ gates with their qudit analogs, and the details are left in SM~\cite{supplementary}.

The circuit  compilation in the many-qubit case is non-trivial and indeed not a direct extension from the single-qubit case, since $\mc{R}\neq \bigotimes_i \mc{R}_{(i)}$, i.e., $\mc{R}$ cannot be written in a tensor-product form. Here, $\mc{R}_{(i)}$ denotes the local unitary on the $i$-th qubit pair.
Note that the standard Bell measurement nevertheless can be written into the qubit-wise tensor-product form \cite{Huggins2021Virtual,volkoff2022ancilla}, which highlights the significant difference between these two bases.
The main strategy of the following compilation is to add more entangling operations inside the first and(or) second replicas, beyond the product structure where the initial $\mc{R}_{(i)}$ only supplies entanglement across two replicas. 

Use $\mb{a}$ and $\mb{b}$ to label the $n$-bit string for the first and second replicas, respectively. 
Without loss of generality, assume that the first $n-l$ bits of $\mb{a}$ and $\mb{b}$ share the same string, and thus the last $l$ bits are different (opposite indeed). That is, $\mb{a}=\mb{a}^{[n-l]}\mb{a}^{\lfloor l \rfloor}$ and $\mb{b}=\mb{a}^{[n-l]}\overline{\mb{a}^{\lfloor l \rfloor}}$. Here, $\mb{a}^{[n-l]}$ with $[n-l]:=\{1,2,\ldots,n-l\}$ denotes the restricted bit-string on the subsystem of the first $n-l$ qubits; $\mb{a}^{\lfloor l \rfloor}$ with $\lfloor l \rfloor:=\{n-l+1,\ldots,n-1,n\}$ denotes the restricted bit-string of the last $l$ qubits; and $\overline{\mb{a}^{\lfloor l \rfloor}}$ denotes the bit-flip of $\mb{a}^{\lfloor l \rfloor}$.
By \cref{eq:Rany} we have 
\begin{widetext}
\begin{align}\label{eq:Rnbit}
\mc{R}\ket{\mb{a}^{[n-l]}\mb{a}^{\lfloor l \rfloor}}_1
\ket{\mb{a}^{[n-l]}\overline{\mb{a}^{\lfloor l \rfloor}}}_2
= \ket{\mb{a}^{[n-l]}}_1 \ket{\mb{a}^{[n-l]}}_2 \otimes
\left\{
\begin{array}{ll}
\!\frac{1}{\sqrt{2}}
\left(\ket{\overline{\mb{a}^{\lfloor l \rfloor}}}_1\ket{\mb{a}^{\lfloor l \rfloor}}_2 + \ket{\mb{a}^{\lfloor l \rfloor}}_1\ket{\overline{\mb{a}^{\lfloor l \rfloor}}}_2\right),
\quad  \mb{a}^{\lfloor l \rfloor}<\overline{\mb{a}^{\lfloor l \rfloor}}, 
\\[2ex]
\!\frac{1}{\sqrt{2}}\left(\ket{\overline{\mb{a}^{\lfloor l \rfloor}}}_1\ket{\mb{a}^{\lfloor l \rfloor}}_2 - \ket{\mb{a}^{\lfloor l \rfloor}}_1\ket{\overline{\mb{a}^{\lfloor l \rfloor}}}_2\right),
\quad  \mb{a}^{\lfloor l \rfloor}>\overline{\mb{a}^{\lfloor l \rfloor}},  
\end{array}
\right.
\end{align}
\end{widetext}
where we just extract the common state on the first $n-l$ qubits on both replicas; and the inequality $\mb{a}^{\lfloor l \rfloor}<\overline{\mb{a}^{\lfloor l \rfloor}}$ ($\mb{a}^{\lfloor l \rfloor}>\overline{\mb{a}^{\lfloor l \rfloor}}$) means that the first bit of $\mb{a}^{\lfloor l \rfloor}$ is 0 (1). 
It is not hard to see that the state on the last $l$ qubits is a GHZ-like state, which inspire the following compilation.

Let us summarize the action of $\mc{R}$ in \cref{eq:Rnbit}: if $a_i=b_i$, $\mc{R}$ does nothing; and for the rest of qubits in $[l]$ with $b_i=\overline{a_i}$, it creates a `large' superposition in the subspace 
$\mc{V}={\rm span}\left(\ket{\mb{a}^{\lfloor l \rfloor}}\ket{\overline{\mb{a}^{\lfloor l \rfloor}}}, \ket{\overline{\mb{a}^{\lfloor l \rfloor}}}\ket{\mb{a}^{\lfloor l \rfloor}}\right)$, which degenerates to ${\rm span}\left(\ket{01},\ket{10}\right)$ in  the single-qubit case. By applying the technique of QECC, we construct the quantum circuit as shown in \cref{fig:1Qandmore} (c). Conditioned on the parity check information on the second replica, we additionally add a joint unitary on the qubits in the first replica whose parity information $c^i=1$. Such joint unitary is composed of CNOTs and a Hadamard gate, which are both Clifford gates \cite{gottesman2002introduction} and can be realized in a fault-tolerant way in the future \cite{campbell2017roads}. 

Some remarks on the structure of the many-qubit quantum circuit in \cref{fig:1Qandmore} (c) are as follows. First, after the initial parity check which needs a depth-1 circuit, the remaining circuit depth equals the number of $c^i=1$, and thus the maximal circuit depth could be $n+1$ in total. Second, one needs long-range CNOTs which are realizable in systems like ion-trap \cite{figgatt2019parallel} and optical-tweezer platforms \cite{evered2023high}, but would be challenging for the ones with the nearest-neighbour architecture, like superconducting-qubit system \cite{huang2022Science,cao2023generation}. The compilation of the long-range gates in such systems would futher increase the circuit depth.
To address this,  in the main text, we developed the advanced local-AFRS protocol beyond AFRS to significantly reduce the circuit depth from $\mc{O}(n)$ to $\mc{O}(1)$ for local observables.

\begin{figure}[t]
\centering
\includegraphics[width=\linewidth]{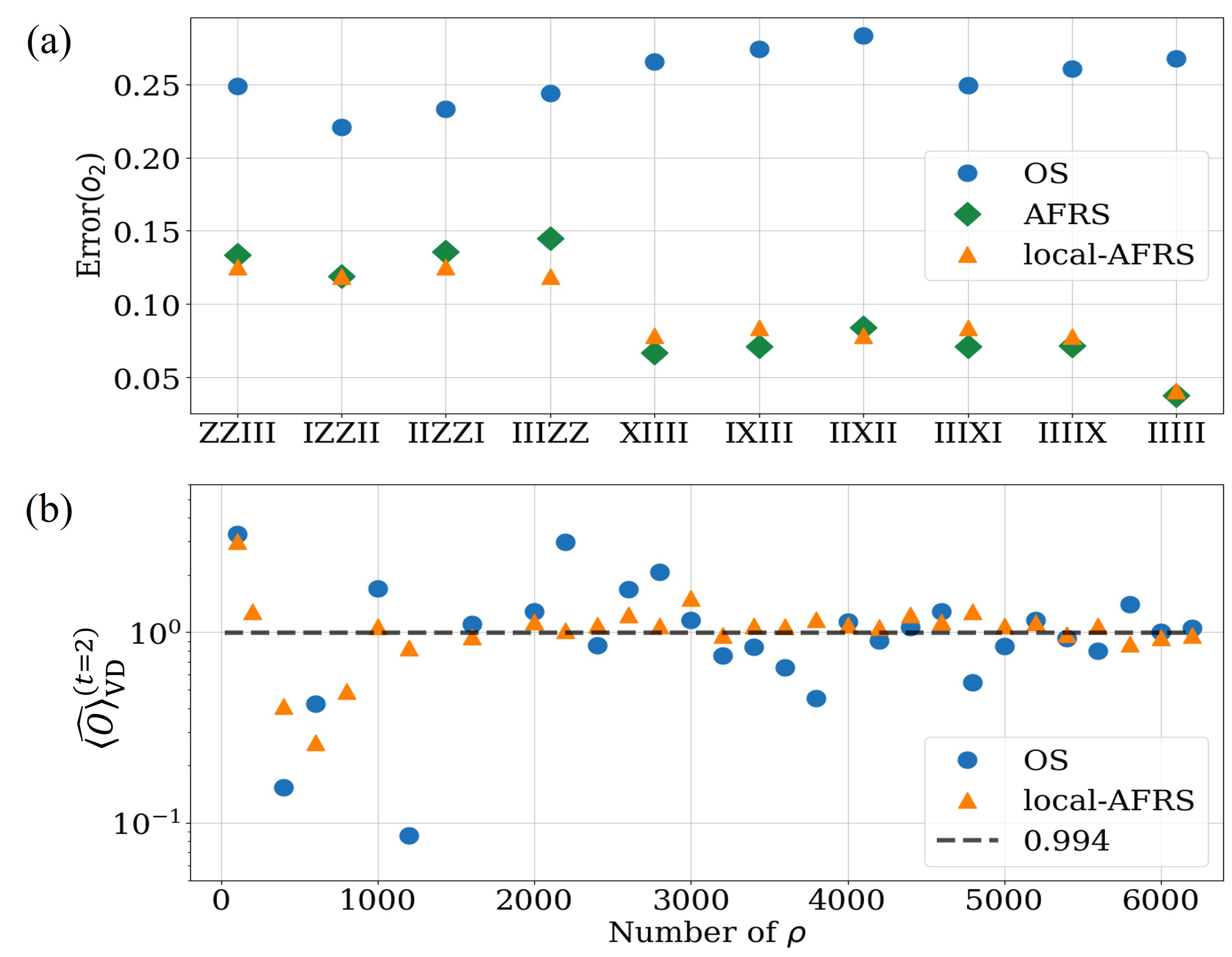}
\caption{\justifying{Estimation performance of OS, AFRS and local-AFRS protocols. The quantum state is a noisy $5$-qubit GHZ state $\rho = 0.7\ket{\text{GHZ}}\bra{\text{GHZ}}+0.3\identity_d/d$.
Local Clifford ensemble with $V\in\mc{E}_{\mathrm{LCl}}$ is utilized for evaluating $o_2 = \tr(O\rho^2)$ with the sample number $N=1000$ in (a), and $\langle O \rangle_{\mathrm{VD}}^{(t=2)}=\tr(O\rho^2)/\tr(\rho^2)$ in (b). 
The local observable in (b) is set as $O=Z_1 Z_2$, and the black dashed line represents the theoretical result of $\langle O \rangle_{\mathrm{VD}}^{(t=2)}=0.994$.}}
\label{fig:num_pauli_cliff2}
\end{figure}

\red{\textit{Estimator for local-AFRS---}Given a local observable $O=O_A\otimes \id_{\bar{A}}$, the unbiased estimator of the reduced operator $\rho^t_A=\tr_{\bar{A}}(\rho^t)$ is constructed from a single-shot measurement outcome $\mb{x}=\{\mb{x}^A,\mb{x}^{\bar{A}}\}$ as
\begin{equation}\label{eq:local_AFRS_estimator}
\widehat{\rho^t_A} = \mathrm{Re}\left(f(\mb{x}^A,\mb{x}^{\bar{A}})\right)\mathcal{M}_A^{-1}\left( V_A^\dag \ketbra{\mb{b}^A}{\mb{b}^A} V_A \right),
\end{equation}
where $\mb{b}^A$ is obtained via the mapping $\mb{x}^A\to \mb{b}^A$ and $\mathcal{M}_A^{-1}$ depends on the ensemble $\mathcal{E}^{(A)}$. The observable is then estimated by $\widehat{o_t}=\tr(\widehat{\rho^t_A} O_A)$, with variance bounded as in \cref{eq:varmLocal}. 
Further details of the construction and proof are provided in SM~\cite{supplementary}. }

\textit{Application in virtual distillation---}\red{Virtual distillation (VD) aims to purify a noisy state by accessing its `virtual' power state $\smash{\rho^{(t)}_\text{VD}=\rho^t/ \tr(\rho^t)}$, enabling improved algorithmic cooling and error mitigation \cite{cotler2019cooling,Huggins2021Virtual,Koczor2021Exponential}. Within AFRS, such tasks reduce to estimating ratios of observables, $\langle O \rangle^{(t)}_\text{VD}=\text{tr}(O\rho^{(t)}_\text{VD})=\tr(O\rho^t)/ \tr(\rho^t)$, which can be handled naturally by the AFRS framework. In particular, local-AFRS provides efficient access to Pauli observables relevant for many-body Hamiltonians, using only shallow entangling circuits and logarithmic sampling overhead.

We demonstrate the advantages of local-AFRS for VD using the transverse-field Ising Hamiltonian $H = -J \sum_j Z_jZ_{j+1}-h \sum_j X_j$, targeting Pauli observables $O\in{Z_jZ_{j+1},X_j}$. By \cref{co:localAFRS}, only two entangling unitaries, $U_\text{odd}$ and $U_\text{even}$ (see \cref{fig:local_afrs_sketch}~(c)), are needed. Measurement results from $U_\text{odd/ even}$ cover ${Z_jZ_{j+1}}$ with odd/even $j$, while both patterns suffice for ${X_i}$ and $\id$. This gives sampling time $M=\mc O(\log n)$ and entangling depth $\le 3$ for $t=2$.
\cref{fig:num_pauli_cliff2}~(a) and (b) show numerical results: local-AFRS significantly outperforms the OS protocol in estimating $\tr(O\rho^2)$ and $\langle O \rangle_\text{VD}^{(t=2)}$, with faster convergence and lower error, confirming its efficiency in VD tasks.
A full description and detailed analysis are presented in SM \cite{supplementary}.}

\onecolumngrid

\newpage

\begin{appendix}

\renewcommand{\figurename}{Supplementary Figure}
\renewcommand{\theequation}{S\arabic{equation}}
\renewcommand{\thetable}{S\arabic{table}}
\renewcommand{\thetheorem}{S\arabic{theorem}}
\renewcommand{\thelemma}{S\arabic{lemma}} 
\renewcommand{\theobservation}{S\arabic{observation}} 
\renewcommand{\theprop}{S\arabic{prop}} 
\renewcommand{\thefigure}{S\arabic{figure}}

\setcounter{secnumdepth}{3}

\renewcommand*{\theHobservation}{\theobservation}
\renewcommand*{\theHlemma}{\thelemma}

\def\eqref#1{\textup{(\ref{#1})}}  
\setcounter{figure}{0}

\tableofcontents
    
\section{Proofs of \texorpdfstring{Observation 1}{Observation 1} and \texorpdfstring{Theorem 1}{Theorem 1} in the main text}\label{SM:Th1}

\subsection{Proofs of Eq.~(1) and Observation 1 in the main text}\label{SM:Ob1}
Recall that $\Pr(\mb{b}|V)= \bra{\mb{b}}V(\rho^t)\ket{\mb{b}}$ and $Q_\mathbf{b} = t^{-1} \sum_{i=1}^t \ketbra{\mb{b}}{\mb{b}}_i \otimes \id_d^{\otimes (t-1)}$. Hence, 
Eq.~(1) in the main text can be proved by showing that 
\begin{align}\label{eq:obs1_proof}
\bra{\mb{b}}V(\rho^t)\ket{\mb{b}}=\tr \left[ S_t  \left(\ketbra{\mb{b}}{\mb{b}}_i \otimes \id_d^{\otimes (t-1)}\right) V(\rho)^{\otimes t}\right]
\quad\forall i=1,2,\dots,t.
\end{align}
We graphically show the derivation of this relation
with the diagram in \cref{fig:obs_sketch}. 

Equation~(2) in the main text can be derived as follows, 
\begin{align}
\bra{\mb{b}}V(\rho^t)\ket{\mb{b}}
&=\tr \left( S_t Q_\mathbf{b}V(\rho)^{\otimes t}\right)
=
\sum_\mb{x} \bra{\psi_\mb{x}}S_t Q_\mb{b} V(\rho)^{\otimes t}\ket{\psi_\mb{x}}
\nonumber\\
&=
\sum_\mb{x} \bra{\psi_\mb{x}}S_t\ket{\psi_\mb{x}}\bra{\psi_\mb{x}}Q_\mb{b}\ket{\psi_\mb{x}}\bra{\psi_\mb{x}}V(\rho)^{\otimes t}\ket{\psi_\mb{x}}, 
\end{align}
where the second equality holds because 
$\{\ket{\psi_\mb{x}}\}_{\mb{x}}$ forms an orthonormal basis on $\mc{H}_d^{\otimes t}$, 
and the third equality holds because 
$Q_\mb{b}$, $S_t$, and $\ketbra{\psi_\mb{x}}{\psi_\mb{x}}$ commute with each other (recall that $S_t$ and $Q_\mathbf{b}$ can be diagonalized simultaneously in the basis $\{\ket{\psi_\mb{x}}\}$). 

\begin{figure}[htbp]
\centering
\includegraphics[width=0.76\textwidth]{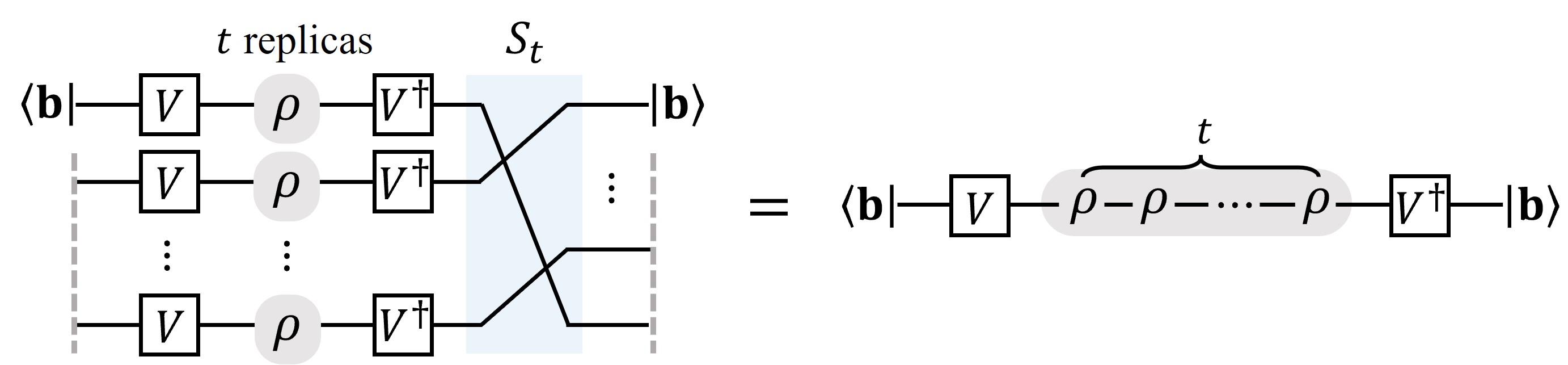}\\
\caption{\justifying{A proof sketch of \cref{eq:obs1_proof} for the AFRS protocol. The gray dashed lines denote the periodic boundary, that is, the trace operation.
}}
\label{fig:obs_sketch}
\end{figure}

\subsection{Proof of 
Theorem 1 in the main text}\label{SM:Th1proof}

The proof of 
Theorem 1 is separated into two parts. Using 
Observation 1 in the main text, the first part in Supplementary Note~\ref{SM:Th1part1} shows that $\widehat{\rho^t}$ defined in 
Eq.~(3) of the main text reproduces the $t$-th power of the underlying state $\rho$ exactly in expectation (over the randomness in the choice of unitary $V$, measurement outcome $\mb{x}$, and the mapping strategy $\mathcal{G}_{\mb{x}\to \mb{b}}$), i.e., $\mathbb{E}_{\{V,\mb{x},\mb{b}\}}\widehat{\rho^t} = \rho^t$. The second part in Supplementary Note~\ref{SM:Th1part2} proves the upper bound of $\mathrm{Var}(\widehat{o_t})$ given in 
Eq.~(4) of the main text.

\subsubsection{Unbiasedness of the estimator \texorpdfstring{$\widehat{\rho^t}$}{}}\label{SM:Th1part1}
\red{Our estimator $\widehat{\rho^t}$ for $\rho^t$ is defined as
\begin{equation}\label{eq:shadow2}
    \widehat{\rho^t}:= \mathrm{Re}(f(\mb{x}))\cdot\mathcal{M}^{-1}\left( V^{\dag} \ketbra{\mb{b}}{\mb{b}} V \right), 
\end{equation}
where $\mathcal{M}^{-1}$ is the inverse of the following measurement channel introduced in Ref.~\cite{huang2020predicting}:
\begin{align}\label{eq:channel_def}
\mathcal{M}(\cdot):=
\mathbb{E}_{V \sim \mc{E}}  \sum_{\mb{b}}
\bra{\bf{b}}V(\cdot)\ket{\bf{b}}  V^{\dag} \ketbra{\mb{b}}{\mb{b}} V . 
\end{align}
By construction, $\mc{M}$ is linear, invertible, and depends only on the system dimension and ensemble $\mc{E}$. For the global Clifford ensemble (a unitary 2-design), one obtains the closed form $\mathcal{M}(\cdot)=
\frac{(\cdot)+\tr(\cdot)\mathbb{I}}{d+1}$ and $
\mathcal{M}^{-1}(\cdot)= (d+1)(\cdot)-\tr(\cdot)\mathbb{I}$, with $d=2^n$. For the local Clifford (or Pauli) ensemble, the channel factorizes qubit-wise.
}

The expectation of $\widehat{\rho^t}$ is 
\begin{align}\label{eq:praverage}
\mathbb{E}_{\{V,\mb{x},\mb{b}\}} \widehat{\rho^t}
&=\mathbb{E}_{\{V,\mb{x},\mb{b}\}} \, {\mathrm{Re}}(f(\mb{x}))\cdot\mathcal{M}^{-1}\left( V^{\dag} \ketbra{\mb{b}}{\mb{b}} V \right) 
\nonumber\\
&=\mathbb{E}_{V \sim \mc{E}} \sum_{\mb{x},\mb{b}} \Pr(\mb{x},\mb{b}|V) \,
{\mathrm{Re}}(f(\mb{x}))
\mathcal{M}^{-1}\left( V^{\dag} \ketbra{\mb{b}}{\mb{b}} V \right) 
\nonumber\\
&=
\mathbb{E}_{V \sim \mc{E}}  \sum_{\mb{b}}\mathcal{M}^{-1}\left( V^{\dag} \ketbra{\mb{b}}{\mb{b}} V \right) 
\left[\sum_{\mb{x}}
\Pr(\mb{x}|V) 
\Pr(\mb{b}|\mb{x})\,
{\mathrm{Re}}(f(\mb{x}))
\right]. 
\end{align}
Here, the third equality holds because $\Pr(\mb{x},\mb{b}|V)=\Pr(\mb{x}|V) 
\Pr(\mb{b}|\mb{x})$, since the mapping result $\mb{b}$ only depends on the real measurement outcome $\mb{x}$ according to  the  mapping strategy $\mathcal{G}_{\mb{x}\to \mb{b}}$ in the main text.

Since $\Pr(\mb{x}|V)$ and $\Pr(\mb{b}|\mb{x})$ are real numbers, the last term in the parentheses of \cref{eq:praverage} can be further written as 
\begin{align}
{\mathrm{Re}}\left[\sum_{\mb{x}}
\Pr(\mb{x}|V) 
\Pr(\mb{b}|\mb{x})
f(\mb{x})
\right] 
= {\mathrm{Re}}\left[\Pr(\mb{b}|V) \right]
= \bra{\bf{b}}V(\rho^t)\ket{\bf{b}},
\end{align}
where the first equality follows from 
Observation 1 in the main text. 
By plugging this relation into \cref{eq:praverage}, we get 
\begin{align}
\mathbb{E}_{\{V,\mb{x},\mb{b}\}} \widehat{\rho^t}
&=\mathbb{E}_{V \sim \mc{E}}  \sum_{\mb{b}}
\bra{\bf{b}}V(\rho^t)\ket{\bf{b}} \mathcal{M}^{-1}\left( V^{\dag} \ketbra{\mb{b}}{\mb{b}} V \right) 
\nonumber\\
&=
\mathcal{M}^{-1}\left[ \mathbb{E}_{V \sim \mc{E}}  \sum_{\mb{b}}
\bra{\bf{b}}V(\rho^t)\ket{\bf{b}}  V^{\dag} \ketbra{\mb{b}}{\mb{b}} V \right]
=
\mathcal{M}^{-1}\left[ \mathcal{M} (\rho^t) \right] 
=\rho^t,
\end{align}
where we have used the definition of $\mathcal{M}$ in \cref{eq:channel_def}. 
In conclusion, $\widehat{\rho^t}$ is an unbiased estimator of $\rho^t$.

\subsubsection{Proof of 
Eq.~(4) in the main text}\label{SM:Th1part2}

Since $\widehat{\rho^t}$ and $O$ are Hermitian operators, the estimator $\widehat{o_t}=\tr\left(O\widehat{\rho^t}\right)$ is a real number, and its variance can be written as 
\begin{align}\label{eq:Var1}
\mathrm{Var}(\widehat{o_t}) 
= \mathbb{E}_{\{V,\mb{x},\mb{b}\}} \left[\widehat{o_t}^2 \right] - \left(\mathbb{E}_{\{V,\mb{x},\mb{b}\}} \widehat{o_t}\right) ^2 
= \mathbb{E}_{\{V,\mb{x},\mb{b}\}} \left[\widehat{o_t}^2\right] - \left[\tr(O\rho^t)\right]^2. 
\end{align}
Note that the inverse of $\mathcal{M}$ in \cref{eq:channel_def} is self-adjoint, i.e., $\tr\!\left[X \mathcal{M}^{-1}(Y)\right]=\tr\!\left[\mathcal{M}^{-1}(X) Y\right]$ for any pair of operators $X, Y$ on $\mc{H}_d$. 
So we have 
\begin{align}
\widehat{o_t}
=
\tr(O\widehat{\rho^t})
=  
{\mathrm{Re}}(\bra{\psi_\mb{x}}S_t\ket{\psi_\mb{x}})
\tr\left[O \mathcal{M}^{-1}\left( V^{\dag} \ketbra{\mb{b}}{\mb{b}} V \right)\right]
=  
{\mathrm{Re}}(\bra{\psi_\mb{x}}S_t\ket{\psi_\mb{x}}) 
\bra{\mb{b}} V \mathcal{M}^{-1}(O)  V^{\dag} \ket{\mb{b}} .
\end{align}
It follows that 
\begin{align}\label{eq:Var2}
\mathbb{E}_{\{V,\mb{x},\mb{b}\}} \left[\widehat{o_t}^2 \right]
&=
\mathbb{E}_{V \sim \mc{E}} \sum_{\mb{x},\mb{b}} \Pr(\mb{x}|V) 
\Pr(\mb{b}|\mb{x}) \,\widehat{o_t}^2
\nonumber\\
&\stackrel{(a)}{=} 
\mathbb{E}_{V \sim \mc{E}} \sum_{\mb{x},\mb{b}} \bra{\psi_\mb{x}}V(\rho)^{\otimes t}\ket{\psi_\mb{x}} \bra{\psi_\mb{x}} Q_\mb{b}\ket{\psi_\mb{x}} \left[{\mathrm{Re}}(\bra{\psi_\mb{x}}S_t\ket{\psi_\mb{x}})\right]^2 
\bra{\mb{b}} V \mathcal{M}^{-1}(O)  V^{\dag} \ket{\mb{b}}^2
\nonumber\\
&\stackrel{(b)}{\leq} 
\mathbb{E}_{V \sim \mc{E}} \sum_{\mb{b}} \bra{\mb{b}} V \mathcal{M}^{-1}(O)  V^{\dag} \ket{\mb{b}}^2
\underbrace{\left(\sum_{\mb{x}}  \bra{\psi_\mb{x}}V(\rho)^{\otimes t}\ket{\psi_\mb{x}} \bra{\psi_\mb{x}} Q_\mb{b}\ket{\psi_\mb{x}} \right)}_{(*)} , 
\end{align}
where $(a)$ holds because $\Pr(\mb{x}|V)=\bra{\psi_\mb{x}}V(\rho)^{\otimes t}\ket{\psi_\mb{x}}$ and $\Pr(\mb{b}|\mb{x})=\bra{\psi_\mb{x}} Q_\mb{b}\ket{\psi_\mb{x}}$, 
and $(b)$ holds because $\left[\mathrm{Re}(f(\mb{x}))\right]^2=\left[\mathrm{Re}(\bra{\psi_\mb{x}}S_t\ket{\psi_\mb{x}})\right]^2\leq1$.
The term $(*)$ in \cref{eq:Var2} can be further written as  
\begin{align}\label{eq:Var3}
(*)
&\stackrel{(a)}{=} 
\sum_{\mb{x}} \bra{\psi_\mb{x}}V(\rho)^{\otimes t} Q_\mb{b} \ket{\psi_\mb{x}}
\stackrel{(b)}{=} 
\tr\!\left[V(\rho)^{\otimes t} Q_\mb{b}\right]
=
\bra{\mb{b}}V(\rho)\ket{\mb{b}}, 
\end{align}
where $(a)$ holds because $Q_\mb{b}$ commutes with $\ketbra{\psi_\mb{x}}{\psi_\mb{x}}$, and $(b)$ holds because $\{\ket{\psi_\mb{x}}\}_{\mb{x}}$ forms an orthonormal basis on $\mc{H}_d^{\otimes t}$.
Equations \eqref{eq:Var2} and \eqref{eq:Var3} together imply that 
\begin{align}\label{eq:Var4}
\mathbb{E}_{\{V,\mb{x},\mb{b}\}} \left[\widehat{o_t}^2 \right]
\leq 
\mathbb{E}_{V \sim \mc{E}} \sum_{\mb{b}} \bra{\mb{b}} V \mathcal{M}^{-1}(O)  V^{\dag} \ket{\mb{b}}^2
\bra{\mb{b}}V(\rho)\ket{\mb{b}}
=\mathbb{E}\left[\operatorname{tr}(O \hat{\rho})^2\right] ,
\end{align}
where $\hat{\rho}$ is the original single-copy shadow snapshot \cite{huang2020predicting}. 
Hence, the expectation value of $\widehat{o_t}^2$ is not larger than that of $\left[\tr(O \hat{\rho})\right]^2$ in the original shadow protocol on a single-copy of $\rho$, no matter what value $t$ takes.

By Eqs.~\eqref{eq:Var1} and \eqref{eq:Var4}, we have 
\begin{align}\label{eq:Var5}
\mathrm{Var}(\widehat{o_t}) 
&\leq \mathbb{E} \left[\operatorname{tr}(O \hat{\rho})^2\right]- \left[\tr(O\rho^t)\right]^2
\nonumber\\
&= \mathbb{V}^{\mc{E}}(O,\rho) + \left[\tr(O\rho)\right]^2 - \left[\tr(O\rho^t)\right]^2
\nonumber\\
&\stackrel{(a)}{=} 
\mathbb{V}^{\mc{E}}(O_0,\rho) + \left[\tr(O\rho)\right]^2 - \left[\tr(O\rho^t)\right]^2
\nonumber\\
&\stackrel{(b)}{\leq}
\left\|O_0\right\|_{\mathrm{sh},\mathcal{E}}^2 + \|O\|_{\infty}^2 , 
\end{align}
which confirms 
Eq.~(4) in the main text. Here,  $\mathbb{V}^{\mc{E}}(O,\rho)$ is the variance of $\tr(O\hat\rho)$ associated with the original shadow protocol on a single-copy of $\rho$ \cite{huang2020predicting}, $(a)$ holds because shifting the operator $O$ to its traceless part $O_0 $  does not change the variance of $\tr(O\hat\rho)$, $(b)$ holds because 
$\mathbb{V}^{\mc{E}}(O_0,\rho)\leq \left\|O_0\right\|_{\mathrm{sh},\mathcal{E}}^2$ and $\left|\tr(O\rho)\right| \leq \|O\|_\infty$.

\section{Proof of 
Lemma 1 in END MATTER}\label{SM:Lemma1}
Suppose $[\mathbf{z}]$, $[\mathbf{y}]$ are two classes of $\red{[d]^t}$, integers $k_1\in\{0,1,\dots,|[\mathbf{z}]|-1\}$ and $k_2\in\{0,1,\dots,|[\mathbf{y}]|-1\}$. 
When $[\mathbf{z}]\ne [\mathbf{y}]$, the two states $\ket{\Psi_{[\mathbf{z}]}^{(k_1)}}$ and $\ket{\Psi_{[\mathbf{y}]}^{(k_2)}}$ are orthogonal since $\langle\tau^{r_2}(\mathbf{z})|\tau^{r_1}(\mathbf{y})\rangle=0$ for any integers $r_1$ and $r_2$. 
When $[\mathbf{z}]=[\mathbf{y}]$ but $k_1\ne k_2$, the two states $\ket{\Psi_{[\mathbf{z}]}^{(k_1)}}$ and $\ket{\Psi_{[\mathbf{y}]}^{(k_2)}}$ are still orthogonal since 
\begin{align}
\left\langle \Psi_{[\mathbf{z}]}^{(k_1)}\Big|\Psi_{[\mathbf{y}]}^{(k_2)}\right\rangle &= \frac{1}{|[\mathbf{z}]|} \sum_{r=0}^{|[\mathbf{z}]|-1} \exp\left(\frac{r(k_2-k_1)\cdot 2\pi \rm{i}}{|[\mathbf{z}]|}\right) 
=0,  
\end{align}
where the last inequality holds because $k_2-k_1\ne 0\mod |[\mathbf{z}]|$. 
In addition, notice that the size of the set $\left\{\ket{\Psi_{[\mathbf{z}]}^{(k)}}\right\}_{[\mathbf{z}],k}$ equals $\red{d^t}$, and thus equals the dimension of 
$\mc{H}_d^{\otimes t}$. 
Therefore, $\left\{\ket{\Psi_{[\mathbf{z}]}^{(k)}}\right\}_{[\mathbf{z}],k}$
forms an orthonormal basis on $\mathcal{H}_d^{\otimes t}$. 

Recall that $S_t\ket{\mathbf{x}}=\ket{\tau(\mathbf{x})}$ for $\mathbf{x}\in\red{[d]^t}$. 
So we have 
\begin{align}
S_t\ket{\Psi_{[\mathbf{z}]}^{(k)}} 
&= \frac{1}{\sqrt{|[\mathbf{z}]|}} \sum_{r=0}^{|[\mathbf{z}]|-1} \exp\left(\frac{rk\cdot 2\pi \rm{i}}{|[\mathbf{z}]|}\right) S_t\ket{\tau^r(\mathbf{z})}
\nonumber\\
&= \frac{1}{\sqrt{|[\mathbf{z}]|}} \sum_{r=0}^{|[\mathbf{z}]|-1} \exp\left(\frac{rk\cdot 2\pi \rm{i}}{|[\mathbf{z}]|}\right) \ket{\tau^{r+1}(\mathbf{z})}
\nonumber\\
&= \frac{1}{\sqrt{|[\mathbf{z}]|}} 
\sum_{r'=0}^{|[\mathbf{z}]|-1} \exp\left(\frac{(r'-1)k\cdot 2\pi \rm{i}}{|[\mathbf{z}]|}\right) \ket{\tau^{r'}(\mathbf{z})}
\nonumber\\
&= \exp\left(-\frac{k\cdot 2\pi \rm{i}}{|[\mathbf{z}]|}\right) \frac{1}{\sqrt{|[\mathbf{z}]|}} \sum_{r'=0}^{|[\mathbf{z}]|-1} \exp\left(\frac{r'k\cdot 2\pi i}{|[\mathbf{z}]|}\right) \ket{\tau^{r'}(\mathbf{z})}
\nonumber\\
&= \exp\left(-\frac{k\cdot 2\pi \rm{i}}{|[\mathbf{z}]|}\right) \ket{\Psi_{[\mathbf{z}]}^{(k)}} . 
\end{align}
Hence, $\ket{\Psi_{[\mathbf{z}]}^{(k)}}$ is an eigenstate of  $S_t$ with eigenvalue
$\exp\left(-\frac{k\cdot 2\pi \rm{i}}{|[\mathbf{z}]|}\right)$.

For $\mathbf{x}\in\red{[d]^t}$, we have 
\begin{align}
\left(\ketbra{\mb{b}}{\mb{b}}_i \otimes \id_d^{\otimes (t-1)}\right) \ket{\mathbf{x}}= 
\begin{cases}
\ket{\mathbf{x}}  & \text{the $i$-th number of $\mathbf{x}$ equals $\mb{b}$}, \\
0   & \text{otherwise}.  
\end{cases}
\end{align}
Consequently,  we have 
\begin{align}
Q_\mathbf{b}\ket{\mathbf{x}}
= \frac{1}{t} \sum_i \left( \ketbra{\mb{b}}{\mb{b}}_i \otimes \id_d^{\otimes (t-1)}\right) \ket{\mathbf{x}}
= \frac{\text{the number of $\mb{b}$ in the string $\mathbf{x}$}}{t} \ket{\mathbf{x}},  
\end{align} 
and 
\begin{align}
Q_\mathbf{b}\ket{\Psi_{[\mathbf{z}]}^{(k)}} 
&=\frac{1}{\sqrt{|[\mathbf{z}]|}} \sum_{r=0}^{|[\mathbf{z}]|-1} \exp\left(\frac{rk\cdot 2\pi \rm{i}}{|[\mathbf{z}]|}\right) Q_\mathbf{b}\ket{\tau^r(\mathbf{z})}
\nonumber\\
&=\frac{1}{\sqrt{|[\mathbf{z}]|}} \sum_{r=0}^{|[\mathbf{z}]|-1} \exp\left(\frac{rk\cdot 2\pi \rm{i}}{|[\mathbf{z}]|}\right) 
\frac{\text{the number of $\mb{b}$ in the string $\tau^r(\mathbf{z})$}}{t} \ket{\tau^r(\mathbf{z})}
\nonumber\\
&=\frac{\text{the number of $\mb{b}$ in the string $\mathbf{z}$}}{t} \times 
\frac{1}{\sqrt{|[\mathbf{z}]|}} \sum_{r=0}^{|[\mathbf{z}]|-1} \exp\left(\frac{rk\cdot 2\pi \rm{i}}{|[\mathbf{z}]|}\right) 
\ket{\tau^r(\mathbf{z})}
\nonumber\\
&=\frac{\text{the number of $\mb{b}$ in the string $\mathbf{z}$}}{t} \ket{\Psi_{[\mathbf{z}]}^{(k)}} , 
\end{align} 
where the third equality holds because the strings $\mathbf{z}$ and $\tau^r(\mathbf{z})$ contains the same number of $\mb{b}$ for any integer $r$. 
Hence, $\ket{\Psi_{[\mathbf{z}]}^{(k)}}$ is an eigenstate of  $Q_\mathbf{b}$ with eigenvalue $t^{-1}\times($the number of $\mb{b}$ in the string $\mathbf{z})$. 
This completes the proof.

\section{More details on the circuit compilation of \texorpdfstring{$\mc{R}$}{R}}\label{SM:compilation}

\subsection{Single-qudit case}

\subsubsection{Circuit compilation}
One can follow the route of the single-qubit case in END MATTER to find the circuit compilation of $\mc{R}$ in the general single-qudit cases.
For $t=2$, the permutation $\tau$ of $S_2$ is just the swap operation, and any string in $\red{[d]^2}$ can be written as $ab$ for $a,b\in\{0,\dots,d-1\}$. 
If $a= b$, the class $[ab]=\{ab\}$; if $a<b$, the class $[ab]=\{ab,ba\}$. 
In this way, the unitary $\mc{R}$ acting on $\mc{H}_d^{\otimes 2}$ satisfies 

\begin{align}\label{eq:Rqudit}
\mathcal{R}|ab\>= 
\begin{cases}
|ab\>     & a=b, \\
\frac{1}{\sqrt{2}}(|ba\>+|ab\>)  & a<b, \\
\frac{1}{\sqrt{2}}(|ba\>-|ab\>)  & a>b, \\  
\end{cases}
\end{align} 
where$\mc{R}=\mc{R}^{\dag}$ is Hermitian for $t=2$.

\red{The compilation strategy is similar to the single-qubit case: first use a measurement on the second replica to project the input into a sum-subspace, then apply a controlled single-qudit unitary on the first replica that creates the required symmetric/antisymmetric superpositions on each two-dimensional block.
Concretely, write the integer sum $c=a+b$ with $a,b\in\{0,\dots,d-1\}$, and consider the subspace $\mathcal{V}_{c}={\rm span}\left(\{\ket{ab}\}_{a+b=c}\right)$ with $c=0,1,\dots,2d-2$.} Each $\mathcal{V}_{c}$ can be written as a direct sum of several 1- or 2-dimensional subspace $\mathcal{V}_{a,b}={\rm span}\left(\ket{ab},\ket{ba}\right)$, i.e., 
\begin{align}
\mathcal{V}_{c}=\bigoplus_{\substack{a,b\in\{0,1,\dots,d-1\} \\ a+b=c}}\mathcal{V}_{a,b},  
\end{align}
{Hence, after the projection onto $\mc{V}_c$, one can create the desired superposition on each $\mathcal{V}_{a,b}$ by performing an appropriate unitary $U_c$ on the first replica.}

\red{To make the integer sum $c=a+b$ an unambiguous label, we enlarge the local dimension of the second replica from $d$ to $2d$. 
Physically, this means the second replica has computational basis $\{\ket{0},\dots.\ket{2d-1}\}$, so that each integer $c\in\{0,1,\dots,2d-2\}$ can be represented by a distinct basis state $\ket{c}$. Operationally, we perform a controlled $2d$-shift from the first to the second replica, $X_{2d}=\sum_{j=0}^{2d-1}\ket{j+1 \ (\text{mod}\ 2d)}\bra{j}$, so that a controlled-$X_{2d}$ maps $\ket{a}\ket{b}\mapsto\ket{a}\ket{a+b\ (\text{mod}\ 2d)}$. Measuring the second replica in the $\{|i\>\}_{i=0}^{2d-1}$ basis then gives the integer $c$; because the protocol populates only $c\in\{0,\dots,2d-2\}$, the modulo-$2d$ wrap-around does not create ambiguity and the measured $c$ directly identifies $\mc{V}_c$.}

\begin{figure}[htbp]
\centering
\includegraphics[width=0.45\textwidth]{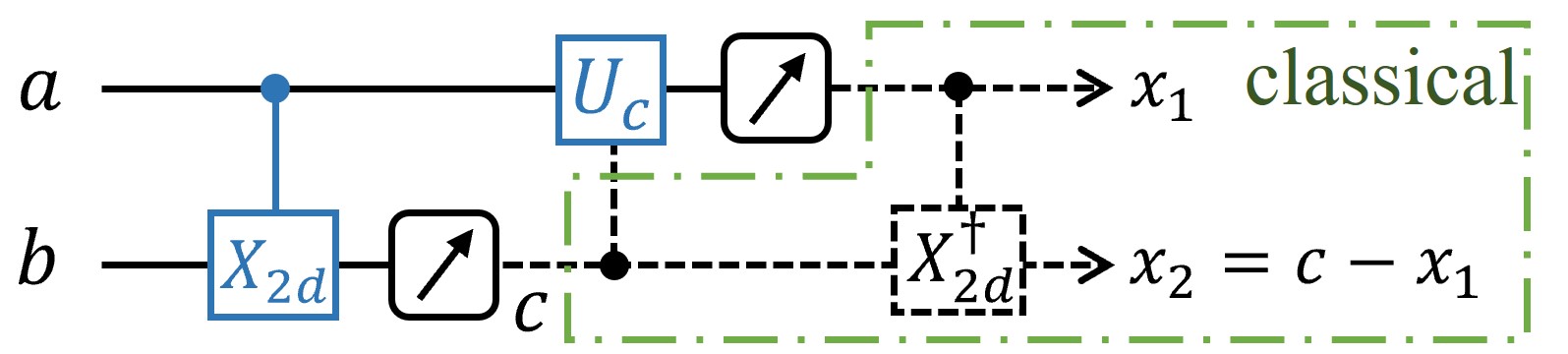}\\
\caption{Circuit compilation of $\mathcal{R}$ in the single-qudit case.
}
\label{fig:quditR}
\end{figure}

\red{After obtaining $c$, we set $q_{1}:=\max\{0,c-d+1\}$ and $q_{2}:=\min\{c,d-1\}$, so that the allowed values of $a$ with $a+b=c$ are $a\in\{q_1,\dots,q_2\}$.}
The form of the unitary $U_c$ controlled by the measurement result $c$ depends on the parity of $c$.
For odd $c$, we have
\begin{align}
U_c=  &\sum_{j=q_1}^{\frac{c-1}{2}} \left( \frac{|j\>+|c-j\>}{\sqrt{2}}\right) \<j| 
+ \sum_{j=\frac{c+1}{2}}^{q_2} \left( \frac{|c-j\>-|j\>}{\sqrt{2}}\right) \<j| + \sum_{j<q_1, j>q_2} \ketbra{j}{j}.
\end{align}
For even $c$, we have
\begin{align}
U_c=  &\sum_{j=q_1}^{\frac{c}{2}-1} \left( \frac{|j\>+|c-j\>}{\sqrt{2}}\right) \<j| 
+ \ketbra{\frac{c}{2}}
+ \sum_{j=\frac{c}{2}+1}^{q_2} \left( \frac{|c-j\>-|j\>}{\sqrt{2}}\right) \<j| + \sum_{j<q_1, j>q_2} \ketbra{j}{j}. 
\end{align}
\red{The index ranges are algebraically consistent: with $q_{1}:=\max\{0,c-d+1\}$ and $q_{2}:=\min\{c,d-1\}$, one always has $q_2\geq q_1$ for every $c\in{0,\dots,2d-2}$. Hence, $U_c$ is well-defined.} Besides, $U_c$ only introduces some superposition on selected pairs of computational bases of a single-qudit, and we remark such operations could be realized, for example, in optical systems by path encoding.
Finally, we measure the first replica in the computational basis $\{|i\>\}_{i=0}^{d-1}$, and denote the measurement result as $x_1$. By classical post-processing, we let $x_1$ and $x_2=c-x_1$ be the final output of the first and second replica, respectively. 

\red{The procedure is illustrated in \cref{fig:quditR}: we expand the second replica to $2d$ levels, perform controlled-$X_{2d}$, measure the second replica to obtain $c$, apply $U_c$ on the first replica, and finally measure the first replica. The same idea can be generalized to higher $t$, where class decomposition and conditional block unitaries become more involved.} 
\red{Note that the $2d$-expansion is an enlargement of the local levels of the existing second replica. In practice, many platforms allow access to higher local levels (or multiple paths/modes) and thus implementing a $2d$-level system can be experimentally more convenient than introducing and controlling an additional qudit. Examples include higher internal or motional levels in trapped ions, multi-path/multi-mode encodings in photonics, and higher excited levels in superconducting circuits. }

\subsubsection{Verification}
To verify that our circuit indeed implements the desired unitary $\mc{R}$ and the subsequent measurement, let us consider the evolution of 
the 2-qudit joint system. 
The initial state
of the joint system can be expressed as 
$\ket{\Phi_0}=\sum_{a,b=0}^{d-1} t_{a,b}\ket{ab}$. 
After the $CX$ operation, the joint state evolves into
$\ket{\Phi_1}=\sum_{a,b=0}^{d-1} t_{a,b}\ket{a}\ket{a+b}$. 
After measuring the second qudit, depending on the measurement outcome $c$, the reduced state on the first replica reads 
\begin{align}
\ket{\Phi_2}=\frac{(\id\otimes \bra{c})\ket{\Phi_1}}{\sqrt{\Pr_2(c)}}
= \frac{1}{\sqrt{\Pr_2(c)}} \sum_{j=q_1}^{q_2} t_{j,c-j} \ket{j}, 
\end{align}
where
\begin{align}
{\Pr}_2(c)= \bra{\Phi_1}\left(\id\otimes \ketbra{c}\right)\ket{\Phi_1}
= \sum_{j=q_1}^{q_2} \left| t_{j,c-j} \right|^2
\end{align}
is the probability of obtaining outcome $c$. Next, we consider two cases depending on the parity of $c$. 
\begin{itemize}
\item[1.] The measurement outcome $c$ is odd. \\
After the unitary evolution $U_c$, the state on the first replica reads 
\begin{align}
\ket{\Phi_3}=U_c\ket{\Phi_2}
= \frac{1}{\sqrt{\Pr_2(c)}} \left( \sum_{j=q_1}^{\frac{c-1}{2}} t_{j,c-j} \frac{|j\>+|c-j\>}{\sqrt{2}}
+\sum_{j=\frac{c+1}{2}}^{q_2} t_{j,c-j} \frac{|c-j\>-|j\>}{\sqrt{2}} \right). 
\end{align}
If we measure the first replica at this point, then the conditional probability of obtaining outcome $x_1$ is given by 
\begin{align}\label{eq:Pr1x1odd}
{\Pr}_1(x_1|c)= \left|\<x_1\!\ket{\Phi_3} \right|^2
=\frac{1}{\Pr_2(c)}\times \left\{
\begin{array}{lll}
0  &\  \text{when $x_1<q_1$ or $x_1>q_2$}, 
\\ [0.5ex]
\frac{1}{2}\left| t_{c-x_1,x_1}+t_{x_1,c-x_1} \right|^2 &\  \text{when $q_1\leq x_1\leq \frac{c-1}{2}$} ,
\\ [0.5ex]
\frac{1}{2}\left| t_{c-x_1,x_1}-t_{x_1,c-x_1}\right|^2 &\  \text{when $\frac{c+1}{2}\leq x_1\leq q_2$} . 
\end{array}
\right.
\end{align}

\item[2.] The measurement outcome $c$ is even. \\
After the unitary evolution $U_c$, the state on the first replica reads 
\begin{align}
\ket{\Phi_3}=U_c\ket{\Phi_2}
= \frac{1}{\sqrt{\Pr_2(c)}} \left( \sum_{j=q_1}^{\frac{c}{2}-1} t_{j,c-j} \frac{|j\>+|c-j\>}{\sqrt{2}}
+t_{\frac{c}{2},\frac{c}{2}} \ket{\frac{c}{2}}
+\sum_{j=\frac{c}{2}+1}^{q_2} t_{j,c-j} \frac{|c-j\>-|j\>}{\sqrt{2}} \right).
\end{align}  
If we measure the first replica at this point, then the conditional probability of obtaining outcome $x_1$ is given by 
\begin{align}\label{eq:Pr1x1even}
{\Pr}_1(x_1|c)= \left|\<x_1\!\ket{\Phi_3} \right|^2
=\frac{1}{\Pr_2(c)}\times \left\{
\begin{array}{lll}
0  &\  \text{when $x_1<q_1$ or $x_1>q_2$}, 
\\ [0.5ex]
| t_{\frac{c}{2},\frac{c}{2}} |^2  &\  \text{when $x_1=\frac{c}{2}$}, 
\\ [0.5ex]
\frac{1}{2}\left| t_{c-x_1,x_1}+t_{x_1,c-x_1} \right|^2 &\  \text{when $q_1\leq x_1\leq \frac{c}{2}-1$} ,
\\ [0.5ex]
\frac{1}{2}\left| t_{c-x_1,x_1}-t_{x_1,c-x_1}\right|^2 &\  \text{when $\frac{c}{2}+1\leq x_1\leq q_2$} . 
\end{array}
\right.
\end{align} 
\end{itemize} 
Recall that the final result for the second qudit is $x_2=c-x_1$. 
According to Eqs.~\eqref{eq:Pr1x1odd} and \eqref{eq:Pr1x1even}, for our circuit, 
the overall probability of obtaining outcomes $x_1$ on the first qudit and $x_2$ on the second qudit reads 
\begin{align}\label{eq:Prx1x2} 
\Pr(x_1,x_2)
&={\Pr}_1(x_1|c=x_1+x_2)
\,{\Pr}_2(c=x_1+x_2)
=\left\{
\begin{array}{lll}
| t_{x_1,x_1} |^2  &\  \text{when $x_1=x_2$}, 
\\ [0.5ex]
\frac{1}{2} \left| t_{x_2,x_1}+t_{x_1,x_2} \right|^2 &\  \text{when $x_1< x_2$} ,
\\ [0.5ex]
\frac{1}{2} \left| t_{x_2,x_1}-t_{x_1,x_2}\right|^2 &\  \text{when $x_1> x_2$} . 
\end{array}
\right.
\end{align} 

On the other hand, if the initial state $\ket{\Phi_0}$ is evolved under $\mc{R}$, then  it becomes 
\begin{align}
\ket{\Phi_1'}=\mc{R}\ket{\Phi_0}
= \sum_{a=0}^{d-1} t_{a,a}\ket{aa} 
 +\sum_{a<b} t_{a,b}\frac{\ket{ba}+\ket{ab}}{\sqrt{2}}
 +\sum_{a>b} t_{a,b}\frac{\ket{ba}-\ket{ab}}{\sqrt{2}}. 
\end{align}
If we measure this state in the computational basis, then the probability of obtaining outcomes $x_1$ on the first qudit and $x_2$ on the second qudit reads 
\begin{align}
\left|\<x_1,x_2\!\ket{\Phi_1'} \right|^2
=\left\{
\begin{array}{lll}
| t_{x_1,x_1} |^2  &\  \text{when $x_1=x_2$}, 
\\ [0.5ex]
\frac{1}{2} \left| t_{x_2,x_1}+t_{x_1,x_2} \right|^2 &\  \text{when $x_1< x_2$} ,
\\ [0.5ex]
\frac{1}{2} \left| t_{x_2,x_1}-t_{x_1,x_2}\right|^2 &\  \text{when $x_1> x_2$} . 
\end{array}
\right.
\end{align}
This probability exactly equals the outcome probability of our circuit shown in \cref{eq:Prx1x2}. 
Due to linearity, this result also holds when the initial state of the 2-qudit system is mixed. 
Therefore, our circuit in \cref{fig:quditR} can indeed realize the desired unitary $\mc{R}$ and the subsequent measurement.

\subsection{Many-qubit case}

\subsubsection{Circuit compilation}
Here, we give more explanations on our circuit compilation of $\mc{R}$ in the many-qubit case. 
We use $\mb{a}$ and $\mb{b}$ to label the $n$-qubit basis for the first and second replicas, respectively. As mentioned in END MATTER, 
the action of $\mc{R}$ on $|\mb{a}\>|\mb{b}\>$ is as follows: if $a^i=b^i$, $\mc{R}$ does nothing; and for the rest of qubits in $[l]:=\{i\,|\,c^i=1\}$ with $a^i\ne b^i$, $\mc{R}$ creates a GHZ-like `large' superposition in the subspace $\mathcal{V}_{\mb{a}^{[l]}}={\rm span}\!\left(\ket{\mb{a}^{[l]}}\ket{\overline{\mb{a}^{[l]}}}, \ket{\overline{\mb{a}^{[l]}}}\ket{\mb{a}^{[l]}}\right)$.

Following the circuit compilation in the single-qubit case, our circuit for $\mc{R}$ in the many-qubit case is designed as follows; see \cref{fig:noancilla1}~(a) for an illustration.  
Still, we use qubits of the second replica to check the parity information between the two replicas. 
This is again achieved by first applying a $CX_{a^i\rightarrow b^i}$ for each qubit pair  $(a^i, b^i)$, 
and then measuring the second qubit in each pair in the computational basis. 
If the measurement result is $c^i=0$, then we know that $a^i=b^i$. 
Otherwise, if the measurement result is $c^i=1$, we conclude that $a^i\ne b^i$. Notice that  
the measurement effectively projects (partially decoherence) the qubit pairs in $[l]$ into the following subspace
\begin{align}
\mathcal{V}_{[l]}=\bigoplus_{\mb{a}^{[l]}\in \{0,1\}^l }\mathcal{V}_{\mb{a}^{[l]}}.  
\end{align}
Hence, after the measurement, one can then create the desired superposition on each $\mathcal{V}_{\mb{a}^{[l]}}$ by performing a Hadamard-like  transformation $U_{[l]}'$ on all qubits $\{a^i\}_{i\in[l]}$. 
The action of $U_{[l]}'$ is given as  
\begin{align}\label{eq:U[l]'}
U_{[l]}'\ket{\mb{a}^{[l]}}
= \begin{cases}
\frac{1}{\sqrt{2}}\left(\ket{\mb{a}^{[l]}} + \ket{\overline{\mb{a}^{[l]}}}\right)  
&\quad \mb{a}^{[l]}<\overline{\mb{a}^{[l]}}, 
\\[1ex]
\frac{1}{\sqrt{2}}\left(\ket{\overline{\mb{a}^{[l]}}} -\ket{\mb{a}^{[l]}}\right)   
&\quad \mb{a}^{[l]}>\overline{\mb{a}^{[l]}}.  
\end{cases}
\end{align}
At last, we measure all $n$ qubits of the first replica in the computational basis. 
As shown in \cref{fig:noancilla1}~(a), 
for the $i$-th qubit, we denote the measurement result as $x_1^{i}$. In classical post-processing, we let $x_1^{i}$ and $x_2^{i}=c^i\oplus x_1^{i}$ be the final output of the $i$-th qubit pair.

\begin{figure}[htbp]
\centering
\includegraphics[width=0.75\textwidth]{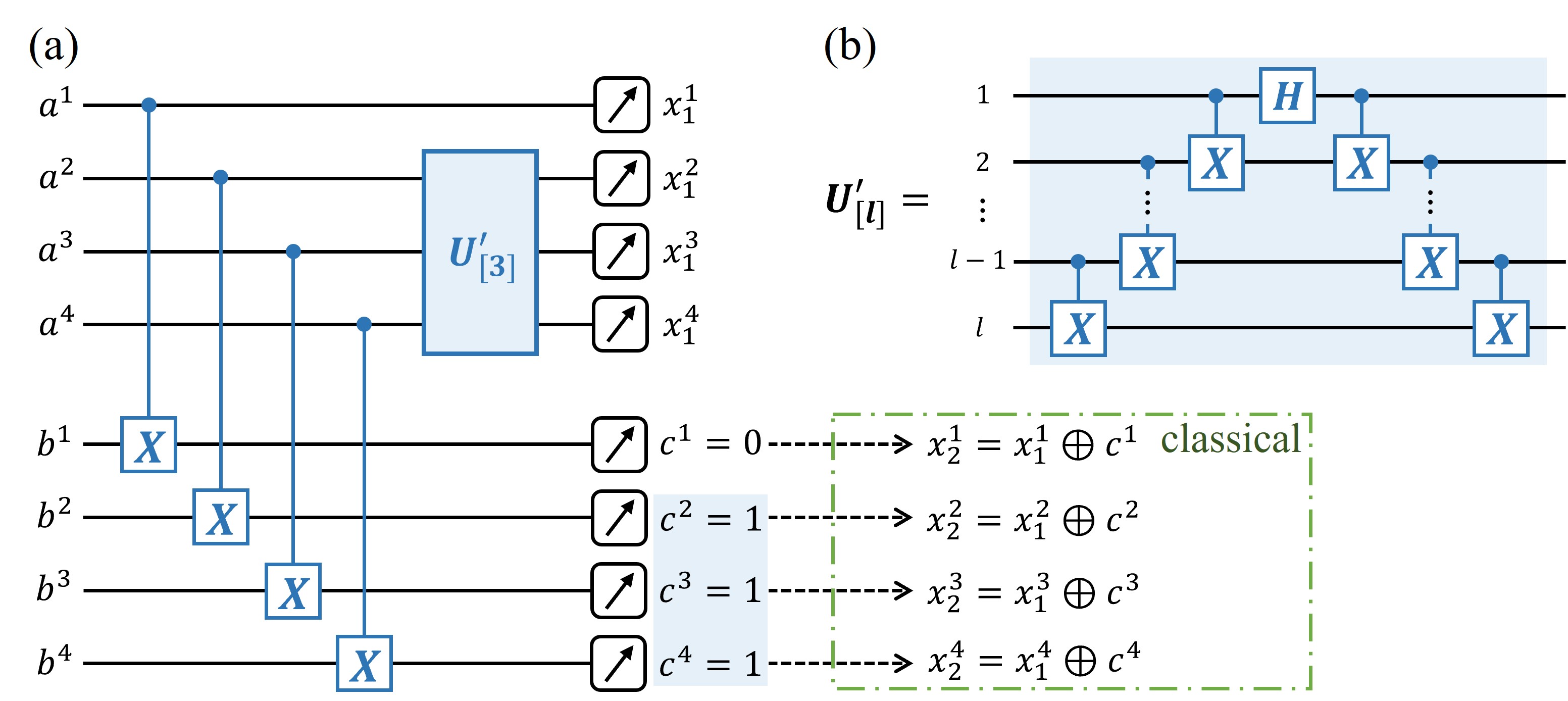}
\caption{\justifying{(a) Circuit compilation of $\mathcal{R}$ in the many-qubit case.  
Here, we take $n=4$, and assume that
the parity check $c^1=0$ and $c^i=1$ for $i\in\{2,3,4\}$. (b) Gate synthesis of $U_{[l]}'$.
}}
\label{fig:noancilla1}
\end{figure}

The remaining task is to decompose the unitary $U_{[l]}'$ into elementary quantum gates. To this end, we consider the subspace ${\rm span}\!\left(\ket{\mb{a}^{[l]}}, \ket{\overline{\mb{a}^{[l]}}}\right)$ as a code space, where $\ket{0_L}=\ket{\mb{a}^{[l]}}$ and $\ket{1_L}=\ket{\overline{\mb{a}^{[l]}}}$ are the logical 0 and 1 states, respectively (here we assume $\mb{a}^{[l]}<\overline{\mb{a}^{[l]}}$ without lose of generality). 
Recall that the condition $\mb{a}^{[l]}<\overline{\mb{a}^{[l]}}$ means that the first bit in the string $\mb{a}^{[l]}$ $\left(\overline{\mb{a}^{[l]}}\right)$ is 0 (1). Using this fact, one can easily verify that the operations $X^{\otimes l}$ and $Z_1$ act as logical $X$ and $Z$ gates on the code space, where $Z_1$ denotes the Pauli-$Z$ operator on the first qubit. According to \cref{eq:U[l]'}, the unitary $U_{[l]}'$ acts as a logical Hadamard operation on the code space. Therefore, it can be chosen as the following form, 
\begin{align}\label{eq:LHadamard}
U_{[l]}'=\frac1{\sqrt{2}}\left(X^{\otimes l}+Z_1\right),  
\end{align}
which is a Clifford operator. Now it suffices to find another Clifford unitary $\tilde U_{[l]}$ such that $\tilde U_{[l]} H_1 \tilde U_{[l]}^{\dag}=U_{[l]}'$, where $H_1=(X_1+Z_1)/\sqrt{2}$ is the Hadamard gate on the first qubit. One can check that the goal can be achieved by constructing $\tilde U_{[l]}$ from a series of $CX$ gates, that is,  
\begin{align}\label{eq:ClUinfo1}
\begin{aligned}
\tilde U_{[l]}=CX_{l-1\rightarrow l}\, CX_{l-2\rightarrow l-1}\cdots CX_{1\rightarrow 2}. 
\end{aligned}
\end{align}
With this construction, the gate synthesis of $U_{[l]}'$ is shown in \cref{fig:noancilla1}~(b).

By plugging the gate synthesis of $U_{[l]}'$ into the circuit of \cref{fig:noancilla1}~(a), we are very close to our final circuit. 
It worth noting that the $CX$ gates after $H_1$ and before the final measurement can be substituted by classical post-processing.
The final circuit after this simplification is shown in 
Fig.~4~(c) of the END MATTER.

\subsubsection{Verification}
To verify that our circuit constructed in the previous section indeed implements the desired unitary $\mc{R}$ and the subsequent measurement, let us consider the evolution of 
the $2n$-qubit joint system. 
The initial state
of the joint system can be expressed as 
$\ket{\Phi_0}=\sum_{\textbf{a},\textbf{b}\in\{0,1\}^n}t_{\textbf{a},\textbf{b}} \ket{\textbf{a}}_1\ket{\textbf{b}}_2$. 
After the $CX_{a_i\rightarrow b_i}$ operations, the joint state evolves into
$\ket{\Phi_1}=\sum_{\textbf{a},\textbf{b}\in\{0,1\}^n}t_{\textbf{a},\textbf{b}} \ket{\textbf{a}}_1\ket{\textbf{a}\oplus\textbf{b}}_2$.  
After measuring the $n$ qubits of the second replica, depending on the measurement outcome $\mb{c}\in\{0,1\}^n$, the reduced state on the first replica reads 
\begin{align}
\ket{\Phi_2^{\mb{c}}}=\frac{(\id\otimes \bra{\mb{c}})\ket{\Phi_1}}{\sqrt{\Pr_2(\mb{c})}}
= \frac{1}{\sqrt{\Pr_2(\mb{c})}} \sum_{\textbf{a},\textbf{b}\in\{0,1\}^n}t_{\textbf{a},\textbf{b}}\, \delta_{\textbf{c},\textbf{a}\oplus\textbf{b}}\ket{\textbf{a}} 
= \frac{1}{\sqrt{\Pr_2(\mb{c})}} \sum_{\textbf{a}\in\{0,1\}^n}t_{\textbf{a},\textbf{a}\oplus\textbf{c}}\ket{\textbf{a}} , 
\end{align}
where
\begin{align}
{\Pr}_2(\mb{c})= \bra{\Phi_1}\left(\id\otimes \ketbra{\mb{c}}\right)\ket{\Phi_1}
= \sum_{\textbf{a}\in\{0,1\}^n}\left|t_{\textbf{a},\textbf{a}\oplus\textbf{c}}\right|^2 
\end{align}
is the probability of obtaining outcome $\mb{c}$.
Next, we consider two cases depending on the outcome $\mb{c}$. 
\begin{itemize}
\item[1.] The outcome $\mb{c}=\mb{0}$.\\
In this case, we then measure $\ket{\Phi_2^{\mb{c}}}$ directly in the computational basis. The conditional probability of obtaining outcome $\mb{x}\in\{0,1\}^n$ is given by
\begin{align}\label{eq:Pr1xcis0}
{\Pr}_1(\mb{x}|\mb{c}=\mb{0})
=\left|\<\mb{x}\!\ket{\Phi_2^\mb{c}} \right|^2
= \left|\frac{1}{\sqrt{\Pr_2(\mb{c})}} \sum_{\textbf{a}\in\{0,1\}^n}t_{\textbf{a},\textbf{a}}\delta_{\mb{a},\mb{x}}\right|^2
= \frac{\left|t_{\mb{x},\mb{x}}\right|^2}{\Pr_2(\mb{c})}. 
\end{align}

\item[2.] The outcome $\mb{c}\ne\mb{0}$.\\
Let $[k]:=\{i\,|\,c^i=0\}$ and $[l]:=\{i\,|\,c^i=1\}$ with $k+l=n$. For any bit string $\textbf{a}\in\{0,1\}^n$, we write $\mb{a}=\mb{a}^{[l]}\mb{a}^{[k]}$. 
In this case, the unitary $U_{[l]}'$ is applied on the qubits of $\ket{\Phi_2^{\mb{c}}}$ in $[l]$, which evolves the state into  
\begin{align}
\ket{\Phi_3^{\mb{c}}}
&=
\left(\id_{[k]}\otimes U_{[l]}'\right)\ket{\Phi_2^{\mb{c}}}
\nonumber\\&= 
\frac{1}{\sqrt{\Pr_2(\mb{c})}} \sum_{\textbf{a}\in\{0,1\}^n}t_{\textbf{a},\textbf{a}\oplus\textbf{c}}
\ket{\textbf{a}^{[k]}}\otimes \left(U_{[l]}'\ket{\textbf{a}^{[l]}}\right)
\nonumber\\&=
\frac{1}{\sqrt{\Pr_2(\mb{c})}} 
\sum_{ \mb{a}^{[k]}\in\{0,1\}^{n-l}}
\ket{\textbf{a}^{[k]}}\otimes
\left[
\sum_{ \mb{a}^{[l]}|\, \mb{a}^{[l]}< \overline{\mb{a}^{[l]} }}
t_{\textbf{a},\textbf{a}\oplus\textbf{c}}\, 
\frac{\ket{\mb{a}^{[l]}} + \ket{\overline{\mb{a}^{[l]}}}}{\sqrt{2}} 
+
\sum_{ \mb{a}^{[l]}|\, \mb{a}^{[l]}> \overline{\mb{a}^{[l]} }}
t_{\textbf{a},\textbf{a}\oplus\textbf{c}}\, 
\frac{\ket{\overline{\mb{a}^{[l]}}} -\ket{\mb{a}^{[l]}}}{\sqrt{2}} 
\right]. 
\end{align}  
If we measure the first replica at this point, then the conditional probability of obtaining outcome $\mb{x}\in\{0,1\}^n$ is given by
\begin{align}\label{eq:Pr1xcne0}
{\Pr}_1(\mb{x}|\mb{c})
= \left|\<\mb{x}\!\ket{\Phi_3^{\mb{c}}} \right|^2
&=\frac{1}{\Pr_2(\mb{c})}\times \left\{
\begin{array}{lll}
\frac{1}{2}
\left| 
t_{\overline{\mb{x}^{[l]}}\mb{x}^{[k]},\mb{x}^{[l]}\mb{x}^{[k]}} 
+ t_{\mb{x}^{[l]}\mb{x}^{[k]},\overline{\mb{x}^{[l]}}\mb{x}^{[k]}}\right|^2
&\  \text{when $\mb{x}^{[l]}< \overline{\mb{x}^{[l]} }$}, 
\\ [1.2ex]
\frac{1}{2}
\left| 
t_{\overline{\mb{x}^{[l]}}\mb{x}^{[k]},\mb{x}^{[l]}\mb{x}^{[k]}} 
-t_{\mb{x}^{[l]}\mb{x}^{[k]},\overline{\mb{x}^{[l]}}\mb{x}^{[k]}} \right|^2
&\  \text{when $\mb{x}^{[l]}> \overline{\mb{x}^{[l]}} $}, 
\end{array}
\right.
\nonumber\\[0.7ex]
&=\frac{1}{\Pr_2(\mb{c})}\times \left\{
\begin{array}{lll}
\frac{1}{2}
\left| t_{\mb{x}\oplus\mb{c},\mb{x}} 
+ t_{\mb{x},\mb{x}\oplus\mb{c}}\right|^2
&\  \text{when $\mb{x}<\mb{x}\oplus\mb{c} $}, 
\\ [0.7ex]
\frac{1}{2}
\left| t_{\mb{x}\oplus\mb{c},\mb{x}} 
- t_{\mb{x},\mb{x}\oplus\mb{c}}\right|^2
&\  \text{when $\mb{x}>\mb{x}\oplus\mb{c}$}, 
\end{array}
\right.
\end{align} 
where we have used the relation $\mb{x}\oplus\mb{c}=\overline{\mb{x}^{[l]}}\mb{x}^{[k]}$. 
\end{itemize}
Recall that the final output result for the second replica is $\mb{x}\oplus\mb{c}$ if the measurement outcome of the first replica is $\mb{x}$.
According to Eqs.~\eqref{eq:Pr1xcis0} and \eqref{eq:Pr1xcne0}, for our
circuit, the overall probability of obtaining outcomes $\mb{x}_1\in\{0,1\}^n$ on the first replica and $\mb{x}_2\in\{0,1\}^n$ on the second replica reads
\begin{align}\label{eq:Prx1x2nqubit} 
\Pr(\mb{x}_1,\mb{x}_2)
&={\Pr}_1(\mb{x}_1|\mb{c}=\mb{x}_1\oplus\mb{x}_2)
\,{\Pr}_2(\mb{c}=\mb{x}_1\oplus\mb{x}_2)
=\left\{
\begin{array}{lll}
\left|t_{\mb{x}_1,\mb{x}_1}\right|^2.   &\  \text{when $\mb{x}_1=\mb{x}_2$}, 
\\ [0.8ex]
\frac{1}{2} \left| t_{\mb{x}_2,\mb{x}_1}+t_{\mb{x}_1,\mb{x}_2} \right|^2 &\  \text{when $\mb{x}_1<\mb{x}_2$} ,
\\ [0.8ex]
\frac{1}{2} \left| t_{\mb{x}_2,\mb{x}_1}-t_{\mb{x}_1,\mb{x}_2}\right|^2 &\  \text{when $\mb{x}_1>\mb{x}_2$} . 
\end{array}
\right.
\end{align}

On the other hand, if the initial state $\ket{\Phi_0}$ is evolved under $\mc{R}$, then  it becomes 
\begin{align}
\ket{\Phi_1'}=\mc{R}\ket{\Phi_0}
= \sum_{\mb{a}\in\{0,1\}^n} t_{\mb{a},\mb{a}} \ket{\mb{a}}_1 \ket{\mb{a}}_2 
+\sum_{\mb{a},\mb{b}|\,\mb{a}<\mb{b}} t_{\mb{a},\mb{b}}\frac{\ket{\mb{b}}_1\ket{\mb{a}}_2+\ket{\mb{a}}_1\ket{\mb{b}}_2}{\sqrt{2}}
+\sum_{\mb{a},\mb{b}|\,\mb{a}>\mb{b}} t_{\mb{a},\mb{b}}\frac{\ket{\mb{b}}_1\ket{\mb{a}}_2-\ket{\mb{a}}_1\ket{\mb{b}}_2}{\sqrt{2}}. 
\end{align}
If we measure this state in the computational basis, then the probability of obtaining outcomes $\mb{x}_1\in\{0,1\}^n$ on the first replica and $\mb{x}_2\in\{0,1\}^n$ on the second replica reads 
\begin{align}
\left|(\bra{\mb{x}_1}\otimes\bra{\mb{x}_2})\!\ket{\Phi_1'} \right|^2
=\left\{
\begin{array}{lll}
\left|t_{\mb{x}_1,\mb{x}_1}\right|^2.   &\  \text{when $\mb{x}_1=\mb{x}_2$}, 
\\ [0.8ex]
\frac{1}{2} \left| t_{\mb{x}_2,\mb{x}_1}+t_{\mb{x}_1,\mb{x}_2} \right|^2 &\  \text{when $\mb{x}_1<\mb{x}_2$} ,
\\ [0.8ex]
\frac{1}{2} \left| t_{\mb{x}_2,\mb{x}_1}-t_{\mb{x}_1,\mb{x}_2}\right|^2 &\  \text{when $\mb{x}_1>\mb{x}_2$} . 
\end{array}
\right.
\end{align}
This probability exactly equals the outcome probability of our circuit shown in \cref{eq:Prx1x2nqubit}. 
Due to linearity, this result also hold when the initial state of the $2n$-qubit system is mixed. 
Therefore, our circuit can indeed realize the desired unitary $\mc{R}$ and the subsequent measurement in the many qubit case. 

\subsection{\red{Single-qubit case for $t=3$}}
\red{
As discussed in the main text, the joint unitary  $\mc{R}$ of Eq.~(9) is defined to simultaneously diagonalize the cyclic permutation $S_t$ and the projectors $O_{\mb{b}}$. 
For the case $t=2$, this reduces to the symmetric/antisymmetric decomposition, which admits a Clifford-only circuit implementation as demonstrated in the previous subsection. 
However, for $t>2$, the eigenvalues of the cyclic shift are complex roots of unity, $\exp(-2\pi\mathrm{i}k/|[\mb{z}]|)$, and thus the exact compilation of $\mc{R}$ may involve non-Clifford rotations. Systematic synthesis of such unitaries is a highly valuable open problem, and we regard it as an important direction for future work.

\begin{figure}[htbp]
\centering
\includegraphics[width=0.55\textwidth]{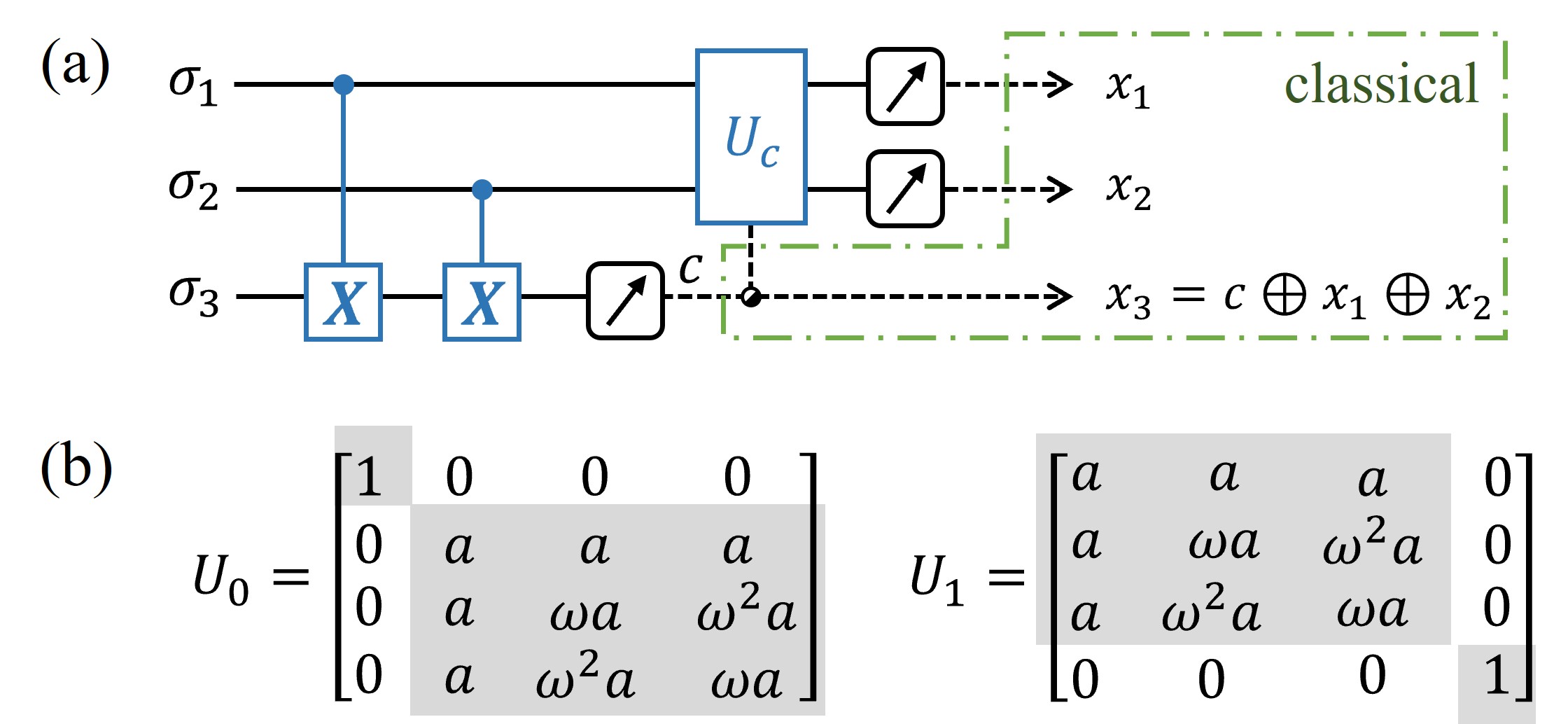}
\caption{\red{\justifying{(a) Circuit implementation of the joint unitary $\mc{R}$ for the single-qubit, $t=3$ case. Two $CX$ gates map the parity of the first two replicas onto the third one, producing the classical outcome $c$, which selects the block unitary $U_c$ acting on the first two replicas. The measurement results $\{x_1, x_2\}$ are obtained directly, while $x_3$ is classically reconstructed as $x_3=c\oplus x_1\oplus x_2$. (b) Matrix form of $U_c$ for $c=0,1$. Each $U_c$ contains a $1\times 1$ identity block and a $3\times 3$ Fourier-transform-like block with non-Clifford phases. Here, $a=1/\sqrt{3}$ and $\omega=\exp(-2\pi\mathrm{i}/3)$. 
}}}
\label{fig:t3}
\end{figure}

As a first step, here we present the compilation for the smallest nontrivial instance: 
the $t=3$ replicas associated with a single qubit of $\rho$.
The explicit circuit is shown in \cref{fig:t3}~(a). 
Two $CX$ gates are applied to the third replica to extract a parity-check-like bit $c$. Depending on the measurement outcome $c$, we apply a corresponding two-qubit unitary $U_c$ on the first two replicas. The explicit forms of $U_0$ and $U_1$, as shown in \cref{fig:t3}~(b), are derived directly from the block structure of the $\mc{R}$ matrix. Each $U_c$ is block diagonal: it contains a $1\times 1$ identity block and a $3\times 3$ Fourier-transform-like block acting on the three elements within a class. This $3\times 3$ block carries the non-Clifford phases $\exp(\pm 2\pi \mathrm{i}/3)$, making explicit where non-Clifford resources first appear in exact compilation.
After applying $U_c$, one measures the first two replicas to obtain outcomes $\{x_1, x_2\}$. The outcome of the third replica is then constructed classically as 
\begin{equation}
    x_3=c\oplus x_1\oplus x_2,
\end{equation}
where $\oplus$ denotes XOR (i.e., a modulo-$2$ addition). Finally, following the mapping strategy $\mc{G}_{\mb{x}\to\mb{b}}$ described in Lemma 1 of the End Matter, one obtains the post-processed string $\mb{b}$ and the correction factor $f(\mb{x})$ for the construction of $\widehat{\rho^3}$ defined in Theorem 1 of the main text.

We note that this parity-check structure is conceptually similar to the $t=2$ case: one first projects (partially decoheres) into the relevant subspace, then applies a unitary within the subspace (Hadamard gate $H$ for $t=2$ and block unitary $U_c$ here).
Although this explicit construction is limited to the single-qubit, $t=3$ case, it illustrates several important features: (\romannumeral1) the emergence of Fourier-type blocks in $\mc{R}$, (\romannumeral2) the potential need for non-Clifford gates beyond $t=2$, and (\romannumeral3) the possibility of combining parity-check-like measurements with conditional block unitaries to reduce implementation overhead. 
Notably, such small-scale constructions can be directly applied with local-AFRS, making them practically useful even for multi-qubit states.
We believe this example provides useful hints and guiding structure for constructing more general realizations of $\mc{R}$ at larger system sizes and higher replica orders.
In addition to explicit construction, larger-$t$ realizations can be pursued along other practical methods. One method is variational quantum compilation \cite{Khatri2019quantumassisted,sim2019expressibility,mitarai2018quantum}, where a hardware-efficient parameterized circuit is optimized to approximate the known matrix form of $\mc{R}$ with controllable approximation error; an alternative is the hybrid strategy \cite{zhou2024hybrid} that assembles unbiased estimators for $\tr(\rho^m)$ from repeated, efficient $t=2$ AFRS estimators (e.g., $\widehat{P_m}=\tr(S_L\widehat{\rho^{m_1}}\otimes\cdots\widehat{\rho^{m_L}})$ with $m=\sum_i m_i$), thereby avoiding heavy non-Clifford compilation at the expense of a milder variance scaling $\mc{O}(d^L)$. In this work we do not further elaborate on these alternative implementations; we leave detailed exploration and platform-specific optimization to future studies.
}


\section{More information about the AFRS protocol}\label{SM:AFRSmore}

\subsection{\red{Computational cost in the classical post-processing stage}}\label{SM:postcost}
\red{
After each experiment of the AFRS protocol, one obtains a measurement outcome string $\mb{x}=\{\mb{x}_1 \mb{x}_2\cdots \mb{x}_t\}\in[d]^t$ ($d=2^n$ is the system dimension). 
In the subsequent classical post-processing stage, one needs to convert $\mb{x}$ to an $n$-bit string $\mb{b}\in[d]$ through the mapping strategy $\mc{G}_{\mb{x}\to\mb{b}}$, and calculate 
the correlation function $f(\mb{x})$.
In this subsection, we show that these operations can be 
efficiently performed on a classical computer with time complexity $\mathcal{O}(nt^2)$, where $n$ is the number of qubits and $t$ is the number of replicas. 
}
\subsubsection{\red{Sampling $\mathbf{b}$ via mapping strategy \texorpdfstring{$\mathcal{G}_{\mathbf{x}\to\mathbf{b}} = \Pr(\mathbf{b}|\mathbf{x})$}{}}}

\red{
Our mapping probability in Observation~1 of the main text is defined as $\Pr(\mathbf{{b}}|\mathbf{x}):= \bra{\psi_{\mathbf{x}}}Q_{\mathbf{b}}\ket{\psi_{\mathbf{x}}}$, where
$\{\ket{\psi_\mb{x}}\}_{\mathbf{x}}$ is a common eigenbasis of $S_t$ and $Q_\mathbf{b}$.  According to Lemma~1 in End Matter, we can choose $\ket{\psi_\mb{x}}$ to be the state $\left|\Psi_{[\mathbf{z}]}^{(k)}\right\rangle$ defined in Eq.~(8) of the End Matter, where $[\mb{z}]$ is the class to which $\mb{x}$ belongs, $k$ is the integer that satisfies $\tau^k(\mb{z})=\mb{x}$, and $\tau$ is the one-step cyclic permutation. 
In this case, the mapping probability can be expressed as 
\begin{equation}\label{eq:r8_8}
\Pr(\mathbf{{b}}|\mathbf{x})= \left\langle\Psi_{[\mathbf{z}]}^{(k)}\right| Q_{\mathbf{b}}\left|\Psi_{[\mathbf{z}]}^{(k)}\right\rangle=\frac{\text{the number of }\mathbf{b}\ \text{in the string}\ \mathbf{x}}{t}=\frac{1}{t}\sum_{i=1}^{t}\delta_{\mathbf{b},\mathbf{x}_i}, 
\end{equation}
where the second equality follows because $\tau^k(\mb{z})=\mb{x}$ and $\ket{\Psi_{[\mb{z}]}^{(k)}}$ is an eigenstate of $Q_{\mb{b}}$ with eigenvalue
$t^{-1}\times$(\textit{the number of $\mathbf{b}$ in the string $\mathbf{z}$}) (see Lemma~1 in End Matter).

Therefore, the procedure (the practical method introduced in End Matter) of obtaining the sampling result $\mathbf{b}\sim\Pr(\mathbf{b}|\mathbf{x})$ is given as follows: 
\begin{itemize}
\item Input the measurement outcome string $\mathbf{x}=\{\mathbf{x}_1 \mathbf{x}_2\dots\mathbf{x}_t\}$.
\item Select an index $i$ from $\{1,2,\dots,t\}$ uniformly at random. 
\item Output $\mathbf{b}=\mathbf{x}_i$.
\end{itemize}
In short, the mapping strategy $\mathcal{G}_{\mb{x}\to \mb{b}}$ simply consists of randomly selecting $\mb{b}$ from $\{\mb{x}_1, \mb{x}_2, \ldots, \mb{x}_t\}$ with a uniform probability of $1/t$ (as stated in the End Matter). This process requires no knowledge of $[\mathbf{z}]$ or $k$.  
It is straightforward to verify that the resulting mapping probability matches the expression given in \cref{eq:r8_8}.  
Since the procedure only involves generating a random integer $i \in \{1,2,\dots,t\}$, its time complexity is $\mathcal{O}(1)$, making it highly efficient.
}

\subsubsection{\red{Computation of the function \texorpdfstring{$f(\mathbf{x})$}{}}} 

\red{
The function $f(\mathbf{x})$ is defined as $f(\mathbf{x}):=\bra{\psi_{\mathbf{x}}}S_t\ket{\psi_{\mathbf{x}}}$, 
where $\ket{\psi_\mb{x}}$ can be chosen as the state $\left|\Psi_{[\mathbf{z}]}^{(k)}\right\rangle$ defined in Eq.~(8) of the End Matter. In this case, we have 
\begin{equation}
f(\mathbf{x})=\left\langle\Psi_{[\mathbf{z}]}^{(k)}\right| S_t\left|\Psi_{[\mathbf{z}]}^{(k)}\right\rangle
=\exp\!\left(-\frac{k\cdot2\pi \mathrm{i}}{|[\mathbf{z}]|}\right). 
\end{equation}
To compute $f(\mathbf{x})$ from a measurement result $\mathbf{x}$, we need to  
\begin{itemize}
\item[(A)] find the cardinality  $|[\mathbf{z}]|$ of the class $[\mb{z}]$ to which $\mb{x}$ belongs; 
\item[(B)] find the integer $k$ that satisfies $\tau^k(\mb{z})=\mb{x}$, where $\mb{z}$ is the smallest element in the class $[\mb{z}]$; and 
\item[(C)] calculate the value of $\exp\!\left(-k\cdot2\pi \mathrm{i}/|[\mathbf{z}]|\right)$. 
\end{itemize}
 Below, we describe the algorithms for solving these tasks and analyze their time complexities.
It turns out that all these tasks can be solved in $\mathcal{O}(nt^2)$ time.

First, we consider task A. Notice that the class cardinality $|[\mathbf{z}]|$ equals the minimal period of $\mathbf{x}$, i.e., the smallest positive integer $r$ such that $\tau^r(\mathbf{x})=\mathbf{x}$.  To compute this minimal period, it suffices to iteratively check whether $\tau^j(\mathbf{x}) = \mathbf{x}$ for $j=1,2,3,\dots$, and return the smallest such integer $j$. The complete procedure is described in Algorithm~\ref{alg:FindCardinality}. 
Next, we analyze the time complexity of Algorithm~\ref{alg:FindCardinality}. Since $\mathbf{x}$ contains $t$ strings $\mathbf{x}_1 \mathbf{x}_2 \dots \mathbf{x}_t$, the number of iterations is at most $t$. Each $\mathbf{x}_i$ is an $n$-bit string, so $\mathbf{x}$ (and $\mathbf{y}$) has $nt$ bits. Hence, in each iteration, shifting $\mathbf{y}$ by one position (Step 4) requires time $\mathcal{O}(nt)$. 
In addition, comparing $\mathbf{y}$ with $\mathbf{x}$ (Step 6) also takes time $\mathcal{O}(nt)$, since we need to check each bit pair one by one. 
Therefore, the time complexity of Algorithm~\ref{alg:FindCardinality} is 
\begin{equation}
T_A=\mathcal{O}(t) \cdot \mathcal{O}(nt)=\mathcal{O}(nt^2).
\end{equation}

\begin{figure}
\begin{algorithm}[H]
{\small
\hspace{-398pt}\textbf{Input:}  $\mathbf{x}=\{\mathbf{x}_1 \mathbf{x}_2\dots\mathbf{x}_t\}$ \\
\hspace{-448pt} \textbf{Output:} $|[\mathbf{z}]|$ 

\begin{algorithmic}[1]
\caption{{\small Algorithm for finding the cardinality $|[\mathbf{z}]|$} \ \ }
\label{alg:FindCardinality}
\State{$\mathbf{y}\leftarrow \mathbf{x}$}
\State{$r\leftarrow 0$} 

\While{True}

\State{$\mathbf{y}\leftarrow \tau(\mathbf{y})$} 
\State{$r\leftarrow r+1$}

\If{$\mathbf{y} = \mathbf{x}$} 
    \State \textbf{break} \Comment{Terminate when the original $\mathbf{x}$ is recovered}
\EndIf

\EndWhile

\State{\textbf{return} $|[\mathbf{z}]|=r$}  \Comment{Output the number of shifts needed}
\end{algorithmic}
}
\end{algorithm}
\end{figure}

\begin{figure}
\begin{algorithm}[H]
{\small
\hspace{-398pt}\textbf{Input:}  $\mathbf{x}=\{\mathbf{x}_1 \mathbf{x}_2\dots\mathbf{x}_t\}$ \\
\hspace{-458pt} \textbf{Output:} $k$ 

\begin{algorithmic}[1]
\caption{{\small Algorithm for finding the integer $k$ that satisfies $\tau^k(\mb{z})=\mb{x}$} \ \ }
\label{alg:Findk}
\State{$\mathbf{y}\leftarrow \mathbf{x}$}
\State{$j\leftarrow 1$}

\While{$j< |[\mathbf{z}]|$}

\State{$j\leftarrow j+1$} 

\If{$\tau(\mathbf{y}) < \mathbf{y}$} \Comment{Compare $\tau(\mathbf{y})$ and $\mathbf{y}$, and store the smaller one}
\State{$\mathbf{y}\leftarrow \tau(\mathbf{y})$} 
\EndIf

\EndWhile

\State{$\mathbf{z}\leftarrow \mathbf{y}$} \Comment{Find the smallest element in the class $[\mb{z}]$ to which $\mb{x}$ belongs}

\State{$r\leftarrow 0$}

\While{True}

\If{$\mathbf{y} = \mathbf{x}$} 
    \State \textbf{break} \Comment{Terminate when $\mathbf{x}$ is recovered}
\EndIf

\State{$\mathbf{y}\leftarrow \tau(\mathbf{y})$}
\State{$r\leftarrow r+1$} 

\EndWhile

\State{\textbf{return} $k=r$}  \Comment{Output the number of shifts needed to recover $\mathbf{x}$ from $\mathbf{z}$}
\end{algorithmic}
}
\end{algorithm}
\end{figure}

Second, we consider task B. 
To solve this task, we need to first find the smallest element $\mathbf{z}$ in the class to which $\mathbf{x}$ belongs. 
This can be completed by virtue of a simple linear search algorithm.  
We start by storing $\mathbf{x}$ in a variable $\mathbf{y}$, then iterate through the remaining elements in the class one by one. For each element, we compare it with the current $\mathbf{y}$; if it is smaller, we update $\mathbf{y}$ to this new string. By the end of the iteration, $\mathbf{y}$ will be the smallest number in the class. This procedure is described in Steps 1--9 of Algorithm~\ref{alg:Findk}, in which  
the number of iterations is $|[\mathbf{z}]|-1<t$. In each iteration, comparing $\tau(\mathbf{y})$ with $\mathbf{y}$ (Step 5) takes time $\mathcal{O}(nt)$. 
So the time cost of Steps 1--9 in Algorithm~\ref{alg:Findk} is 
\begin{equation}
T_B^1=(|[\mathbf{z}]|-1) \cdot \mathcal{O}(nt)=\mathcal{O}(nt^2).
\end{equation}
After obtaining $\mathbf{z}$, one can then calculate the integer $k$ by  iteratively checking whether $\tau^j(\mathbf{z}) = \mathbf{x}$ for $j=0,1,2,\dots$, and return the smallest such integer $j$. This procedure is described in Steps 10--18 of Algorithm~\ref{alg:Findk}, and similarly requires $\mathcal{O}(nt^2)$ time. 
Therefore, the overall time complexity of Algorithm~\ref{alg:Findk} for task B is $T_B=\mathcal{O}(nt^2)$.

Third, task C is trivial and can be solved in $\mathcal{O}(\text{poly log} \; t)$ time using an ordinary calculator, as both $k$ and $|[\mathbf{z}]|$ are $\mathcal{O}(t)$ and do not depend on $n$.

In summary, the computation of $f(\mathbf{x})$ can be done efficiently in $\mathcal{O}(nt^2)$  time. 
}

\subsection{\red{Sample complexity for estimating a collection of observables}}
\red{
Our AFRS protocol inherits the mindset of shadow estimation  that 
the experimental procedure is independent of the specific properties to be predicted.
This feature enables it to simultaneously predict multiple different nonlinear properties of the system without losing high efficiency. 
In this subsection, we analyze  the sample cost of AFRS for estimating $\tr(O_l\rho^t)$ with a collection of $L$ observables $\mc{W}=\{O_l\}_{l=1}^L$.

In each experimental round of our AFRS protocol, we employ $t$ copies of the unknown state $\rho$ to construct an unbiased estimator $\hat{o}_t$ for $\tr(O\rho^t)$. 
After $M$ experiments, we obtain a collection of estimators $\{\hat{o}_{t,(i)}\}_{i=1}^M$. For estimating a single $\tr(O\rho^t)$ with precision $\epsilon$, the empirical mean estimator $\frac{1}{M}\sum_{i=1}^M \hat{o}_{t,(i)}$ can be employed, and Chebyshev's inequality dictates that the required number of experimental repetitions scales as $M\sim\mathrm{Var}(\widehat{o_t})\epsilon^{-2}$.
When the goal is to simultaneously estimate multiple properties $\{\tr(O_l\rho^t)\}_{l=1}^L$, we adopt the median-of-means technique \cite{huang2020predicting,jerrum1986random}. This approach requires $M = NR$ repetitions, where $R\sim\log(L)$ and $N\sim\max_l \mathrm{Var}\Big(\widehat{o_t^{(l)}}\Big)\epsilon^{-2}$.
More concretely, for each observable $O_l \in \mathcal{W}$, we construct $R$ independent sample means of size $N$, and output their median as our estimator for $\tr(O_l\rho^t)$.  
Thus, the sample cost of our AFRS protocol is 
\begin{equation}\label{eq:SamCost}
Mt=NR\,t\sim\frac{t  \log(L)}{\epsilon^2} \cdot \max_l \mathrm{Var} \Big(\widehat{o_t^{(l)}}\Big). 
\end{equation}
As shown in Theorem~1 of the main text, the variance admits the following upper bound:
\begin{equation}\label{eq:VarAFRS}
    \mathrm{Var}(\widehat{o_t})
    \leq \|O_0\|^2_{\mathrm{sh},\mathcal{E}}+\|O\|^2_{\infty},
\end{equation}
where $O_0$ denotes the traceless part of observable $O$, 
and $\|O_0\|_{\mathrm{sh},\mathcal{E}}$ represents the shadow norm of $O_0$, which depends on the unitary ensemble $\mathcal{E}$.
We explicitly characterize the scaling for two practically important ensembles:
\begin{itemize}
\item[a)] For the local Clifford ensemble $\mathcal{E}_{\rm LCl}$, 
Eq.~(S17) in Ref.~\cite{huang2020predicting} establishes:
\begin{equation}
    \|O_0\|^2_{\mathrm{sh, LCl}}\leq 4^{\text{locality}(O)} \|O\|_\infty^2, 
\end{equation}
where $\text{locality}(O)$ counts the number of qubits on which $O$ acts nontrivially. 
Combining this with Eqs.~\eqref{eq:SamCost} and \eqref{eq:VarAFRS} yields the sample complexity
\begin{equation}
Mt\sim \frac{t  \log(L)}{\epsilon^2} \cdot 
\max_l \left[ \left(4^{\text{locality}(O_l)}+1\right) \|O_l\|_\infty^2 \right].
\end{equation} 
The scaling is linear in $t$ and independent of the total qubit number $n$.

\item[b)] For the Global Clifford ensemble $\mathcal{E}_{\rm Cl}$, Eq.~(S16) in Ref.~\cite{huang2020predicting} gives:
\begin{equation}
    \|O_0\|^2_{\mathrm{sh, Cl}}\leq 3\,\|O_0\|_2^2.
\end{equation}
Combining this with Eqs.~\eqref{eq:SamCost} and \eqref{eq:VarAFRS} leads to the sample complexity of AFRS:
\begin{equation}
Mt\sim \frac{t  \log(L)}{\epsilon^2} \cdot 
\max_l \Big( 3\,\|(O_l)_0\|_2^2 + \|O_l\|_\infty^2 \Big),
\end{equation}
which similarly exhibits linear $t$-dependence and no explicit scaling with $n$.

\end{itemize}


When the observables $\{O_l\}_{l=1}^L$ of interest have small localities, one can further employ the local-AFRS protocol to significantly reduce the circuit depth, thereby improving practicality.  
According to Corollary 1 in the main text, to estimate all $\tr(O_l\rho^t)$ with error at most $\epsilon$, it suffices to run the local-AFRS circuit the following number of times:  
\begin{equation}  
M \sim \frac{K \log(L)}{\epsilon^2} \max_l \mathrm{Var}\Big(\widehat{o^{(l)}_t}\Big),  
\end{equation}  
where $K \leq L$ is the number of subsets obtained by partitioning $\{O_l\}_{l=1}^L$ [see Eq.~(6) in the main text].  
In local-AFRS, we adopt the local Clifford ensemble $\mathcal{E}_{\rm LCl}$, so the variance of the estimator satisfies  
\begin{equation}  
\mathrm{Var}(\widehat{o_t}) \leq \left\|O_{0}\right\|_{\mathrm{sh},\mathcal{E}}^2 + \|O\|_{\infty}^2 \leq \left(4^{\text{locality}(O)} + 1\right) \|O\|_\infty^2.  
\end{equation}  
Moreover, each round of experiments utilizes $t$ copies of the unknown state $\rho$. Combining these observations, we obtain the following sample complexity for the local-AFRS protocol:  
\begin{equation}  
Mt \sim \frac{t K  \log(L)}{\epsilon^2} \max_l \left[ \left(4^{\text{locality}(O_l)} + 1\right) \|O_l\|_\infty^2 \right]. 
\end{equation}  
Again, the scaling is still linear in $t$ and independent of $n$.
In particular, if $L=1$ and the observable of interest is the identity operator $O = \mathbb{I}$ (i.e., we aim to estimate the moment $P_t = \tr(\rho^t)$), then we have $K=1$, $\text{locality}(O)=0$, and the sample complexity reduces to  $Mt \sim t\epsilon^{-2}$. 

In summary, for estimating a collection of $L$ nonlinear properties 
$\{\tr(O_l\rho^t)\}_{l=1}^L$, the sample complexities  of our AFRS and local-AFRS protocols both scale linearly in $t$ and logarithmically in $L$; they have no explicit dependence on the qubit number $n$, but depend on the observable structure, i.e., the locality when using the local Clifford ensemble $\mathcal{E}_{\rm LCl}$, and the Frobenius norm when using the global Clifford ensemble $\mathcal{E}_{\rm Cl}$. 
All these scaling behaviors are the best that one can expect. 
Compared to the original shadow protocol \cite{huang2020predicting} for estimating linear properties $\{\tr(O_l\rho)\}_{l=1}^L$, our sample costs only increase by a factor of $\mathcal{O}(t)$. In fact, as shown by recent works Refs.~\cite{chen2025simultaneous,ZhangWuZhou25}, this linear dependence in $t$ is inevitable, even when estimating a single nonlinear property $\tr(O\rho^t)$. Therefore, our AFRS and local-AFRS protocol achieve the optimal scaling behavior in $t$.  
}

\section{Local-AFRS protocol for local observables}
In this section, we present more details of the local-AFRS protocol mentioned in the main text. 
In Supplementary Note~\ref{SM:ob2}, we introduce the key observation  for constructing the protocol. 
In Supplementary Note~\ref{SM:Th2}, we present the detailed form of the estimator used in local-AFRS and prove the performance shown in Theorem 2 of the main text. 
In Supplementary Note~\ref{SM:Th2part3}, 
we analyze the efficiency of applying local-AFRS to estimate the state moment.

\subsection{Key observation behind local-AFRS}\label{SM:ob2} 
Consider the observable $O$ of interest being local and acting nontrivially on the subsystem $A$, say, $O=O_{A}\otimes \id_{\bar{A}}$. 
With respect to the locality property, the random unitary used in local-AFRS can be taken as the product form $V=V_A\otimes V_{\bar{A}}$.
Similar to the framework of AFRS introduced in the main text, here one needs to effectively generate the \emph{pseudo} but now \emph{marginal} conditional probability 
\begin{equation}
\Pr(\mb{b}^{A}|V_{A})=\bra{\mb{b}^{A}} V_{A} (\tr_{\bar{A}}(\rho^t)) \ket{\mb{b}^{A}},
\end{equation}
where $\mb{b}^{A}$ labels the bit-string restricted on $A$. Actually, Observation 1 in the main text
can be applied to the local case here, but with the observable $Q_{\mb{b}}$ 
updated to 
\begin{equation}\label{eq:QbA}
Q^{\mathrm{local}} (\mb{b}^{A})=\left[t^{-1}\sum_{i} \ket{\mb{b}^{A}}\bra{\mb{b}^{A}}_i\otimes \id^{\otimes (t-1) }_{A} \right] \otimes \id^{\otimes t }_{\bar{A}}=Q_{\mb{b}^{A}} \otimes \id^{\otimes t }_{\bar{A}}.
\end{equation}
Here, $Q_{\mb{b}^{A}}$ is symmetric on all $t$ replicas, but restricted on $A$; the complement subsystem $\bar{A}$ takes the identity $\id_{\bar{A}}$ for all $t$ replicas.

Recall that the key step to develop AFRS is to find the common eigenbasis of both $S_t$ and $Q_\mathbf{b}$. 
Now the updated $Q^{\mathrm{local}}(\mb{b}^{A})$ 
is in the tensor-product form. Also, $S_t$ can be factorized into $S_t=S_t^{A}\otimes S_t^{\bar{A}}$, and even further into the qubit level operators on $t$ replicas. 
In this way, for any $\mb{x}^{A}$ and $\mb{x}^{\bar{A}}$, the state 
\begin{equation}
\mc{R}^{\dag}_A\otimes \mc{R}^{\dag}_{\bar{A}}\ket{\mb{x}^{A},\mb{x}^{\bar{A}}}=\ket{\psi_{\mb{x}^{A}}}\ket{\psi_{\mb{x}^{\bar{A}}}}
\end{equation}
is a common eigenstate for $Q^{\mathrm{local}}(\mb{b}^{A})$ and $S_t$, with $\mb{x}^{A}=\{\mb{x}_1^{A}\mb{x}_2^{A}\cdots\mb{x}_t^{A}\}$ \big($\mb{x}^{\bar{A}}=\{\mb{x}_1^{\bar{A}}\mb{x}_2^{\bar{A}}\cdots\mb{x}_t^{\bar{A}}\}$\big) being the bit-strings restricted on subsystem $A$ ($\bar{A}$) for all replicas. 
Accordingly, for local-AFRS here, the function $f(\mb{x})$, mapping strategy $\mathcal{G}_{\mb{x}\to \mb{b}}$, and the quantum probability $\Pr(\mb{x}|V)$ in Observation 1 are updated to 
\begin{align}\label{eq:mapLocal}
f(\mb{x}^{A},\mb{x}^{\bar{A}})=
\bra{\phi_{\mb{x}}} S_t^{A}\otimes S_t^{\bar{A}} \ket{\phi_{\mb{x}}},
\quad
\mathcal{G}_{\mb{x}^A\to \mb{b}^A}:\Pr(\mb{b}^{A}|\mb{x}^{A})=\bra{\psi_{\mb{x}^{A}}} Q_{\mb{b}^{A}}\ket{\psi_{\mb{x}^{A}}}, 
\quad
\Pr(\mb{x}|V)= \bra{\phi_{\mb{x}}}V(\rho)^{\otimes t} \ket{\phi_{\mb{x}}},
\end{align}
respectively, where we denote by $\ket{\phi_{\mb{x}}}:=\ket{\psi_{\mb{x}^{A}}}\otimes\ket{\psi_{\mb{x}^{\bar{A}}}}$ for simplicity, with $\mb{x}=\mb{x}^{A}\mb{x}^{\bar{A}}$.

The following observation clarifies the relation between $\Pr(\mb{b}^A|V_A)$, $f(\mb{x}^A,\mb{x}^{\bar{A}})$,  $\Pr(\mb{b}^A|\mb{x}^A)$, and $\Pr(\mb{x}|V)$. It is the counterpart of 
Observation 1 in the main text.

\begin{observation}\label{obs:local}
The pseudo and marginal conditional probability $\Pr(\mb{b}^A|V_A)$ can be written as
\begin{equation}\label{eq:local_conprob}
\begin{aligned}
\Pr(\mb{b}^A|V_A) &= \tr\left[ S_t \left(Q_{\mb{b}^A}\otimes \identity_{\bar{A}}^{\otimes t}\right) V(\rho)^{\otimes t} \right]\\
&= \sum_{\mb{x}} \bra{\phi_{\mb{x}}}S_t\ket{\phi_{\mb{x}}} \bra{\phi_{\mb{x}}}\left(Q_{\mb{b}^A}\otimes \id_{\bar{A}}^{\otimes t}\right)\ket{\phi_{\mb{x}}} \bra{\phi_{\mb{x}}}V(\rho)^{\otimes t}\ket{\phi_{\mb{x}}}\\
&=\sum_{\mb{x}} f(\mb{x}^A,\mb{x}^{\bar{A}}) \Pr(\mb{b}^A|\mb{x}^A) \Pr(\mb{x}|V).
\end{aligned}
\end{equation}
\end{observation}

\begin{proof}[Proof of Observation~\ref{obs:local}]
Recall that $\Pr(\mb{b}^A|V_A)= \bra{\mb{b}^A}V_A(\tr_{\bar{A}}(\rho^t)) \ket{\mb{b}^A}$ and $Q_{\mb{b}^A}=t^{-1}\sum_{i} \ketbra{\mb{b}^{A}}_i\otimes \id^{\otimes (t-1) }_{A}$. 
So the first line of \cref{eq:local_conprob} can be proved by showing that
\begin{equation}\label{eq:obs2_proof}
    \bra{\mb{b}^A}V_A(\tr_{\bar{A}}(\rho^t)) \ket{\mb{b}^A} = \tr\left( S_t \left[ \left(\ketbra{\mb{b}^A}_i \otimes \identity_{A}^{\otimes (t-1)} \right)\otimes \identity_{\bar{A}}^{\otimes t} \right]V(\rho)^{\otimes t} \right)\quad \forall i=1,2,\cdots,t.
\end{equation}
And we graphically show the derivation of this relation in \cref{fig:obs2_sketch}.

\begin{figure}[t]
\centering
\includegraphics[width=0.8\textwidth]{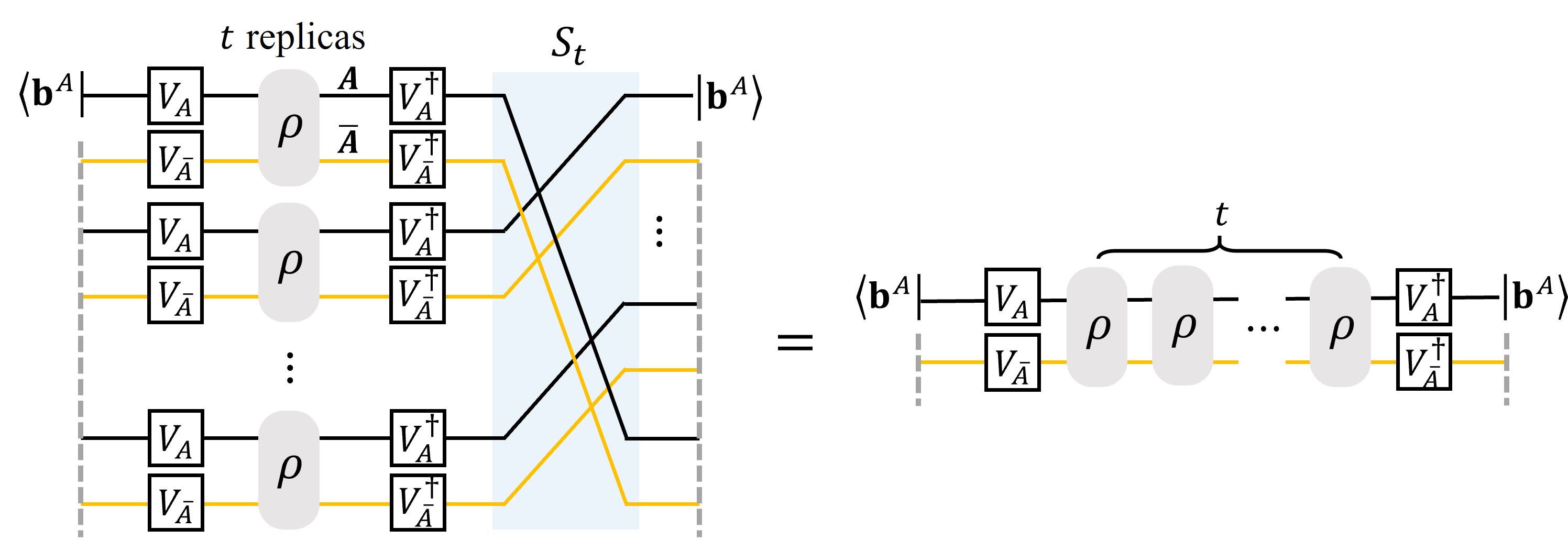}\\
\caption{\justifying{A proof sketch of \cref{eq:obs2_proof} for local-AFRS protocol. Here, the yellow lines represent the subsystem $\bar{A}$ of $\rho$, and the gray dashed lines denote the trace operation. It indicates in the right sketch that $V_{\bar{A}}$ is not essential as $V_{\bar{A}} V_{\bar{A}}^\dag = \identity_{\bar{A}}$.
Thus, in practice, we can set $V_{\bar{A}}= \identity_{\bar{A}}$ for simplicity.
}}
\label{fig:obs2_sketch}
\end{figure}

Recall that $S_t^A$ and $Q_{\mb{b}^A}$ can be diagonalized simultaneously under the basis $\left\{\ket{\psi_{\mb{x}^{A}}}\right\}$, and $S_t^{\bar{A}}$ can be diagonalized in the basis $\left\{\ket{\psi_{\mb{x}^{\bar{A}}}}\right\}$. Hence, the operators $\ketbra{\phi_{\mb{x}}}=\ketbra{\psi_{\mb{x}^{A}}}\otimes\ketbra{\psi_{\mb{x}^{\bar{A}}}}$, $S_t=S_t^A\otimes S_t^{\bar{A}}$, and $Q_{\mb{b}^A}\otimes \id_{\bar{A}}^{\otimes t}$ commute with each other. 
Using this fact, the second line of \cref{eq:local_conprob} can be derived as
\begin{align}\label{eq:localtr_sum}
    \tr\left[ S_t \left(Q_{\mb{b}^A}\otimes \identity_{\bar{A}}^{\otimes t}\right) V(\rho)^{\otimes t} \right] &= \sum_{\mb{x}} \bra{\phi_{\mb{x}}} S_t \left(Q_{\mb{b}^A}\otimes\identity_{\bar{A}}^{\otimes t}\right) V(\rho)^{\otimes t} \ket{\phi_{\mb{x}}} \nonumber\\
    & = \sum_{\mb{x}} \bra{\phi_{\mb{x}}}S_t\ket{\phi_{\mb{x}}} \bra{\phi_{\mb{x}}}\left(Q_{\mb{b}^A}\otimes \id_{\bar{A}}^{\otimes t}\right)\ket{\phi_{\mb{x}}} \bra{\phi_{\mb{x}}}V(\rho)^{\otimes t}\ket{\phi_{\mb{x}}}.
\end{align}

Finally, the third line of \cref{eq:local_conprob} holds because
\begin{equation}
\begin{aligned}
    \bra{\phi_{\mb{x}}}\left(Q_{\mb{b}^A}\otimes \id_{\bar{A}}^{\otimes t}\right)\ket{\phi_{\mb{x}}} = \bra{\psi_{\mb{x}^A}}Q_{\mb{b}^A}\ket{\psi_{\mb{x}^A}} \bra{\psi_{\mb{x}^{\bar{A}}}} \ket{\psi_{\mb{x}^{\bar{A}}}} = \Pr(\mb{b}^A|\mb{x}^A).
\end{aligned}
\end{equation}
This completes the proof.
\end{proof}

We remark that, as mentioned in the main text, one can set $V_{\bar{A}}=\id_{\bar{A}}$ here, which does not change the relation in \cref{eq:obs2_proof}, and thus the final conclusion. This is also demonstrated by \cref{fig:obs2_sketch}. 
Moreover, the used eigenbasis $\ket{\phi_{\mb{x}}}:=\ket{\psi_{\mb{x}^{A}}}\otimes\ket{\psi_{\mb{x}^{\bar{A}}}}$ can also be simplified. To get the result in \cref{eq:localtr_sum}, one just need to find a basis on $\bar{A}$ of $t$ replicas which is the eigenbasis of $S_t^{\bar{A}}$. As mentioned in the main text, one thus can choose $\ket{\phi_{\mb{x}}}=\ket{\psi_{\mb{x}^{A}}}\otimes \bigotimes  _{i\in \bar{A}}\ket{\psi_{\mb{x}^{i}}}$, which can be realized by $\mc{R}_A\otimes \bigotimes _{i\in \bar{A}}\mc{R}_{(i)}$. 
In summary, \cref{obs:local} holds if one takes $V_{\bar{A}}=\id_{\bar{A}}$, and $\ket{\phi_{\mb{x}}}=\ket{\psi_{\mb{x}^{A}}}\otimes \bigotimes  _{i\in \bar{A}}\ket{\psi_{\mb{x}^{i}}}$, with the function $f$ updated to
\begin{align}
f(\mb{x}^{A},\mb{x}^{\bar{A}})
&\rightarrow
\bigg(\bra{\psi_{\mb{x}^{A}}}\otimes \bigotimes  _{i\in \bar{A}}\bra{\psi_{\mb{x}^{i}}}\bigg)
\bigg(S_t^{A}\otimes \bigotimes_{i\in \bar{A}}   S_t^{i}\bigg) 
\bigg(\ket{\psi_{\mb{x}^{A}}}\otimes \bigotimes _{i\in \bar{A}}\ket{\psi_{\mb{x}^{i}}}\bigg) 
\nonumber\\
&\qquad =f(\mb{x}^{A})\prod_{i\in \bar{A}}f(\mb{x}^{i}).
\end{align}
With such changes, the results of the estimator and the variance given  in Supplementary Note~\ref{SM:Th2} also hold.

\subsection{The estimator and performance of local-AFRS}
\label{SM:Th2}
The experimental procedure of local-AFRS is described in the main text. 
Let $V_A$ be the random unitary performed on subsystem $A$ of each replica, and $\mb{x}=\{\mb{x}^{A},\mb{x}^{\bar{A}}\}$ represent the actual measurement outcome [see Fig.~3~(a) of the main text]. 
In the postprocessing stage, we first map $\mb{x}^{A}$ to a string $\mb{b}^{A}$ with probability $\Pr(\mb{b}^{A}|\mb{x}^{A})$, then use 
$\mb{b}^{A}$, $\mb{x}^{A}$, $\mb{x}^{\bar{A}}$, and $V_A$ to construct the estimator $\widehat{\rho^t_A}$.
The detailed form of $\widehat{\rho^t_A}$ and the performance of local-AFRS are shown in \cref{th:localAFRS} below, which is the formal version of Theorem~2 in the main text.

\begin{theorem}\label{th:localAFRS}
    For a single implementation of the circuit in Fig.~3~(a) (single-shot, $M=1$). The unbiased estimator of $\rho^t$ reduced on subsystem $A$ is constructed as
\begin{equation}\label{eq:unA-Ptr}
\begin{aligned}
\widehat{\rho^t_A}=\mathrm{Re}\left(f(\mb{x}^{A},\mb{x}^{\bar{A}})\right)\mathcal{M}_A^{-1}\left( 
V_A^{\dag} \ketbra{\mb{b}^A}{\mb{b}^A}V_A\right).
\end{aligned}
\end{equation}
Here, the inverse channel $\mathcal{M}_A^{-1}$ on $A$ depends on the random unitary ensemble $\mc{E}^{(A)}$, the function $f$ and the mapping result $\mb{b}^A$ are given in \cref{eq:mapLocal}.

The local observable $O=O_A\otimes \id_{\bar{A}}$ on $\rho^t$ is estimated by $\widehat{o_t}=\tr\!\left(\widehat{\rho^t_A}\,O_A\right)$,
and its variance shows 
\begin{equation}\label{eq:varmLocal}
\begin{aligned}
\mathrm{Var}(\widehat{o_t}) \leq \left\|O_{A,0}\right\|_{\mathrm{sh},\mc{E}^{(A)}}^2 + \|O\|_{\infty} ^2,
\end{aligned}
\end{equation}
where 
$O_{A,0}$ is the traceless part of $O_A$; and 
$\|\cdot\|_{\mathrm{sh},\mc{E}^{(A)}}$ is the shadow norm depending on the unitary ensemble $\mc{E}^{(A)}$. 
\end{theorem}

The proof of \cref{th:localAFRS} is separated into two parts. 
Using \cref{obs:local}, the first part in Supplementary Note~\ref{SM:Th2part1} shows that $\widehat{\rho^t_A}$ defined in 
\cref{eq:unA-Ptr} reproduces $\tr_{\bar{A}}(\rho^t)$ exactly in expectation (over the randomness in the choice of unitary $V=V_A\otimes V_{\bar{A}}$, measurement outcome $\mb{x}=\{\mb{x}^A, \mb{x}^{\bar{A}}\}$, and the mapping result $\mb{b}^A$ from $\mb{x}^A$).
The second part in Supplementary Note~\ref{SM:Th2part2} proves the bound of $\mathrm{Var}(\widehat{o_t})$ in \cref{eq:varmLocal}.

\subsubsection{Unbiasedness of the estimator \texorpdfstring{$\widehat{\rho^t_A}$}{}}\label{SM:Th2part1}

The expectation of the estimator $\widehat{\rho^t_A}$ is given as 
\begin{equation}\label{eq:localunbiased}
\begin{aligned}
\mathbb{E}_{\{ V,\mb{x},\mb{b}^A\}}\widehat{\rho^t_A}
&= 
\mathbb{E}_{\{ V,\mb{x},\mb{b}^A\}}
{\mathrm{Re}} \left(f(\mb{x}^A,\mb{x}^{\bar{A}}) \right) \mathcal{M}_A^{-1}\left(V_A^{\dag}\ketbra{\mb{b}^A}{\mb{b}^A}V_A \right)\\
&= 
\mathbb{E}_{V \sim \mc{E}}\sum_{\mb{x},\mb{b}^A} \Pr(\mb{x},\mb{b}^A|V) 
{\mathrm{Re}} \left( f(\mb{x}^A,\mb{x}^{\bar{A}}) \right)
\mathcal{M}_A^{-1}\left(V_A^{\dag}\ketbra{\mb{b}^A}{\mb{b}^A}V_A \right)   \\
&=
\mathbb{E}_{V \sim \mc{E}}\sum_{\mb{b}^A} \mathcal{M}_A^{-1}\left(V_A^{\dag}\ketbra{\mb{b}^A}{\mb{b}^A}V_A \right) 
\left[\sum_{\mb{x}} \Pr(\mb{x}|V) \Pr(\mb{b}^A|\mb{x}^A)
{\mathrm{Re}} \left( f(\mb{x}^A,\mb{x}^{\bar{A}}) \right)\right] ,
\end{aligned}
\end{equation}
where we used that $\Pr(\mb{x},\mb{b}^A|V)=\Pr(\mb{x}|V) \Pr(\mb{b}^A|\mb{x}^A)$. 
Since $\Pr(\mb{x}|V)$ and $\Pr(\mb{b}^A|\mb{x}^A)$ are real numbers, the last term in the parentheses can be 
further written as 
\begin{align}
{\mathrm{Re}} \left[\sum_{\mb{x}} \Pr(\mb{x}|V) \Pr(\mb{b}^A|\mb{x}^A)
f(\mb{x}^A,\mb{x}^{\bar{A}}) \right]
= {\mathrm{Re}}\left[\Pr(\mb{b}^A|V_A)\right]
= \bra{\mb{b}^A}V_A\left( 
\tr_{\bar{A}}(\rho^t) \right) \ket{\mb{b}^A}, 
\end{align}
where the first equality follows from 
\cref{obs:local}. 
By plugging this relation into \cref{eq:localunbiased}, we get 
\begin{equation}
\begin{aligned}
\mathbb{E}_{\{ V,\mb{x},\mb{b}^A\}}\widehat{\rho^t_A}
&= 
\mathbb{E}_{V \sim \mc{E}}\sum_{\mb{b}^A} \mathcal{M}_A^{-1}\left(V_A^{\dag}\ketbra{\mb{b}^A}{\mb{b}^A}V_A \right) 
\bra{\mb{b}^A}V_A\left( 
\tr_{\bar{A}}(\rho^t) \right) \ket{\mb{b}^A}
\\
&= \mathcal{M}_A^{-1}\left[ \mathbb{E}_{V_A \sim \mc{E}^{(A)}} \sum_{\mb{b}^A} \bra{\mb{b}^A}V_A\left( 
\tr_{\bar{A}}(\rho^t) \right) \ket{\mb{b}^A}  V_A^{\dag}\ketbra{\mb{b}^A}{\mb{b}^A}V_A \right]\\
    &=\mathcal{M}_A^{-1}\left[ \mathcal{M}_A\left(\tr_{\bar{A}}(\rho^t)\right)\right]=\tr_{\bar{A}}(\rho^t).
\end{aligned}
\end{equation}
In conclusion, $\widehat{\rho^t_A}$ is an unbiased estimator of $\tr_{\bar{A}}(\rho^t)$.

\subsubsection{Proof of the estimation variance in \texorpdfstring{\cref{eq:varmLocal}}{Eq.(S53)}}\label{SM:Th2part2}

The variance of $\widehat{o_t}$ reads 
\begin{equation}\label{eq:partial_otv}
\mathrm{Var}(\widehat{o_t}) 
= \mathbb{E}_{\{V,\mb{x},\mb{b}^A\}}[\widehat{o_t}^2]-\left[\mathbb{E}_{\{V,\mb{x},\mb{b}^A\}}\widehat{o_t}\right]^2=\mathbb{E}_{\{V,\mb{x},\mb{b}^A\}}[\widehat{o_t}^2]-\left[\tr(O_A \rho^t)\right]^2.
\end{equation}
By the self-adjoint property of $\mathcal{M}_A^{-1}$, 
the estimator $\widehat{o_t}$ can be rewritten as 
\begin{align}
\widehat{o_t}
=
\tr\!\left[O_A\widehat{\rho^t_A}\right]
&=
{\mathrm{Re}} \left(f(\mb{x}^A,\mb{x}^{\bar{A}}) \right) \tr\left[ O_A\mathcal{M}_A^{-1}\!\left(V_A^\dag \ketbra{\mb{b}^A}{\mb{b}^A}V_A\right)\right]
\nonumber\\
&={\mathrm{Re}} \left(\bra{\phi_\mb{x}}S_t\ket{\phi_\mb{x}} \right)\bra{\mb{b}^A}V_A \mathcal{M}_A^{-1}(O_A)V_A^\dag \ket{\mb{b}^A}.
\end{align}
It follows that
\begin{equation}\label{eq:ot_local}
\begin{aligned}
&\mathbb{E}_{\{V,\mb{x},\mb{b}^A\}}\left[\widehat{o_t}^2\right]
\\
&= 
\mathbb{E}_{V\sim \mathcal{E}}\sum_{\mb{x},\mb{b}^A} \Pr(\mb{x}|V) \Pr(\mb{b}^A|\mb{x}^A) \,\widehat{o_t}^2\\
&\stackrel{(a)}{=} 
\mathbb{E}_{V\sim \mathcal{E}} \sum_{\mb{x},\mb{b}^A} \bra{\phi_{\mb{x}}} V(\rho)^{\otimes t}\ket{\phi_{\mb{x}}} \bra{\psi_{\mb{x}^{A}}} Q_{\mb{b}^{A}}\ket{\psi_{\mb{x}^{A}}} \left[{\mathrm{Re}} \left(\bra{\phi_\mb{x}}S_t\ket{\phi_\mb{x}} \right)\right]^2 \bra{\mb{b}^A}V_A \mathcal{M}_A^{-1}(O_A)V_A^\dag \ket{\mb{b}^A}^2\\
&\stackrel{(a)}{\leq} 
\mathbb{E}_{V\sim \mathcal{E}} \sum_{\mb{b}^A} 
\bra{\mb{b}^A}V_A \mathcal{M}_A^{-1}(O_A)V_A^\dag \ket{\mb{b}^A}^2
\underbrace{\left(\sum_{\mb{x}}
\bra{\phi_{\mb{x}}} V(\rho)^{\otimes t}\ket{\phi_{\mb{x}}} \bra{\psi_{\mb{x}^{A}}} Q_{\mb{b}^{A}}\ket{\psi_{\mb{x}^{A}}} \right)}_{(*)}, 
\end{aligned}
\end{equation}
where $(a)$ holds  because $\Pr(\mb{x}|V)=\bra{\phi_\mb{x}}V(\rho)^{\otimes t}\ket{\phi_\mb{x}}$ and $\Pr(\mb{b}^A|\mb{x}^A)=\bra{\phi_{\mb{x}}}
Q_{\mb{b}^A}\otimes \id_{\bar{A}}^{\otimes t}\ket{\phi_{\mb{x}}}$, and $(b)$ holds  because $\left[\mathrm{Re}(f(\mb{x}))\right]^2=\left[\mathrm{Re}(\bra{\psi_\mb{x}}S_t\ket{\psi_\mb{x}})\right]^2\leq1$.
The term $(*)$ can be further written as  
\begin{equation}\label{eq:B10UB}
\begin{aligned}
(*)
&\stackrel{(a)}{=}  \sum_{\mb{x}}
\bra{\phi_{\mb{x}}} V(\rho)^{\otimes t}\left(Q_{\mb{b}^A}\otimes \id_{\bar{A}}^{\otimes t}\right)\ket{\phi_\mb{x}}
\stackrel{(b)}{=} 
\tr\left[V(\rho)^{\otimes t}\left(Q_{\mb{b}^A}\otimes \id_{\bar{A}}^{\otimes t}\right)\right]
=
\bra{\mb{b}^A}V_A(\tr_{\bar{A}}\rho)\ket{\mb{b}^A}, 
\end{aligned}
\end{equation}
where $(a)$ holds because $Q_{\mb{b}^A}$ commutes with $\ketbra{\phi_\mb{x}}{\phi_\mb{x}}$, $(b)$ holds because $\{\ket{\phi_\mb{x}}\}_{\mb{x}}$ forms an orthonormal basis on $\mc{H}_d^{\otimes t}$.

Equations \eqref{eq:ot_local} and \eqref{eq:B10UB}  together imply that 
\begin{equation}\label{eq:local_upper}
\mathbb{E}_{\{V,\mb{x},\mb{b}^A\}}\left[\widehat{o_t}^2\right]
\leq 
\mathbb{E}_{V_A\sim \mc{E}^{(A)}} \sum_{\mb{b}^A} 
\bra{\mb{b}^A}V_A \mathcal{M}_A^{-1}(O_A)V_A^\dag \ket{\mb{b}^A}^2
\bra{\mb{b}^A}V_A(\tr_{\bar{A}}\rho)\ket{\mb{b}^A}
= \mathbb{E}\left[ \tr(O_A 
\widehat{\rho_{\bar{A}}})^2 \right], 
\end{equation}
where $\widehat{\rho_{\bar{A}}}$ is the original single-copy shadow snapshot \cite{huang2020predicting}, satisfying $ \mathbb{E}\left[\widehat{\rho_{\bar{A}}}\right]=\tr_{\bar{A}}\rho$.  
Hence, the expectation value of $\widehat{o_t}^2$ is not larger than that of $\left[\tr(O_A \widehat{\rho_{\bar{A}}})\right]^2$ in the original shadow protocol on a single-copy of the state $\tr_{\bar{A}}\rho$ using the unitary ensemble 
$\mc{E}^{(A)}$, no matter what value $t$ takes.

By substituting \cref{eq:local_upper} into \cref{eq:partial_otv}, we have
\begin{equation}\label{app:var:local}
\begin{aligned}
\mathrm{Var}(\widehat{o_t}) 
&\leq 
\mathbb{E}\left[ \tr(O_A \widehat{\rho_{\bar{A}}})^2 \right] - \left[\tr(O_A \rho^t)\right]^2
\\
& = 
\mathbb{V}^{\mc{E}^{(A)}}(O_{A},\tr_{\bar{A}}\rho) + \left[\tr(O_A \tr_{\bar{A}}\rho)\right]^2 - \left[\tr(O_A \rho^t)\right]^2
\\
& \stackrel{(a)}{=} \mathbb{V}^{\mc{E}^{(A)}}(O_{A,0},\tr_{\bar{A}}\rho)
+ \left[\tr(O_A \rho)\right]^2 - \left[\tr(O_A \rho^t)\right]^2
\\
& \stackrel{(b)}{\leq }
\left\|O_{A,0}\right\|_{\mathrm{sh},\mc{E}^{(A)}}^2 + \|O\|_{\infty}^2 - \left[\tr(O_A \rho^t)\right]^2,
\end{aligned}
\end{equation}
which confirms \cref{eq:varmLocal}. Here,  $\mathbb{V}^{\mc{E}^{(A)}}(O_{A},\tr_{\bar{A}}\rho)$ is the variance of $\tr(O_A \widehat{\rho_{\bar{A}}})$ associated with the original shadow protocol on a single-copy of $\tr_{\bar{A}}\rho$ \cite{huang2020predicting},  
$(a)$ holds because shifting the operator $O_A$ to its traceless part $O_{A,0}$ does not change the variance of $\tr(O_A \widehat{\rho_{\bar{A}}})$, $(b)$ holds because 
$\mathbb{V}^{\mc{E}^{(A)}}(O_{A,0},\tr_{\bar{A}}\rho)\leq \left\|O_{A,0}\right\|_{\mathrm{sh},\mc{E}^{(A)}}^2$ \cite{huang2020predicting} and $\left|\tr(O_A \rho)\right|\leq \|O\|_{\infty}$.

\section{New (multi-shot like) estimator \texorpdfstring{$\widehat{\rho^t}_*$}{} and its performance}\label{app:Multishot}

In the following \cref{th:NewEst}, we show a new estimator $\widehat{\rho^t}_*$ of $\rho^t$ with its performance, which is constructed from 
Eq.~(3) of 
Theorem 1 in the main text by averaging the estimator $\mb{b}$ from $\mb{x}=\{\mb{x}_1,\mb{x}_2, \cdots ,\mb{x}_t\}$.

\begin{prop}\label{th:NewEst}
For a single implementation of the circuit in Fig.~1 of the main text (i.e., single-shot, $M=1$), the unbiased estimator of $\rho^t$ shows 
\begin{align}
\widehat{\rho^t}_*
&:= \mathrm{Re}(f(\mb{x}))\cdot
\sum_{\mb{b}} \Pr(\mb{b}|\mb{x}) \cdot
\mathcal{M}^{-1}\left( V^{\dag} \ketbra{\mb{b}}{\mb{b}} V \right), \label{eq:NewEst}
\end{align}
which satisfies $\mathbb{E}_{\{V,\mb{x}\}}\widehat{\rho^t}_* = \rho^t$. 
Here, the inverse channel $\mathcal{M}^{-1}$ depends on the random unitary ensemble $\mathcal{E}$; $\mb{x}$ is the actual measurement outcome from the quantum circuit in 
Fig.~1 of the main text; the function $f$ and the mapping probability $\Pr(\mb{b}|\mb{x})$ are given by 
Observation 1 of the main text. 
Thanks to 
Lemma 1 in END MATTER, the estimator in \cref{eq:NewEst} can be further written as 
\begin{align}
\widehat{\rho^t}_*
= \mathrm{Re}(f(\mb{x}))\cdot
\frac{1}{t}\sum_{i=1}^t 
\mathcal{M}^{-1}\left( V^{\dag} \ketbra{\mb{x}_i}{\mb{x}_i} V \right),   
\end{align}
where $\mb{x}_i$ is the $i$-th element of the string $\mb{x}=\{\mb{x}_1,\mb{x}_2,\cdots,\mb{x}_t\}$. 

The unbiased estimator of some observable $O$ on $\rho^t$, $\tr(O \rho^t)$, is $\widehat{o_t}^*=\tr(O\widehat{\rho^t}_*)$, and its variance shows
\begin{align}\label{eq:VarNewEst}
\mathrm{Var}\left(\widehat{o_t}^*\right) 
\leq \mathbb{V}^{\mc{E}}_t(O, \rho)
+\left[\tr(O\rho)\right]^2,
\end{align}
with
\begin{align}\label{app:MSvar}
\mathbb{V}^{\mc{E}}_t(O, \rho)
&:= \frac{t-1}{t} \mathbb{V}^{\mc{E}}_*(O, \rho)
+\frac{1}{t} \, \mathbb{V}^{\mc{E}}(O, \rho),
\nonumber\\[1.0ex]
\mathbb{V}^{\mc{E}}_*(O, \rho)
&:=
\underset{V \sim \mc{E}}{\mathbb{E}} \left(\sum_{\mb{b}} \bra{\mb{b}}V\mathcal{M}^{-1}(O)V^{\dag} \ket{\mb{b}} \bra{\mb{b}} V(\rho)\ket{\mb{b}} \right)^2
-\left[\tr(O\rho)\right]^2,
\nonumber\\[0.5ex]
\mathbb{V}^{\mc{E}}(O, \rho)&
:=\underset{V \sim \mc{E}}{\mathbb{E}}
\left(\sum_{\mb{b}}\bra{\mb{b}} V(\rho)\ket{\mb{b}}\bra{\mb{b}}V\mathcal{M}^{-1}(O)V^{\dag} \ket{\mb{b}}^2 \right)
-\left[\tr(O\rho)\right]^2. 
\end{align}
Here, the value of $\mathbb{V}^{\mc{E}}_t(O, \rho)$ is not greater than $\mathbb{V}^{\mc{E}}(O,\rho)$, and $\mathbb{V}^{\mc{E}}(O,\rho)$ is the variance of $\tr(O\hat\rho)$ associated with the original shadow protocol on a single-copy of $\rho$ \cite{huang2020predicting}. 
\end{prop}

We remark that $\mathbb{V}^{\mc{E}}_t(O, \rho)$ is a quantity associated with \emph{multi-shot estimation} 
(see e.g., Lemma 1 in \cite{helsen2023thrifty}), which is a variant of the original shadow protocol \cite{huang2020predicting}. Such multi-shot quantity emerges here in our protocol, as we operate same random unitary on $t$ replicas of $\rho$ in AFRS framework. The value of $\mathbb{V}^{\mc{E}}_t(O, \rho)$ depends on the traceless part $O_0$ of the observable $O$ \cite{zhou2023performance}, and is not greater than $\mathbb{V}^{\mc{E}}(O,\rho)$ \cite{helsen2023thrifty}. In addition, it has been systematically studied for various different $\mc{E}$ \cite{helsen2023thrifty,zhou2023performance}. In particular, when $\mc{E}$ is the local Clifford ensemble $\mc{E}_{\mathrm{LCl}}$ and $O=P$ is a Pauli observable with weight $w$, we have \cite[Theorem 2]{zhou2023performance}
\begin{align}\label{eq:multi_pauli_var}
\mathbb{V}^{\mc{E}}_t(P, \rho)
=\frac{3^w}{t} +\frac{t-1}{t} \, 3^w \left[\tr(P\rho)\right]^2 -\left[\tr(P\rho)\right]^2 ; 
\end{align}
when $\mc{E}$ is a unitary 4-design, we have 
\cite[Theorem 2]{helsen2023thrifty}
\begin{align}
\mathbb{V}^{\mc{E}}_t(O, \rho)
=
\frac{\mathbb{V}^{\mc{E}}(O, \rho)}{t} +\frac{t-1}{t}\, \mathcal{O}\!\left(2^{-n} \Tr(O_0^2)\right), 
\end{align}
where $n$ is the number of qubits; 
when $\mc{E}$ is the ensemble of quantum circuits constructed by interpolating $k$ single-qubit  $T$ gates into $k+1$ random Clifford layers, we have \cite[Theorem 4]{helsen2023thrifty}
\begin{align}
\mathbb{V}^{\mc{E}}_t(O, \rho)
=
\frac{\mathbb{V}^{\mc{E}}(O, \rho)}{t} + \frac{t-1}{t}\, \mathcal{O}\!\left(2^{-n} \Tr(O_0^2)\right) 
+\frac{t-1}{t} \, 30 \Tr(O_0^2)\left(1+\mathcal{O}(2^{-n})\right) \left(\frac{3}{4}+\mathcal{O}\left(2^{-n}\right)\right)^k . 
\end{align}

\begin{proof}[Proof of \cref{th:NewEst}]
The unbiasedness of the estimator $\widehat{\rho^t}_*$ can be shown similarly as that of $\widehat{\rho^t}$ in Supplementary Note~\ref{SM:Th1part1}. The difference is just that the summation on the index $\mb{b}$ already appears in the estimator here.
\begin{align}
\mathbb{E}_{\{V,\mb{x}\}} \widehat{\rho^t}_*
&=\mathbb{E}_{V \sim \mc{E}} \sum_{\mb{x}} \Pr(\mb{x}|V) \left[ \mathrm{Re}(f(\mb{x}))\cdot
\sum_{\mb{b}} \Pr(\mb{b}|\mb{x}) \cdot
\mathcal{M}^{-1}\left( V^{\dag} \ketbra{\mb{b}}{\mb{b}} V \right) \right]
\nonumber\\
&=
\mathbb{E}_{V \sim \mc{E}}  \sum_{\mb{b}}\mathcal{M}^{-1}\left( V^{\dag} \ketbra{\mb{b}}{\mb{b}} V \right) 
{\mathrm{Re}}\left[\sum_{\mb{x}}
\Pr(\mb{x}|V) \Pr(\mb{b}|\mb{x})\,f(\mb{x})
\right]
\nonumber\\
&\stackrel{(a)}{=} 
\mathbb{E}_{V \sim \mc{E}}  \sum_{\mb{b}}
\mathcal{M}^{-1}\left( V^{\dag} \ketbra{\mb{b}}{\mb{b}} V \right) \bra{\bf{b}}V(\rho^t)\ket{\bf{b}}
\nonumber\\
&=
\mathcal{M}^{-1}\left[ \mathbb{E}_{V \sim \mc{E}}  \sum_{\mb{b}}
\bra{\bf{b}}V(\rho^t)\ket{\bf{b}}  V^{\dag} \ketbra{\mb{b}}{\mb{b}} V \right]
\stackrel{(b)}{=} 
\mathcal{M}^{-1}\left[ \mathcal{M} (\rho^t) \right] 
=\rho^t,
\end{align}
where $(a)$ follows from 
Observation 1 in the main text, and $(b)$ follows from  the definition of the channel $\mathcal{M}$. 

In the following, we prove the bound of $\mathrm{Var}\big(\widehat{o_t}^*\big)$ in \cref{eq:VarNewEst}. 
Since the estimator 
\begin{align}
\widehat{o_t}^*=\tr\left(O\widehat{\rho^t}_*\right)
= \mathrm{Re}(f(\mb{x}))
\sum_{\mb{b}} \Pr(\mb{b}|\mb{x}) \bra{\mb{b}} V \mathcal{M}^{-1}(O)  V^{\dag} \ket{\mb{b}} 
\end{align}
is a real number, its variance reads  
\begin{align}\label{eq:VarNew1}
\mathrm{Var}(\widehat{o_t}^*) 
&= \mathbb{E}_{\{V,\mb{x}\}} \left[(\widehat{o_t}^*)^{2} \right] - \left(\mathbb{E}_{\{V,\mb{x}\}} \widehat{o_t}^*\right) ^2 
\nonumber\\[1.0ex]
&=
\mathbb{E}_{V \sim \mc{E}} \sum_{\mb{x}} \Pr(\mb{x}|V) \, (\widehat{o_t}^*)^{2} 
- \left[\tr(O\rho^t)\right]^2
\nonumber\\
&\leq
\mathbb{E}_{V \sim \mc{E}} \sum_{\mb{x}} \Pr(\mb{x}|V) \, \mathrm{Re}(f(\mb{x}))^2 
\left( \sum_{\mb{b}} \Pr(\mb{b}|\mb{x}) \bra{\mb{b}} V \mathcal{M}^{-1}(O)  V^{\dag} \ket{\mb{b}} \right)^2 
\nonumber\\
&\leq 
\mathbb{E}_{V \sim \mc{E}} \sum_{\mb{x}} \bra{\psi_\mb{x}}V(\rho)^{\otimes t}\ket{\psi_\mb{x}}  
\left( \sum_{\mb{b}} \bra{\psi_\mb{x}} Q_\mb{b}\ket{\psi_\mb{x}}  \bra{\mb{b}} V \mathcal{M}^{-1}(O)  V^{\dag} \ket{\mb{b}} \right)^2 
\nonumber\\[1.0ex]
&= 
\mathbb{E}_{V \sim \mc{E}}\sum_{\mb{b}_1,\mb{b}_2}\bra{\mb{b}_1}V\mathcal{M}^{-1}(O)V^{\dag} \ket{\mb{b}_1}\bra{\mb{b}_2} V \mathcal{M}^{-1}(O)V^{\dag}\ket{\mb{b}_2} 
\nonumber\\
&\qquad\qquad\qquad\times 
\underbrace{\left(\sum_{\mb{x}} \bra{\psi_\mb{x}}V(\rho)^{\otimes t}\ket{\psi_\mb{x}} \bra{\psi_\mb{x}} Q_{\mb{b}_1}\ket{\psi_\mb{x}}\bra{\psi_\mb{x}} Q_{\mb{b}_2}\ket{\psi_\mb{x}}\right)}_{(*)}, 
\end{align}
where the first inequality is due to ignoring the second term $\left[\tr(O\rho^t)\right]^2$, and the second inequality holds because 
$\left[\mathrm{Re}(f(\mb{x}))\right]^2=\left[\mathrm{Re}(\bra{\psi_\mb{x}}S_t\ket{\psi_\mb{x}})\right]^2\leq1$.
Since $Q_\mb{b}$ and $\ketbra{\psi_\mb{x}}{\psi_\mb{x}}$ commute, the term $(*)$ can be written as 
\begin{align}\label{eq:VarNew2}
(*)
&=\sum_{\mb{x}} \bra{\psi_\mb{x}}V(\rho)^{\otimes t}Q_{\mb{b}_1} Q_{\mb{b}_2}\ket{\psi_\mb{x}}
\stackrel{(a)}{=} 
\tr\!\left[V(\rho)^{\otimes t}Q_{\mb{b}_1} Q_{\mb{b}_2}\right]
\nonumber\\
&= 
\frac{1}{t^2} \tr\!\left[V(\rho)^{\otimes t} \sum_{i,j=1}^t\left( \ketbra{\mb{b}_1}_i \otimes \id_d^{\otimes (t-1)}\right) \!\left(\ketbra{\mb{b}_2}_j \otimes \id_d^{\otimes (t-1)}\right)\right] 
\nonumber\\
&=  
\frac{1}{t^2} \tr\!\left[V(\rho)^{\otimes t} \left( \sum_{i\ne j} \ketbra{\mb{b}_1}_i\otimes \ketbra{\mb{b}_2}_j\otimes \id_d^{\otimes (t-2)} + \delta_{\mb{b}_1,\mb{b}_2} \sum_{i=1}^t \ketbra{\mb{b}_1}_i\otimes \id_d^{\otimes (t-1)}\right)
\right] 
\nonumber\\[1.0ex]
&= 
\frac{t-1}{t} \bra{\mb{b}_1} V(\rho)\ket{\mb{b}_1} \bra{\mb{b}_2} V(\rho)\ket{\mb{b}_2}
+  \frac{1}{t}\, \delta_{\mb{b}_1,\mb{b}_2} \bra{\mb{b}_1} V(\rho)\ket{\mb{b}_1},
\end{align}
where $(a)$ holds because $\{\ket{\psi_\mb{x}}\}_{\mb{x}}$ forms an orthonormal basis on $\mc{H}_d^{\otimes t}$. 
By inserting \cref{eq:VarNew2} into \cref{eq:VarNew1}, we obtain
\begin{align}\label{eq:VarNew3}
\mathrm{Var}(\widehat{o_t}^*) 
&\leq 
\frac{t-1}{t}\, \mathbb{E}_{V \sim \mc{E}}\sum_{\mb{b}_1,\mb{b}_2} \bra{\mb{b}_1}V\mathcal{M}^{-1}(O)V^{\dag} \ket{\mb{b}_1}\bra{\mb{b}_2} V \mathcal{M}^{-1}(O)V^{\dag}\ket{\mb{b}_2} \bra{\mb{b}_1} V(\rho)\ket{\mb{b}_1} \bra{\mb{b}_2} V(\rho)\ket{\mb{b}_2}
\nonumber\\
&\quad\ + 
\frac{1}{t}\,  
\mathbb{E}_{V \sim \mc{E}}\sum_{\mb{b}_1,\mb{b}_2}\bra{\mb{b}_1}V\mathcal{M}^{-1}(O)V^{\dag} \ket{\mb{b}_1}\bra{\mb{b}_2} V \mathcal{M}^{-1}(O)V^{\dag}\ket{\mb{b}_2} \delta_{\mb{b}_1,\mb{b}_2} \bra{\mb{b}_1} V(\rho)\ket{\mb{b}_1}
\nonumber\\[1.0ex]
&= \frac{t-1}{t} \left[ \mathbb{V}^{\mc{E}}_*(O, \rho) + \tr(O\rho)^2\right]
+\frac{1}{t} \left[\mathbb{V}^{\mc{E}}(O, \rho)+ \tr(O\rho)^2\right]
\nonumber\\[1.0ex]
&= 
\mathbb{V}^{\mc{E}}_t(O, \rho)
+\left[\tr(O\rho)\right]^2,
\end{align}
which confirms \cref{eq:VarNewEst} and completes the proof. 
\end{proof}

We remark such multi-shot-like construction can also be applied to local-AFRS. Based on \cref{th:localAFRS}, one has the new estimator and its variance as follows. 
\begin{prop}\label{th:NewEstLocal}
For a single implementation of the circuit in Fig.~3~(a) of the main text (single-shot, $M=1$), the unbiased estimator of $\rho^t$ reduced on subsystem $A$ shows  
\begin{equation}\label{eq:NewEstlocal}
\begin{aligned}
\widehat{\rho^t_A}_*&=\mathrm{Re}\left(f(\mb{x}^{A},\mb{x}^{\bar{A}}))\right)\sum_{\mb{b}^A} \Pr(\mb{b}^A|\mb{x}^A)\mathcal{M}_A^{-1}\left( 
V_A^{\dag} \ketbra{\mb{b}^A}{\mb{b}^A}V_A\right),
\end{aligned}
\end{equation}
which satisfies $\mathbb{E}_{\{V_A,\mb{x}\}}\widehat{\rho^t_A}_* = \tr_{\bar{A}}(\rho^t)$. 
Here, the inverse channel $\mathcal{M}_A^{-1}$ on $A$ depends on the random unitary ensemble $\mc{E}^{(A)}$; $\mb{x}=\{\mb{x}^{A},\mb{x}^{\bar{A}}\}$ represents the actual measurement outcome from the quantum circuit in 
Fig.~3~(a) of the main text; 
the function $f$ and the mapping probability $\Pr(\mb{b}_A|\mb{x}_A)$ are defined in the main text.

The local observerble $O=O_A\otimes \id_{\bar{A}}$ on $\rho^t$ can be estimated by $\widehat{o_t}^*=\tr\!\left(\widehat{\rho^t_A}_*O_A\right)$,
and its variance shows 
\begin{equation}\label{eq:varmLocalMS}
\begin{aligned}
\mathrm{Var}(\widehat{o_t}^*) \leq \mathbb{V}^{\mc{E}^{(A)}}_t\left(O_A, \tr_{\bar{A}}\rho\right)
+\left[\tr(O_A\rho)\right]^2.
\end{aligned}
\end{equation}
Here, $\mathbb{V}^{\mc{E}^{(A)}}_t\!\left(O_A, \tr_{\bar{A}}\rho\right)$ is defined in a similar way as $\mathbb{V}^{\mc{E}}_t(O, \rho)$ in \cref{app:MSvar}, and is not larger than $\mathbb{V}^{\mc{E}^{(A)}}(O_{A},\tr_{\bar{A}}\rho)$, which denotes the variance of $\tr(O_A \widehat{\rho_{\bar{A}}})$ associated with the original shadow protocol on a single-copy of $\tr_{\bar{A}}\rho$ \cite{huang2020predicting}.

\end{prop}
One can prove \cref{th:NewEstLocal} by following the proof of \cref{th:NewEst} in this section together with the proof of \cref{th:localAFRS} in Supplementary Note~\ref{SM:Th2}, and we do not elaborate on it here.

In \cref{fig:multi}, we numerically simulate the error of estimating $o_2=\tr(X_1\rho^2)$ with the multi-shot-enhanced AFRS and local-AFRS protocols, and compare them with the OS, AFRS and local-AFRS protocols. Here, the processed state $\rho = \ket{\text{GHZ}}\bra{\text{GHZ}}$, and thus we have $\tr(X_1\rho)=0$. The numerical results in \cref{fig:multi} (b) demonstrate that the estimation errors are further reduced by about a factor of $\sqrt{2}$ using new estimators $\widehat{\rho^t}_*$ and $\widehat{\rho^t_A}_*$ defined in \cref{eq:NewEst} and \cref{eq:NewEstlocal}, respectively. 
Such results match well with the theoretical prediction in \cref{eq:multi_pauli_var} when $t=2$.

\begin{figure}[h]
\centering
\includegraphics[width=0.75\linewidth]{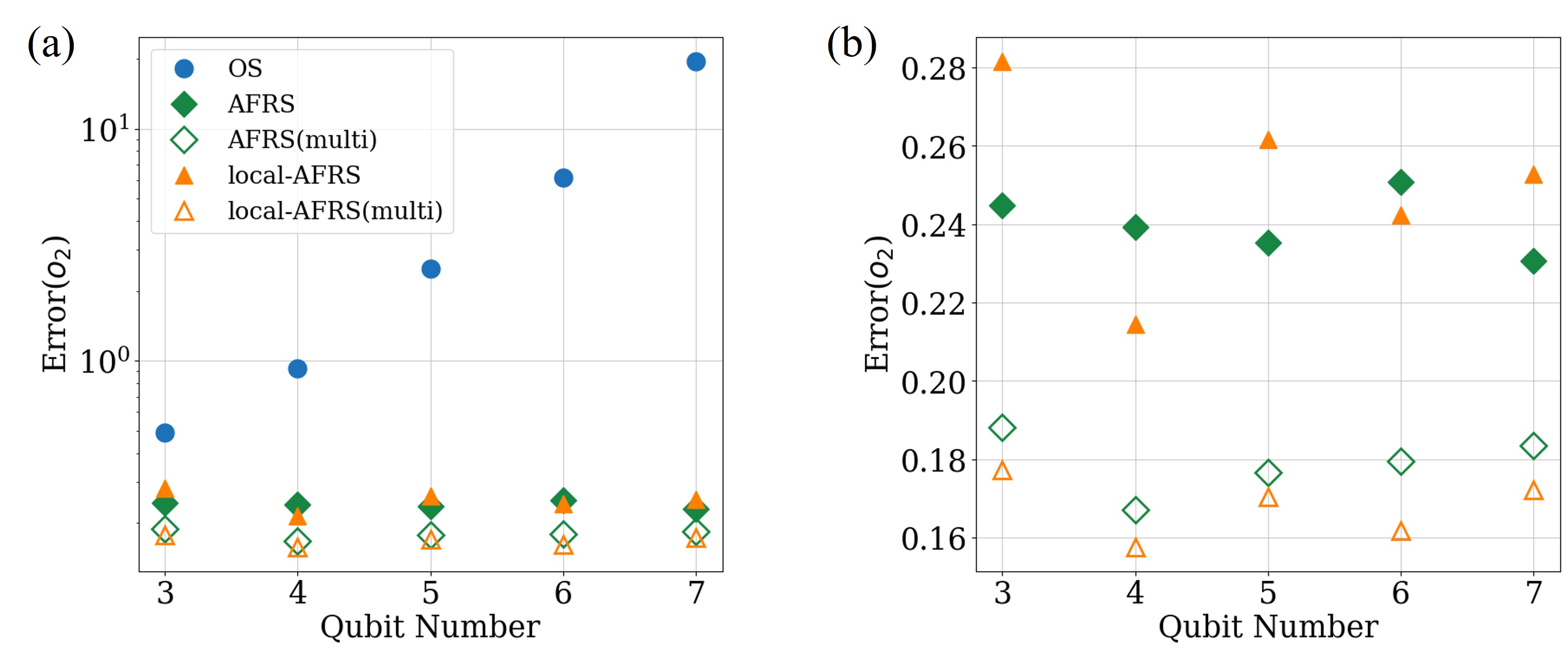}
\caption{\justifying{The estimation error of $o_2 = \tr(O\rho^2)$ with respect to the qubit number for OS, AFRS, local-AFRS and multi-shot-enhanced AFRS and local-AFRS protocols. The quantum state is an $n$-qubit GHZ state, and $V$ is sampled from the local Clifford ensemble $\mc{E}_{\text{LCl}}$ for estimating $\tr(X_1\rho^2)$ with $X_1$ being a local Pauli observable, which is short for $\sigma_{X}^1\otimes \identity_2^{n-1}$. We take the 
same sample number of $\rho$, say $N=100$, for (multi-)AFRS, (multi-)local-AFRS and OS protocols.
In (b), we zoom in the results for the (multi-)AFRS and (multi-)local-AFRS protocols in (a).
}}
\label{fig:multi}
\end{figure}

\section{Proof of 
Corollary 1 in the main text}\label{SM:CorollarylocalAFRS} 

To estimate a set of local observables $\mc{W}=\{O_l\}_{l=1}^L$, we divide the set as $\mc{W}=\cup_{k=1}^K \mc{W}_k$, where each $\mc{W}_k$ satisfies 
\begin{equation}
\begin{aligned}
\supp(O_l) \cap \supp(O_{l'})=\emptyset \quad \text{or} \quad \supp(O_l)\subseteq \supp(O_{l'}) \qquad \forall O_l,O_{l'}\in\mc{W}_k.
\end{aligned}
\end{equation}
For the subset $\mc{W}_k$, our local-AFRS protocol runs as follows. We first apply a random unitary $V\in \mc{E}_{\mathrm{LCl}}$ to each of the $t$ replicas $\rho$, and then use the entangling unitary $U_k=\bigotimes_j \mc{R}_{A_j}$ to rotate the joint state $(V\rho V^{\dag})^{\otimes t}$. Here, $A_j$ are disjoint subsystems that satisfy $\cup_j A_j=[n]$, and any $O_l\in\mc{W}_k$ is supported in one of $A_j$. 
Finally, we measure the joint state on the computational basis, and obtain measurement outcome $\mb{x}=\{\mb{x}^{A_j}\}_j$, where $\mb{x}^{A_j}$ denotes the outcome of the subsystem $A_j$. 

In classical post-processing, the unbiased estimator
of $\rho^t$ reduced on subsystem $A_i$ can be calculated as 
\begin{equation}\label{eq:Wkpost}
\widehat{\rho^t_{\bar{A_i}}}
=\mathrm{Re}\!\left[f\left(\left\{\mb{x}^{A_j}\right\}_j\right)\right]\mathcal{M}_{\mathrm{LCl}}^{-1}\left( 
V_{A_i}^{\dag} \ketbra{\mb{b}^{A_i}}{\mb{b}^{A_i}}V_{A_i}\right),
\qquad
f\left(\left\{\mb{x}^{A_j}\right\}_j\right)= \prod_j \left\<\psi_{\mb{x}^{A_j}}\!\left|S_t^{A_j}\right|\!\psi_{\mb{x}^{A_j}}\right\>,
\end{equation}
where $\mb{b}^{A_i}$ is the mapping result of $\mathcal{G}_{\mb{x}^{A_i}\to \mb{b}^{A_i}}$, and $\ket{\psi_{\mb{x}^{A_j}}}=\mc{R}^{\dag}_{A_j} \ket{\mb{x}^{A_j}}$ is the common eigenstate of $Q^{\mathrm{local}}(\mb{b}^{A_j})$ and $S_t^{A_j}$.  
Note that the form of the estimator here is slightly different from that shown in \cref{th:localAFRS}, since the function $f(\mb{x}^A,\mb{x}^{\bar A})$ is updated to $f\left(\left\{\mb{x}^{A_j}\right\}_j\right)$ considering the
partition of the unitary $U_k=\bigotimes_j \mc{R}_{A_j}$.
If the local observable $O_l\in\mc{W}_k$ is supported in $A_i$, then $\tr(O_l\rho^t)$ can be estimated by $\widehat{o_t^{(l)}}=\tr\!\left[O_l \, \widehat{\rho^t_{\bar{A_i}}}\right]$.
By using a similar argument presented in the proof of \cref{th:localAFRS}, one can show that the variance of this single-shot estimator is bounded as 
\begin{equation}
\begin{aligned}
\mathrm{Var}\left(\widehat{o_t^{(l)}}\right) \leq \left\|(O_l)_{0}\right\|_{\mathrm{sh},\mathcal{E}_\mathrm{LCl}}^2 + \|O_l\|_{\infty} ^2. 
\end{aligned}
\end{equation}

Therefore, one can use $K$ different joint unitaries $U_k$ in the experimental stage to estimate $\tr(O_l\rho^t)$ for all $O_l\in \mc{W}$. 
The performance of this estimation strategy is guaranteed by Corollary 1 in the main text.

\begin{proof}[Proof of Corollary 1]
To achieve the estimation performance stated in 
Corollary 1, our protocol is designed as follows. 
For each observable subset $\mc{W}_k$ with $k\in\{1,2,\dots,K\}$, we apply the unitary $U_k$ to perform 
local-AFRS, and use $M_k=N_k R_k$ repeating measurement time, where 
\begin{align}
N_k = \frac{34}{\epsilon^2}\max_{O_l\in W}\mathrm{Var}\left(\widehat{o^{(l)}_t}\right),\qquad
R_k = 2\log\left(\frac{2\, |\mc{W}_k|}{\delta_k'}\right) ,\qquad
\delta_k':= \frac{\delta |\mc{W}_k|}{L}, 
\end{align}
and $|\mc{W}_k|$ is the size of $\mc{W}_k$.
In post-processing, we use the median
of means technique \cite{huang2020predicting,jerrum1986random} to estimate observables in $\mc{W}_k$. That is, for each $O_l\in \mc{W}_k$, we construct $R_k$ independent sample means of size $N_k$, and then output their median $\hat{o}_l(R_k,N_k)$ as our estimator for $O_l$.  

According to the performance guarantee for the median of means method given in \cref{lem:MofMean} below, the parameters $R_k$ and $N_k$ are chosen such that $\operatorname{Pr}[|\hat{o}_l(R_k,N_k)-\tr(O_l\rho^t)| \geq \epsilon] \leq \delta'_k/|\mc{W}_k|$ for all $l=1,2,\dots,L$.  Consequently, the probability of successfully estimating all observables in $\mc{W}$ within error $\epsilon$ satisfies 
\begin{align}
\Pr({\rm succ})
\geq  
\prod_{k=1}^{K} \left( 1-\frac{\delta'_k}{|\mc{W}_k|} \right)^{|\mc{W}_k|}
\geq
\prod_{k=1}^{K} \left(1-\delta_k'\right) 
\geq 1-\sum_{k=1}^{K}\delta_k'
= 1-\frac{\delta}{L} \sum_{k=1}^{K} |\mc{W}_k|
=1-\delta, 
\end{align}
where the last equality holds because $\sum_{k=1}^{K} |\mc{W}_k|=L$.

Note that the total repeating measurement time used by our protocol reads 
\begin{align}
\sum_{k=1}^K M_k
=
\sum_{k=1}^K
\frac{68}{\epsilon^2}\max_{O_l\in \mc{W}}\mathrm{Var}\left(\widehat{o^{(l)}_t}\right)\log\left(\frac{2\, |\mc{W}_k|}{\delta_k'}\right)
=
\frac{68K}{\epsilon^2}\max_{O_l\in \mc{W}}\mathrm{Var}\left(\widehat{o^{(l)}_t}\right)\log\left(\frac{2L}{\delta}\right), 
\end{align} 
which equals $M$ given in 
Eq.~(7) of the main text. This completes the proof of 
Corollary 1. 
\end{proof}

\begin{lemma}[\cite{huang2020predicting,jerrum1986random}]\label{lem:MofMean}
Suppose $0<\epsilon<1$, $R$ is a positive integer, and $X$ is a random variable with variance $\mathrm{Var}(X)$. Then $R$ independent sample means of size $N=34\mathrm{Var}(X)/\epsilon^2$ is sufficient to construct a median of means estimator $\hat{\mu}(R,N)$ that satisfies $\operatorname{Pr}[|\hat{\mu}(R,N)-\mathbb{E}[X]| \geq \epsilon] \leq 2 \mathrm{e}^{-R/2}$. 
\end{lemma}

\section{Applications}
\red{In this section, we provide practical applications and numerical simulations for estimating nonlinear properties with our AFRS and local-AFRS framework.}

\red{\subsection{Numerical details for typical nonlinear properties $\tr(O\rho^t)$}}
In this part, we clarify the techniques and details of the numerical results given in 
Fig.~2 of the main text. The processed state is a noisy $n$-qubit GHZ state, $\rho = (1-p)\ket{\mathrm{GHZ}}\bra{\mathrm{GHZ}}+p\id_d/d$, with $p=0.3$ for all numerical results.
The estimators of $o_t = \tr(O\rho^t)$ for OS and AFRS protocols are
\begin{align}
    &\widehat{o_t}_\mathrm{OS}=\tr[(O\otimes\id_d^{t-1}) S_t \bigotimes_{i=1}^t \widehat{\rho}_{(i)} ],\label{eq:ot_os}\\
    &\widehat{o_t}_\mathrm{AFRS}=\tr(O\widehat{\rho^t}),\label{eq:ot_af}
\end{align}
by 
Theorem 1 in the main text, respectively. For local observable $O = O_A\otimes \id_{\bar{A}}$, the estimator for local-AFRS is 
\begin{equation}\label{eq:ot_loc_af}
    \widehat{o_t}_\mathrm{local-AFRS} = \tr_{A}\left( O_A \, \widehat{\rho^t_A} \right),
\end{equation}
by \cref{th:localAFRS}.
Note that, in all the comparisons among OS, AFRS and local-AFRS protocols in 
Fig.~2 and 
Fig.~5 for $t=2$, we take the experiment number or the shadow size for OS, say $M_{\text{OS}}=2\times M_{\text{AFRS}}=2\times M_{\text{local-AFRS}}$. In this way, the estimations using different protocols are compared under the same number of the state consumption, i.e. the sampling complexity $N=M_{\text{OS}}$.
Moreover, to obtain the estimation of the statistical error for measuring $o_t$, we repeat the whole procedure for $T=100$ times, and take the estimation error as $\mathrm{Error}(o_t) = \sqrt{\mathrm{Var}(\widehat{o_t})} \sim \left(T^{-1}\sum_{j\in[T]}( \widehat{o_t}_{(j)} - o_t )^2 \right)^{1/2}$, with $\widehat{o_t} = M^{-1}\sum_{i\in [M]} \tr(O\widehat{\rho^t_{(i)}})$.

Figure~2~(a) in the main text shows the error of estimating $\tr(O\rho^2)$, via OS, AFRS and local-AFRS using the estimators in Eqs.~\eqref{eq:ot_os}, \eqref{eq:ot_af} and \eqref{eq:ot_loc_af} respectively with $O =Z_1 Z_2= \sigma_Z^1 \otimes \sigma_Z^2 \otimes \id_2^{\otimes n-2}$ being the local observable on the first two qubits. Specifically, we take $V\in \mc{E}_{\mathrm{LCl}}$ and the sampling number of $\rho$ as $N=100$.
The numerical results show that the estimation error for the OS protocol increases exponentially with the qubit number $n$, that is, $\mathrm{Error}(o_2) \sim \mathcal{O}(d^{1.40 })$. While in AFRS and local-AFRS protocols, the error is independent of the qubit number.
The reason is that for AFRS and local-AFRS protocols, as $O$ and $O_A$ are 2-local Pauli observables, one has 
\begin{equation}
    \|O_{0(A)}\|_{\mathrm{sh},\mathrm{LCl}}^2 \leq 4^{\mathrm{locality}(O_{(A)})} \|O\|_\infty^2,
\end{equation}
which is independent of $n$.
However, for the OS protocol, as $o_t = \tr(O\rho^t) = \tr(O_t \rho^{\otimes t})$ with $O_t := (O\otimes\id_d^{t-1}) S_t$ being a highly global observable (i.e., $\mathrm{locality}(O_t)=\red{nt}$). Therefore, the shadow norm and thus the variance of OS protocol scale exponentially with $n$.

Figure~2~(b) in the main text shows the error of measuring the fidelity $F_2 = \bra{\mathrm{GHZ}}\rho^2\ket{\mathrm{GHZ}}$ with $O = \ket{\mathrm{GHZ}}\bra{\mathrm{GHZ}}$ for $V\in \mathcal{E}_{\mathrm{Cl}}$ and $N=20$ via OS and AFRS protocols. For AFRS with random Clifford measurements, one has 
\begin{equation}
    \|O_0\|_{\mathrm{sh},\mathrm{Cl}}^2  \leq 3\tr(O^2). 
\end{equation}
As for OS, one has
\begin{equation}\label{eq:osvar}
    \|O_t\|^2_{\mathrm{sh, Cl}} \sim d^{t-1}\tr(O^2),
\end{equation}
which implies that the estimation error would scale exponentially with the qubit number for $t=2$. Such exponential scaling advantages for AFRS and local-AFRS with local and global Clifford ensembles are indicated by the numerical results in 
Fig.~2 of the main text.

\red{\subsection{Virtual distillation task}}

\red{
Virtual distillation (VD) aims to purify a mixed state $\rho$ using its $t$-th order `virtual' state $\smash{\rho^{(t)}_\text{VD}=\rho^t/ \tr(\rho^t)}$. In general, the transformation $\rho \rightarrow \rho^{(t)}_\text{VD}$ is not a physical operation; however, one can alternatively construct the expectation value of some observable $O$, say $\langle O \rangle^{(t)}_\text{VD}=\text{tr}(O\rho^{(t)}_\text{VD})=\tr(O\rho^t)/ \tr(\rho^t)$, by quantum measurements and post-processings. VD can be applied to algorithmic cooling of many-body systems \cite{cotler2019cooling}, and to quantum error mitigation protocols by amplifying the dominant state vector \cite{Huggins2021Virtual,Koczor2021Exponential,Brien2022purification}.

We demonstrate the strengths of local-AFRS in such VD task.
Consider the Hamiltonian of the transverse-field Ising model given by $H = -J \sum_j Z_jZ_{j+1}-h \sum_j X_j $. We aim to estimate $\langle O \rangle^{(t)}_\text{VD}$ for all Pauli operators in the summation, that is, $O\in\{Z_jZ_{j+1},X_j\}$. Note that the moment function in the denominator $\tr(\rho^t)$ can be estimated by just taking $O=\id$, and thus the essential part is the numerator $o_t=\tr(O\rho^t)$ for various $O$. By Corollary 1 of the main text, we only need to choose $K=2$ different entangling unitary to accomplish the estimation. That is, 
after the local random Clifford evolution $V=\bigotimes_{i=1}^n V_i$ on each replica, one applies the following two entangling unitaries one at a time (with $n$ being odd for simplicity), $U_{\text{odd}} = \bigotimes_{j=1}^{\lfloor n/2 \rfloor}\mc{R}_{2j-1,2j} \otimes \mc{R}_n$ and $U_{\text{even}} = \mc{R}_1 \otimes\bigotimes_{j=1}^{{\lfloor n/2 \rfloor}}\mc{R}_{2j,2j+1}$,
to measure all $t$ replicas of $\rho$. 
For $\{Z_jZ_{j+1}\}$ with an odd/even $j$, one post-processes the measurement results from  $U_{\text{odd}/\text{even}}$, respectively. For $\{X_i\}$ and $\id$, both measurement patterns could be used, as they cover the supports of the observables. In this way, the sampling time is just $M=\mc{O}(\log(n))$ by Corollary 1. Note that for the case $t=2$, the total circuit depth is at most $3$.}

In 
Fig.~5~(a) of the main text, we manifest the advantages of local-AFRS for estimating various Pauli observables $O=\{Z_j Z_{j+1}, X_j, \identity\}$, compared to the OS protocol. As an advantage, the complexity of the measurement is significantly reduced using local-AFRS protocol with lower circuit depth and the same level of estimation accuracy compared to AFRS. Besides, the estimation errors of both protocols are lower than the one using OS protocol. 
The numerical results show consistency with 
Eqs.~(4) and 
(5) in the main text.

In 
Fig.~5~(b) of the main text, we compare the local-AFRS and OS protocols in the context of VD. Specifically, we use $M^{-1}\sum_{i\in [M]} \tr(O\widehat{\rho^t_{(i)}})$ and $M'^{-1}\sum_{j\in [M']} \tr(\widehat{\rho^t_{(j)}})$ to estimate the numerator and the denominator of $\langle O\rangle_{\mathrm{VD}}^{(t=2)}$, respectively. Note that the estimators of the numerator and the denominator are from independent shadow sets.
The processed state is set as a noisy $5$-qubit GHZ state with $p=0.3$.
One can calculate the theoretical result for $\langle O \rangle_{\mathrm{VD}}^{(t=2)}$ as $\tr(O\rho^2)/\tr(\rho^2)=0.994$.
The estimated values from local-AFRS and OS protocols are compared under the same sample complexity of $\rho$. And the numerical results manifest that the local-AFRS converges faster with the increase of the sampling number of $\rho$, and the result of estimating $\langle O\rangle_{\mathrm{VD}}^{(t=2)}$ is much closer to the theoretical one, compared to OS. The numerical results in 
Fig.~5 of the main text indicates the practicability and superiority of our protocols.


\subsection{\red{The estimator of the moment \texorpdfstring{$P_t$}{Pt} and its variance}}\label{SM:Th2part3}

\red{We now construct the unbiased estimator for the moment $P_t=\tr(\rho^t)$ within the AFRS framework.
In the special case where the subsystem $A$ of the local observable is empty ($A=\emptyset$),
then \cref{eq:unA-Ptr} reduces to} 
\begin{equation}\label{app:PtlocalF}
\begin{aligned}
\widehat{P_t}:=\widehat{\tr(\rho^t)}=\mathrm{Re}(f(\mb{x})).
\end{aligned}
\end{equation}
Note that such an estimator can also be obtained from 
Eq.~(3) of the main text where a global entangling $\mc{R}$ is used. 
However, as shown in the main text, here one can use local circuit $\bigotimes_{i=1}^n\mc{R}_{(i)}$ which entangles the $t$ replicas in a qubit-wise manner. In this way, the eigenbasis of $S_t$ is in the product form $\ket{\phi_{\mb{x}}}=\bigotimes_{i=1}^n\ket{\psi_{\mb{x}^{i}}}$. As a result, the function in \cref{app:PtlocalF} becomes $f(\mb{x})=\prod_{i=1}^n \bra{\psi_{\mb{x}^{i}}}S_t^i\ket{\psi_{\mb{x}^{i}}}=\prod_{i=1}^n f(\mb{x}^{i})$, which proves 
the estimator $\widehat{P_t}$.

For the variance of estimating $P_t$, it follows directly from the general proof of \cref{eq:varmLocal} in Supplementary Note~\ref{SM:Th2part2}. For the case $A=\emptyset$, one has $O=\id$ and $O_{A,0}=0$. Inserting them in the forth line of \cref{app:var:local}, one directly gets $\mathrm{Var}(\widehat{P_t})\leq1-\left[\tr(\rho^t)\right]^2$.

In \cref{fig:purity}, we numerically investigate the error of estimating the purity $P_2=\tr(\rho^2)$ using both the OS and local-AFRS protocols. Our numerical results indicate that the estimation error of the local-AFRS protocol diminishes as the purity of the processed quantum state increases. Furthermore, such error is notably lower compared with that of the OS protocol, particularly in scenarios where the purity is high.
Specifically, when $\rho$ is a pure state, the error of the local-AFRS protocol is zero, which is consistent with our theoretical analysis outlined above. Additionally, as shown in \cref{fig:purity} (b), we find that the estimation error of the local-AFRS protocol is independent of the qubit number $n$, whereas the error of the OS protocol scales exponentially with $n$.

\begin{figure}[htbp]
\centering
\includegraphics[width=0.7\linewidth]{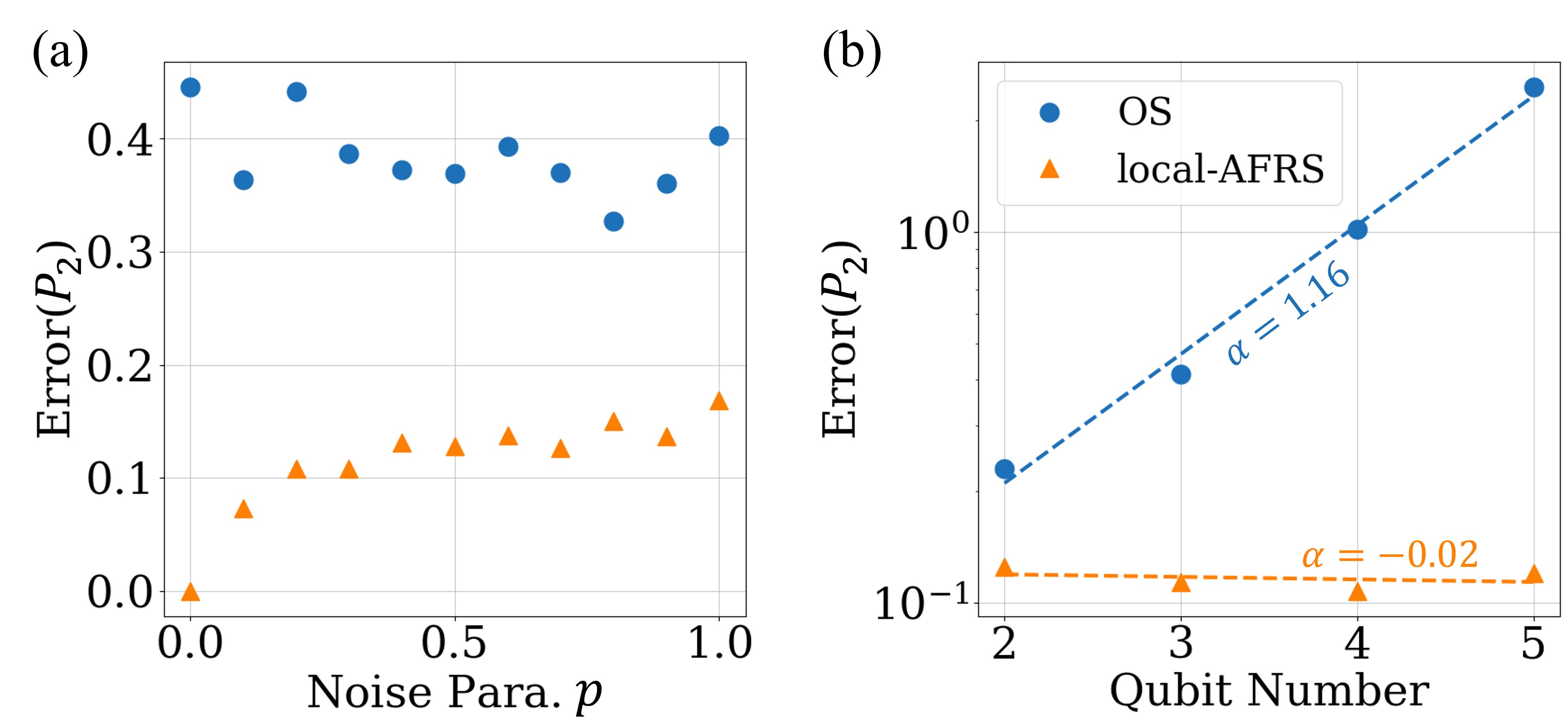}
\caption{\justifying{The estimation error of $P_2 = \tr(\rho^2)$ with respect to the noise parameter $p$ and the qubit number $n$ for OS and local-AFRS protocols. In (a), the quantum state is a noisy $3$-qubit GHZ state $\rho=(1-p)\ket{\mathrm{GHZ}}\bra{\mathrm{GHZ}}+p\identity_d/d$, in which $p$ represents the noise parameter. In (b), the quantum state is a noisy $n$-qubit state with $p=0.3$. The random unitary $V$ is sampled from the local Clifford ensemble $\mc{E}_{\mathrm{LCl}}$. We take the 
sample number $N=100$ for both local-AFRS and OS protocols.
}
}
\label{fig:purity}
\end{figure}


\subsection{\red{Estimation of the quantum Fisher information bound}}
\red{
The \emph{quantum Fisher information} (QFI) is a fundamental quantity that plays a central role in quantum metrology and many other fields, as it provides the ultimate precision limit for estimating an unknown parameter and serves as a witness for multipartite entanglement. The QFI is defined with respect to a given Hermitian operator $O$ and a quantum state $\rho$, and can be expressed as \cite{rath2021Fisher}
\begin{align}
F_Q = 2 \tr\!\left[
\frac{(\rho \otimes \id - \id \otimes \rho)^2}{\rho \otimes \id + \id \otimes \rho}
\mathcal{S}(O \otimes O)
\right]
\end{align}
where $\mathcal{S}$ is the swap operator. 
However, the direct evaluation of $F_Q $ is experimentally challenging. 
To overcome this, recent work \cite{rath2021Fisher} introduced a hierarchy of lower bounds to the QFI using polynomials of the density matrix: 
\begin{align}
F_m = 2 \tr\!\left[\,
\sum_{\ell=0}^m (\rho \otimes \id - \id \otimes \rho)^2 
(\id \otimes \id - \rho \otimes \id - \id \otimes \rho)^\ell 
\mathcal{S}(O \otimes O)
\right], \qquad m=0,1,2,\dots . 
\end{align}
It was proved that $F_m$ converges exponentially with $m$ to $F_Q$, and thus provides an accurate approximation of the QFI \cite{rath2021Fisher}. 
Based on OS, Ref.~\cite{rath2021Fisher} proposed a protocol for estimating $F_m$, but its sample complexity scales exponentially with the qubit number $n$, making the protocol resource-consuming. 
Whether $F_m$ can be estimated with polynomially many samples remains an open problem.

In this subsection, we show that $F_m$ can be estimated efficiently via a variant of our local-AFRS protocol, provided that $m$ is a constant and the observable $O$ admits a decomposition into $\mathcal{O}(\text{poly}(n))$  local observables, i.e., $O = \sum_i O_i$, where each $O_i$ has support on a constant number of qubits. 
Actually, most thermodynamic quantities in quantum many-body physics can be written as observables of this form, such as geometrically-local Hamiltonians. 
In this case, $F_m$ can be expressed as a summation of $\mathcal{O}(\text{poly}(n))$ terms of the form $\tr(O_i\rho^{t_1} O_j\rho^{t_2})$, with $t_1,t_2\geq 0$ being constant integers. 
Therefore, the task of estimating $F_m$ reduces to estimating each such term to sufficient accuracy.  
Denote by $A$ the smallest subsystem on which both $O_i$ and $O_j$ act nontrivially (i.e.,  $A=\text{supp}(O_i) \cap \text{supp}(O_j)$), 
and by $\bar{A}$ its complementary subsystem;  
then we have $O_i=O_{i,A}\otimes \id_{\bar{A}}$ and $O_j=O_{j,A}\otimes \id_{\bar{A}}$. 

\begin{figure}[htbp]
\centering
\includegraphics[width=0.5\linewidth]{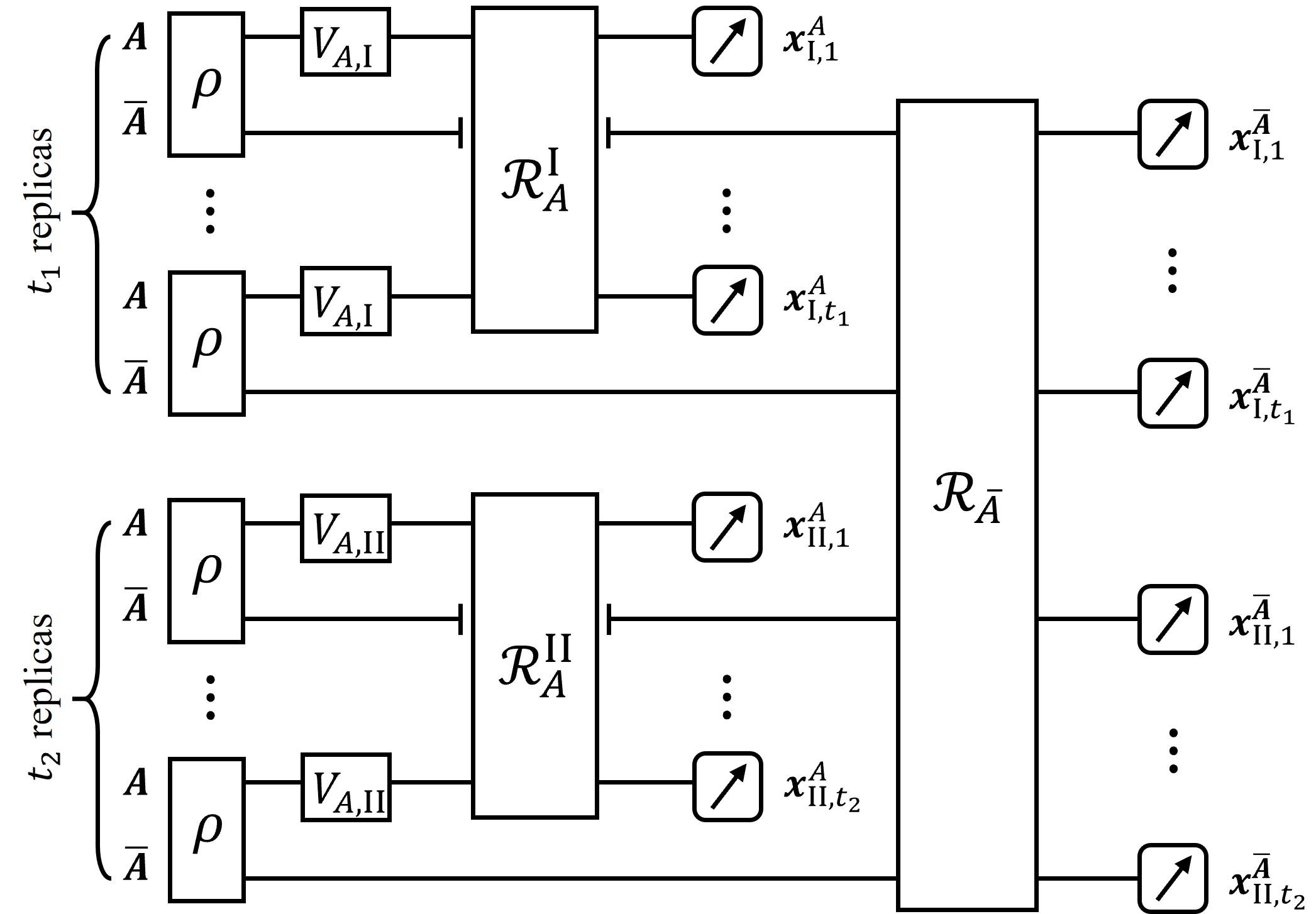}
\caption{\red{Variant local-AFRS circuit for estimating $\tr(O_i\rho^{t_1} O_j\rho^{t_2})$. Here, the joint operation $\mc{R}_{\bar{A}}$ can be replaced by a tensor product of qubit-wise operations $\mc{R}_{(i)}$, i.e., $\bigotimes_{i\in\bar{A}}\mc{R}_{(i)}$.
}}
\label{fig:QFIcircuit1}
\end{figure}

Our quantum circuit for estimating $\tr(O_i\rho^{t_1} O_j\rho^{t_2})$ is illustrated in Fig.~\ref{fig:QFIcircuit1}, which modifies the local-AFRS circuit introduced in the main text.  
First, we prepare $t_1$ identical copies of the quantum state $\rho$ and apply the same random unitary $V_{A,\rm{I}}\in\mathcal{E}^{(A)}$ to the subsystems $A$ of all $t_1$ copies. These $t_1$ subsystems are then entangled through the joint operation $\mathcal{R}_A^{\rm{I}}$, followed by a computational basis measurement yielding the outcome bit-string $\mb{x}^{A}_{\rm{I}}=\{\mb{x}^{A}_{\rm{I},1}\dots\mb{x}^{A}_{{\rm I},t_1}\}$.  
Next, we prepare an additional $t_2$ copies of $\rho$ and repeat the same procedure, now with the random unitary $V_{A,\rm{II}}$, joint operation $\mathcal{R}_A^{\rm{II}}$, and outcome bit-string $\mb{x}^{A}_{\rm{II}}$.  Finally, we apply the joint operation $\mathcal{R}_{\bar{A}}$ to the subsystems $\bar{A}$ of all $t_1+t_2$ copies, followed by a computational basis measurement that produces the bit-string $\mb{x}^{\bar{A}}$.  
Here, the joint operations should satisfy 
\begin{equation}\label{eq:RARAbar}
\mathcal{R}_A^{\rm{I}}\ket{\psi_{\mb{x}^{A}_{\rm{I}}}}=\ket{\mb{x}^{A}_{\rm{I}}},\qquad   \mathcal{R}_A^{\rm{II}}\ket{\psi_{\mb{x}^{A}_{\rm{II}}}}=\ket{\mb{x}^{A}_{\rm{II}}}, 
\qquad \mathcal{R}_{\bar{A}}\ket{\psi_{\mb{x}^{\bar{A}}}}=\ket{\mb{x}^{\bar{A}}},
\end{equation}
where $\{|\psi_{\mb{x}^{A}_{\rm{I}}}\>\}$ \big(respectively, $\{|\psi_{\mb{x}^{A}_{\rm{II}}}\>\}$\big) is a common eigenbasis of 
\begin{equation}
Q_{\mb{b}^{A}}^{\rm{I}}=t_1^{-1}\sum_{i} \ketbra{\mb{b}^{A}}_i\otimes \id^{\otimes (t_1-1) }_{A} 
\quad 
\left(\text{respectively,}\  Q_{\mb{b}^{A}}^{\rm{II}}=t_2^{-1}\sum_{i} \ketbra{\mb{b}^{A}}_i\otimes \id^{\otimes (t_2-1) }_{A} \right)
\end{equation}
and the cyclic permutation operator $S_{t_1}^A$ (respectively, $S_{t_2}^A$) restricted on subsystems $A$; and
$\{|\psi_{\mb{x}^{\bar{A}}}\>\}$ is an eigenbasis of $S_{t_1+t_2}^{\bar{A}}$ restricted on $\bar{A}$. 
Due to this property, analogously to the local-AFRS protocol, $\mathcal{R}_{\bar{A}}$ in Fig.~\ref{fig:QFIcircuit1} can be replaced by a tensor product of qubit-wise operations $\mc{R}_{(i)}$.  
When $|A|$ is a constant, the operations $\mathcal{R}_A^{\rm{I}}$, $\mathcal{R}_A^{\rm{II}}$, and $\mc{R}_{(i)}$ can all be realized with constant-depth quantum circuits. Consequently, \textbf{the entire circuit in Fig.~\ref{fig:QFIcircuit1} has only constant depth}, making it convenient to implement on real-world quantum devices.


In the classical post-processing stage, we map $\mb{x}^{A}_{\rm{I}}$ to a string $\mb{b}^{A}_{\rm{I}}$ with probability $\Pr(\mb{b}^{A}_{\rm{I}}\,\big|\,\mb{x}^{A}_{\rm{I}})= \Big\<\psi_{\mb{x}^{A}_{\rm{I}}}\Big|\,Q_{\mb{b}_{\rm{I}}^{A}}^{\rm{I}}\Big|\psi_{\mb{x}^{A}_{\rm{I}}}\Big\>$, and map $\mb{x}^{A}_{\rm{II}}$ to a string $\mb{b}^{A}_{\rm{II}}$ with probability $\Pr(\mb{b}^{A}_{\rm{II}}\,\big|\,\mb{x}^{A}_{\rm{II}})= \Big\<\psi_{\mb{x}^{A}_{\rm{II}}}\Big|\,Q_{\mb{b}_{\rm{II}}^{A}}^{\rm{II}}\,\Big|\psi_{\mb{x}^{A}_{\rm{II}}}\Big\>$. 
Then we use $\mb{x}^{A}_{\rm{I}}$, $\mb{x}^{A}_{\rm{II}}$, $\mb{x}^{\bar{A}}$, $\mb{b}^{A}_{\rm{I}}$, $\mb{b}^{A}_{\rm{II}}$, $V_{A,\rm{I}}$, and $V_{A,\rm{II}}$ to construct a snapshot as follows, 
\begin{align}
\hat{Y}:= f\!\left(\mb{x}_{\,\rm{I}}^{A}\right)  f\!\left(\mb{x}_{\,\rm{II}}^{A}\right) f\!\left(\mb{x}^{\bar{A}}\right) \cdot \mathcal{M}_A^{-1}\!\left( V_{A,\rm{I}}^{\dag} \ketbra{\mb{b}_{\rm{I}}^A}{\mb{b}_{\rm{I}}^A} V_{A,\rm{I}} \right) \otimes \mathcal{M}_A^{-1}\!\left( V_{A,\rm{II}}^{\dag} \ketbra{\mb{b}_{\rm{II}}^A}{\mb{b}_{\rm{II}}^A} V_{A,\rm{II}} \right) ,
\end{align}
where 
\begin{align}\label{eq:f(xA)f(xAbar)}
f\!\left(\mb{x}_{\,\rm{I}}^{A}\right)=\bra{\psi_{\mb{x}^{A}_{\rm{I}}}} S_{t_1}^A\ket{\psi_{\mb{x}^{A}_{\rm{I}}}},
\qquad 
f\!\left(\mb{x}_{\,\rm{II}}^{A}\right)=\bra{\psi_{\mb{x}^{A}_{\rm{II}}}} S_{t_2}^A\ket{\psi_{\mb{x}^{A}_{\rm{II}}}},
\qquad 
f\!\left(\mb{x}^{\bar{A}}\right)=\left\<\psi_{\mb{x}^{\bar{A}}}\left|S_{t_1+t_2}^{\bar{A}}\right|\psi_{\mb{x}^{\bar{A}}}\right\>. 
\end{align}

\begin{lemma}\label{lem:ExpectY}
The expectation of $\hat{Y}$ (over the choices of random unitaries $V_{A,\rm{I}}$ and $V_{A,\rm{II}}$, measurement outcomes $\mb{x}^{A}$ and $\mb{x}^{\bar{A}}$, and random mapping results $\mb{b}^{A}_{\rm{I}}$ and $\mb{b}^{A}_{\rm{II}}$) is 
\begin{align}\label{eq:defUpsilon}
\mathbb{E}\big[\hat{Y}\big]=\Upsilon
:=\tr_{\bar{A}}\!\left[ \left(\rho^{t_1}\otimes \rho^{t_2}\right)\left(\id_A \otimes \id_A \otimes \mathcal{S}_{\bar{A}}\right)
\right]. 
\end{align}
\end{lemma}

We define the estimator $\widehat{o^{i,j}_{t_1,t_2}}:= \tr[\hat{Y}(O_{j,A}\otimes O_{i,A}) \mathcal{S}_A]$, whose expectation value, by Lemma~\ref{lem:ExpectY}, satisfies 
\begin{align}
\mathbb{E}\!\left[\widehat{o^{i,j}_{t_1,t_2}}\right]
=\tr\!\left\{ \tr_{\bar{A}}\!\left[ \left(\rho^{t_1}\otimes \rho^{t_2}\right)\left(\id_A \otimes \id_A \otimes \mathcal{S}_{\bar{A}}\right)
\right] (O_{j,A}\otimes O_{i,A}) \mathcal{S}_A \right\}
=\tr(O_i\rho^{t_1} O_j\rho^{t_2}),  
\end{align}
confirming it as an unbiased estimator for $\tr(O_i\rho^{t_1} O_j\rho^{t_2})$. The variance of $\widehat{o^{i,j}_{t_1,t_2}}$ remains small because $\hat{Y}$ acts only on a constant number of qubits, enabling our circuit in Fig.~\ref{fig:QFIcircuit1} to estimate $\tr(O_i\rho^{t_1} O_j\rho^{t_2})$ efficiently with sufficient accuracy. According to our previous discussions, these results imply that our variant local-AFRS protocol can estimate the lower bound $F_m$ using 
very few state samples, demonstrating an exponential improvement in sample efficiency over existing approaches.

\begin{proof}[Proof of Lemma~\ref{lem:ExpectY}]
We have 
\begin{align}\label{eq:EhatYproof1}
\mathbb{E}\big[\hat{Y}\big]
&=\underset{V_{A,\rm{I}},V_{A,\rm{II}}}{\mathbb{E}}
\sum_{\mb{b}^{A}_{\rm{I}},\mb{b}^{A}_{\rm{II}}}\sum_{\mb{x}^{A},\mb{x}^{\bar{A}}} 
\Pr(\left.\mb{x}^{A},\mb{x}^{\bar{A}} \,\right|V_{A,\rm{I}},V_{A,\rm{II}})
\Pr(\mb{b}^{A}_{\rm{I}} \,\big|\,\mb{x}^{A}_{\rm{I}})
\Pr(\mb{b}^{A}_{\rm{II}}\,\big|\,\mb{x}^{A}_{\rm{II}})
\cdot \hat{Y}
\nonumber\\
&=\underset{V_{A,\rm{I}},V_{A,\rm{II}}}{\mathbb{E}}
\sum_{\mb{b}^{A}_{\rm{I}},\mb{b}^{A}_{\rm{II}}}
\mathcal{M}_A^{-1}\!\left( V_{A,\rm{I}}^{\dag} \ketbra{\mb{b}_{\rm{I}}^A}{\mb{b}_{\rm{I}}^A} V_{A,\rm{I}} \right) \otimes \mathcal{M}_A^{-1}\!\left( V_{A,\rm{II}}^{\dag} \ketbra{\mb{b}_{\rm{II}}^A}{\mb{b}_{\rm{II}}^A} V_{A,\rm{II}} \right) 
\nonumber\\
&\qquad\qquad\quad 
\times \underbrace{\left[\sum_{\mb{x}^{A},\mb{x}^{\bar{A}}} 
\Pr(\left.\mb{x}^{A},\mb{x}^{\bar{A}} \,\right|V_{A,\rm{I}},V_{A,\rm{II}})
\Pr(\mb{b}^{A}_{\rm{I}} \,\big|\,\mb{x}^{A}_{\rm{I}})
\Pr(\mb{b}^{A}_{\rm{II}}\,\big|\,\mb{x}^{A}_{\rm{II}}) f\!\left(\mb{x}_{\,\rm{I}}^{A}\right)  f\!\left(\mb{x}_{\,\rm{II}}^{A}\right) f\!\left(\mb{x}^{\bar{A}}\right) \right]}_{(*)}. 
\end{align}
Note that 
\begin{align}\label{eq:EhatYproof2}
\Pr(\!\left.\mb{x}^{A},\mb{x}^{\bar{A}} \right|\!V_{A,\rm{I}},V_{A,\rm{II}})
&=
\bra{\mb{x}_{\rm{I}}^{A},\mb{x}_{\rm{II}}^{A},\mb{x}^{\bar{A}}} 
\!\left(\mathcal{R}_A^{\rm{I}}\otimes \mathcal{R}_A^{\rm{II}}\otimes \mathcal{R}_{\bar{A}}\right)\!
\left[ V_{A,\rm{I}}(\rho)^{\otimes t_1} \otimes V_{A,\rm{II}}(\rho)^{\otimes t_2}\right]\!
\left(\mathcal{R}_A^{\rm{I}}\otimes \mathcal{R}_A^{\rm{II}}\otimes \mathcal{R}_{\bar{A}}\right)^{\dag}\!
\ket{\mb{x}_{\rm{I}}^{A},\mb{x}_{\rm{II}}^{A},\mb{x}^{\bar{A}}} 
\nonumber\\
&=
\bra{\psi_{\mb{x}^{A}_{\rm{I}}},\psi_{\mb{x}^{A}_{\rm{II}}},\psi_{\mb{x}^{\bar{A}}}} 
V_{A,\rm{I}}(\rho)^{\otimes t_1} \otimes V_{A,\rm{II}}(\rho)^{\otimes t_2}
\ket{\psi_{\mb{x}^{A}_{\rm{I}}},\psi_{\mb{x}^{A}_{\rm{II}}},\psi_{\mb{x}^{\bar{A}}}}, 
\end{align}
where the second equality follows from \cref{eq:RARAbar}, and we write $|\psi_{\mb{x}^{A}_{\rm{I}}},\psi_{\mb{x}^{A}_{\rm{II}}},\psi_{\mb{x}^{\bar{A}}}\>=|\psi_{\mb{x}^{A}_{\rm{I}}}\>|\psi_{\mb{x}^{A}_{\rm{II}}}\>|\psi_{\mb{x}^{\bar{A}}}\>$ for simplicity. 
Thus, the term $(\cdot)$ in \cref{eq:EhatYproof1} can be expressed as 
\begin{align}\label{eq:EhatYproof3}
(*)&\stackrel{(a)}{=} \sum_{\mb{x}^{A}_{\rm{I}},\mb{x}^{A}_{\rm{II}},\mb{x}^{\bar{A}}} 
\bra{\psi_{\mb{x}^{A}_{\rm{I}}},\psi_{\mb{x}^{A}_{\rm{II}}},\psi_{\mb{x}^{\bar{A}}}} 
V_{A,\rm{I}}(\rho)^{\otimes t_1} \otimes V_{A,\rm{II}}(\rho)^{\otimes t_2}
\ket{\psi_{\mb{x}^{A}_{\rm{I}}},\psi_{\mb{x}^{A}_{\rm{II}}},\psi_{\mb{x}^{\bar{A}}}} 
\nonumber\\
&\qquad\ \   \times \bra{\psi_{\mb{x}^{A}_{\rm{I}}},\psi_{\mb{x}^{A}_{\rm{II}}},\psi_{\mb{x}^{\bar{A}}}} 
S_{t_1}^A\otimes S_{t_2}^A\otimes S_{t_1+t_2}^{\bar{A}}
\ket{\psi_{\mb{x}^{A}_{\rm{I}}},\psi_{\mb{x}^{A}_{\rm{II}}},\psi_{\mb{x}^{\bar{A}}}}
\bra{\psi_{\mb{x}^{A}_{\rm{I}}},\psi_{\mb{x}^{A}_{\rm{II}}},\psi_{\mb{x}^{\bar{A}}}} 
Q_{\mb{b}_{\rm{I}}^{A}}^{\rm{I}} \otimes Q_{\mb{b}_{\rm{II}}^{A}}^{\rm{II}} \otimes \id_{\bar{A}}
\ket{\psi_{\mb{x}^{A}_{\rm{I}}},\psi_{\mb{x}^{A}_{\rm{II}}},\psi_{\mb{x}^{\bar{A}}}}
\nonumber\\[1.2ex]
&\stackrel{(b)}{=}\tr\!\left\{ \left[V_{A,\rm{I}}(\rho)^{\otimes t_1} \otimes V_{A,\rm{II}}(\rho)^{\otimes t_2}
\right] \left( S_{t_1}^A Q_{\mb{b}_{\rm{I}}^{A}}^{\rm{I}} \otimes S_{t_2}^A Q_{\mb{b}_{\rm{II}}^{A}}^{\rm{II}} \otimes S_{t_1+t_2}^{\bar{A}} \right)\right\}
\nonumber\\[0.8ex]
&\stackrel{(c)}{=} \tr\!\left\{ \left(\rho^{t_1}\otimes \rho^{t_2}\right)\left[\left( V_{A,\rm{I}}^{\dag} \ketbra{\mb{b}_{\rm{I}}^A}{\mb{b}_{\rm{I}}^A} V_{A,\rm{I}} \right)\otimes \left( V_{A,\rm{II}}^{\dag} \ketbra{\mb{b}_{\rm{II}}^A}{\mb{b}_{\rm{II}}^A} V_{A,\rm{II}} \right)\otimes \mathcal{S}_{\bar{A}}\right]
\right\}
\nonumber\\[0.8ex]
&\stackrel{(d)}{=} \tr\!\left[\Upsilon \left( V_{A,\rm{I}}^{\dag} \ketbra{\mb{b}_{\rm{I}}^A}{\mb{b}_{\rm{I}}^A} V_{A,\rm{I}} \otimes  V_{A,\rm{II}}^{\dag} \ketbra{\mb{b}_{\rm{II}}^A}{\mb{b}_{\rm{II}}^A} V_{A,\rm{II}} \right)
\right]. 
\end{align}
Here, equality $(a)$ follows from Eqs.~\eqref{eq:f(xA)f(xAbar)} and \eqref{eq:EhatYproof2}, along with the relations $\Pr(\mb{b}^{A}_{\rm{I}}\,\big|\,\mb{x}^{A}_{\rm{I}})= \big\<\psi_{\mb{x}^{A}_{\rm{I}}}\big|\,Q_{\mb{b}_{\rm{I}}^{A}}^{\rm{I}}\big|\psi_{\mb{x}^{A}_{\rm{I}}}\big\>$ and   $\Pr(\mb{b}^{A}_{\rm{II}}\,\big|\,\mb{x}^{A}_{\rm{II}})= \big\<\psi_{\mb{x}^{A}_{\rm{II}}}\big|\,Q_{\mb{b}_{\rm{II}}^{A}}^{\rm{II}}\,\big|\psi_{\mb{x}^{A}_{\rm{II}}}\big\>$; 
equality $(b)$ holds because $\{|\psi_{\mb{x}^{\bar{A}}}\>\}$ is an eigenbasis of $S_{t_1+t_2}^{\bar{A}}$, while $\{|\psi_{\mb{x}^{A}_{\rm{I}}}\>\}$ \big($\{|\psi_{\mb{x}^{A}_{\rm{II}}}\>\}$\big) is a common eigenbasis of 
$Q_{\mb{b}^{A}}^{\rm{I}}$ ($Q_{\mb{b}^{A}}^{\rm{II}}$) and 
$S_{t_1}^A$ ($S_{t_2}^A$); 
equality $(c)$ can be derived analogously to \cref{eq:obs2_proof} using tensor network diagrams; and
equality $(d)$ follows from the definition of $\Upsilon$ in \cref{eq:defUpsilon}.

By combining Eqs.~\eqref{eq:EhatYproof1} and \eqref{eq:EhatYproof3}, we have 
\begin{align}
\mathbb{E}\big[\hat{Y}\big]
&=\underset{V_{A,\rm{I}},V_{A,\rm{II}}}{\mathbb{E}}
\sum_{\mb{b}^{A}_{\rm{I}},\mb{b}^{A}_{\rm{II}}}
\mathcal{M}_A^{-1}\!\left( V_{A,\rm{I}}^{\dag} \ketbra{\mb{b}_{\rm{I}}^A}{\mb{b}_{\rm{I}}^A} V_{A,\rm{I}} \right) \otimes \mathcal{M}_A^{-1}\!\left( V_{A,\rm{II}}^{\dag} \ketbra{\mb{b}_{\rm{II}}^A}{\mb{b}_{\rm{II}}^A} V_{A,\rm{II}} \right) 
\nonumber\\[-1ex]
&\qquad\qquad\qquad\qquad 
\times \tr\!\left[\Upsilon \left( V_{A,\rm{I}}^{\dag} \ketbra{\mb{b}_{\rm{I}}^A}{\mb{b}_{\rm{I}}^A} V_{A,\rm{I}} \otimes  V_{A,\rm{II}}^{\dag} \ketbra{\mb{b}_{\rm{II}}^A}{\mb{b}_{\rm{II}}^A} V_{A,\rm{II}} \right)
\right]
\nonumber\\[1ex]
&= \mathcal{M}_A^{-1}\otimes \mathcal{M}_A^{-1} \Bigg\{
\underset{V_{A,\rm{I}},V_{A,\rm{II}}}{\mathbb{E}}
\sum_{\mb{b}^{A}_{\rm{I}},\mb{b}^{A}_{\rm{II}}}
\tr\!\left[\Upsilon \left( V_{A,\rm{I}}^{\dag} \ketbra{\mb{b}_{\rm{I}}^A}{\mb{b}_{\rm{I}}^A} V_{A,\rm{I}} \otimes  V_{A,\rm{II}}^{\dag} \ketbra{\mb{b}_{\rm{II}}^A}{\mb{b}_{\rm{II}}^A} V_{A,\rm{II}} \right)
\right]
\nonumber\\[-1ex]
&\qquad\qquad\qquad\qquad\qquad\qquad\qquad
\times\left( V_{A,\rm{I}}^{\dag} \ketbra{\mb{b}_{\rm{I}}^A}{\mb{b}_{\rm{I}}^A} V_{A,\rm{I}} \otimes V_{A,\rm{II}}^{\dag} \ketbra{\mb{b}_{\rm{II}}^A}{\mb{b}_{\rm{II}}^A} V_{A,\rm{II}} \right) 
\Bigg\}
\nonumber\\
&= \left(\mathcal{M}_A^{-1}\otimes \mathcal{M}_A^{-1} \right) \Big[\!\left(\mathcal{M}_A \otimes \mathcal{M}_A\right) (\Upsilon) \Big]
=\Upsilon,
\end{align}
where we have used the definition of the measurement channel [see also \cref{eq:channel_def}]: 
\begin{align}
\mathcal{M}_A(\cdot):=
\underset{V_A}{\mathbb{E}}  \sum_{\mb{b}^A}
\tr[(\cdot)\left( V_A^{\dag} \ketbra{\mb{b}^A}{\mb{b}^A} V_A\right)]\left( V_A^{\dag} \ketbra{\mb{b}^A}{\mb{b}^A} V_A\right). 
\end{align}
This completes the proof of Lemma~\ref{lem:ExpectY}. 
\end{proof}

To further demonstrate the power of our variant local-AFRS protocol for QFI, we performed numerical simulations to estimate $\mathcal{F}_0$ for depolarized GHZ states $\rho = (1-p)\ket{\mathrm{GHZ}_n}\bra{\mathrm{GHZ}_n} + p\,\mathbb{I}_d/d$, with depolarization parameter $p=0.3$. The observable $O$ is set to be the Pauli-$Z$ operator acting on the first qubit, i.e., $O=Z_1$. Although these choices are simple, it suffices to reveal the core scaling advantage of the local-AFRS protocol over the OS protocol.  
Fig.~\ref{fig:fisher_renyi}~(a) shows the empirical estimation error of $\mathcal{F}_0$ versus qubit number $n$. For identical sampling budgets, our (variant) local-AFRS achieves a constant estimation error with respect to the qubit number, 
demonstrating an exponential scaling advantage over OS as $n$ increases.
This observation agrees with our variance analysis: while the OS protocol's variance scales exponentially with $n$, local-AFRS's variance depends only on the locality of the local observables of interest.  
These results demonstrate that (variant) local-AFRS offers a practical and scalable approach to estimating metrological quantities.
}

\begin{figure}[htbp]
\centering
\includegraphics[width=0.75\linewidth]{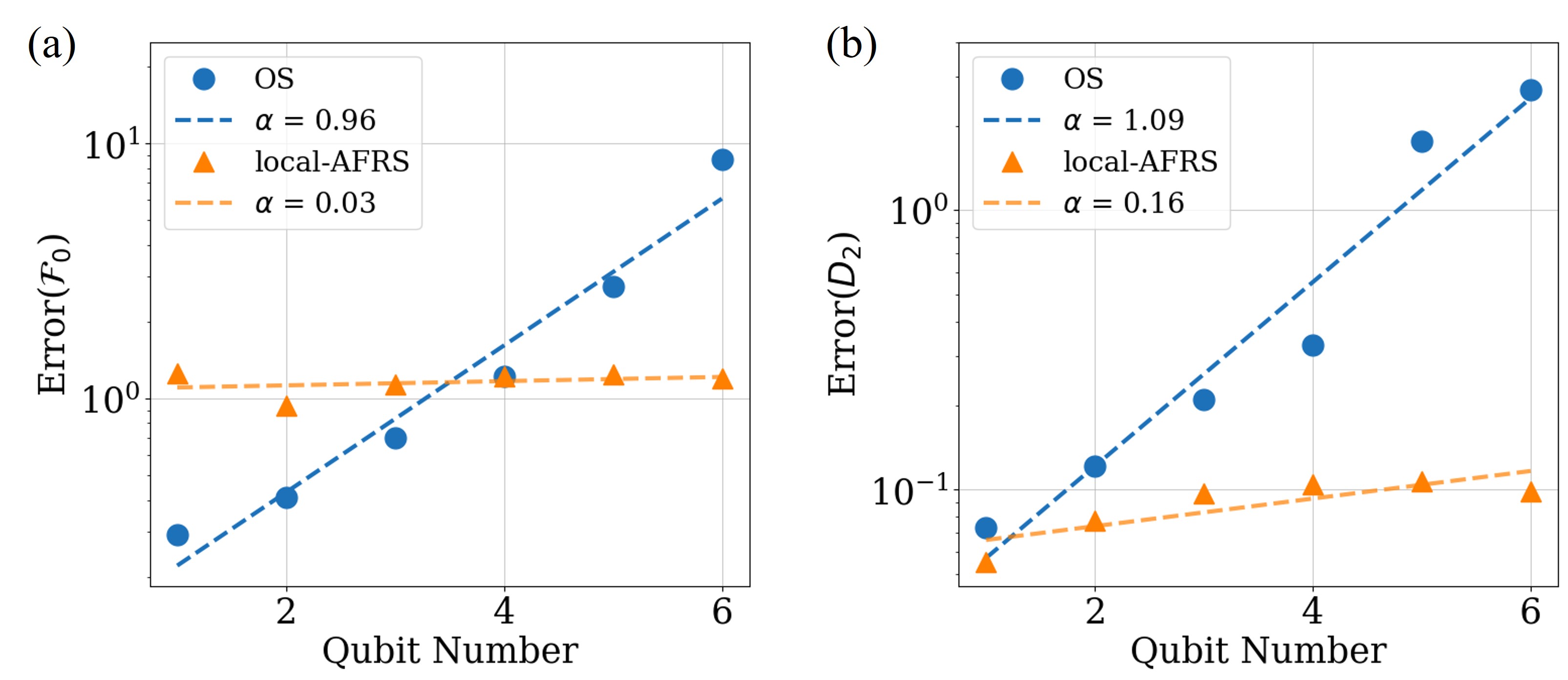}
\caption{\justifying{\red{The estimation error with respect to the system size $n$ for (a) the Fisher information bound $\mathcal{F}_0$ and (b) the R\'{e}nyi relative entropy $D_2(\rho\|\identity_d/d)$. In both cases, the local-AFRS protocol achieves an exponential advantage over OS.}
}}
\label{fig:fisher_renyi}

\end{figure}

\subsection{\red{Estimation of entropic quantities}}
\red{In this subsection, we illustrate a concrete, practically relevant application of AFRS by estimating the Petz-R\'{e}nyi relative entropy, $D_\alpha(\rho\|\sigma) := \frac{1}{\alpha-1}\log_2\tr(\rho^\alpha\sigma^{1-\alpha})$, at order $\alpha=2$ we have 
\begin{equation}
    D_2(\rho\|\sigma)=\log_2\tr(\rho^2\sigma^{-1}).
\end{equation}
We treat the important special case $\sigma=\identity_d/d$ (the maximally mixed state), for which the relative entropy reduces to a simple function of the purity estimation and is frequently used for quantum mixedness testing \cite{chen2022tight} and quantum error correction \cite{cree2022approximate}. In this case,
\begin{equation}
    D_2\left(\rho\left\| \frac{\identity_d}{d}\right.\right)=\log_2\left(d\tr(\rho^2)\right)=n+\log_2P_2,\  P_2:=\tr(\rho^2),
\end{equation}
with $d=2^n$. Hence estimating $D_2(\rho\|\identity_d/d)$ is equivalent to estimating the purity $P_2$, which is directly accessible with the $t=2$ AFRS estimator introduced in the main text.
Below, we describe the estimator, error propagation, 
and the numerical simulation details.

Let $\widehat{P_2}$ denote the unbiased AFRS estimator of $P_2$, which is constructed as in \cref{app:PtlocalF}. 
We estimate the elative entropy by $\widehat{D_2}=n+\log_2\widehat{P_2}$. For small additive estimation error $\delta P_2=\widehat{P_2}-P_2$, a first-order expansion gives
\begin{equation}
    \widehat{D_2}-D_2 \approx \frac{1}{\ln 2}\frac{\delta P_2}{P_2}. 
\end{equation}
Thus, to estimate $D_2$ within error $\epsilon$, it suffices to estimate $P_2$ within error $\epsilon'\approx (\ln 2)\epsilon P_2$. According to our previous discussions, this accuracy  can be efficiently achieved provided that $P_2$ is not too small. 


In \cref{fig:fisher_renyi}~(b), we numerically simulate the error of estimating $D_2$ using the OS and local-AFRS protocol respectively. 
We instantiate the above discussion on the depolarized $n$-qubit GHZ states, $\rho=(1-p)\ket{\mathrm{GHZ}_n}\bra{\mathrm{GHZ}_n}+p\identity_d/d$, with $p=0.3$ for the numerical simulation demonstrated here.
The analytical purity of $\rho$ together with the R\'{e}nyi entropy $D_2(\rho\|\identity_d/d)$ can be therefore calculated and used as ground truth to compute the estimation error. 
For a fair resource comparison, both methods are given the same sampling number of the quantum state $N=1000$. The numerical data confirm the theoretical expectations. The empirical mean-squared error of estimating $D_2$ using OS grows rapidly with the qubit number $n$, which is consistent with \cref{eq:osvar}. In contrast, the estimation error of AFRS remains nearly flat (or grows only mildly) in the tested range; consequently, AFRS attains a given precision with far fewer measurements than OS, and the relative advantage widens rapidly with $n$. Note that increasing the shot budget reduces the error for both methods, with AFRS reaching target errors using substantially fewer shots.
}

\red{Beyond these illustrative example, we note that efficient estimation of nonlinear quantities is also crucial in a wide range of tasks, including protocols for quantum virtual cooling and the detection of mixed-state quantum phases \cite{du2025optimalrandomizedmeasurementsfamily}. Taken together, the analyses and simulations presented here demonstrate that (local-)AFRS provides a powerful and versatile framework for estimation nonlinear properties for quantum states, with broad potential impact for both near-term experiments and future applications in quantum information science.}

\end{appendix}
\end{document}